\documentclass[12pt,english]{article}
\usepackage{lmodern}
\usepackage{lmodern}
\usepackage[T1]{fontenc}
\usepackage[latin9]{inputenc}
\usepackage{geometry}
\usepackage{color}
\usepackage{babel}
\usepackage{amsmath}
\usepackage{amsthm}
\usepackage{amssymb}
\usepackage{float,hhline,microtype,afterpage,soul,enumitem,pgfplots,pdfpages,verbatim,multirow,booktabs,subcaption}
\usepackage{graphicx}
\usepackage{setspace}
\usepackage[authoryear]{natbib}
\usepackage[unicode=true,
 bookmarks=true,bookmarksnumbered=false,bookmarksopen=false,
 breaklinks=false,pdfborder={0 0 1},backref=false,colorlinks=true]
 {hyperref}
\hypersetup{pdftitle={Evolutionarily Stable (Mis)specifications: Theory and Applications},
 pdfborderstyle=,pdfpagelayout=OneColumn,pdfnewwindow=true,pdfstartview=XYZ,plainpages=false,urlcolor=[rgb]{0.0430,0,0.5},linkcolor=[rgb]{0.0430,0,0.5},citecolor=[rgb]{0.0430,0,0.5},hypertexnames=false}

\makeatletter

\providecommand{\tabularnewline}{\\}

\theoremstyle{definition}
\newtheorem{defn}{\protect\definitionname}
\theoremstyle{definition}
 \newtheorem{example}{\protect\examplename}
\theoremstyle{plain}
\newtheorem{thm}{\protect\theoremname}
\theoremstyle{plain}
\newtheorem{prop}{\protect\propositionname}
\theoremstyle{definition}
\newtheorem{condition}{\protect\conditionname}
\theoremstyle{remark}
\newtheorem{notation}{\protect\notationname}
\theoremstyle{plain}
\newtheorem{lem}{\protect\lemmaname}
\theoremstyle{plain}
\newtheorem{assumption}{\protect\assumptionname}
\theoremstyle{remark}
\newtheorem{claim}{\protect\claimname}

\usepackage{color}
\usepackage{babel}

\providecommand{\assumptionname}{Assumption}
\providecommand{\claimname}{Claim}
\providecommand{\conditionname}{Condition}

\providecommand{\definitionname}{Definition}
\providecommand{\examplename}{Example}
\providecommand{\lemmaname}{Lemma}
\providecommand{\notationname}{Notation}
\providecommand{\propositionname}{Proposition}
\providecommand{\theoremname}{Theorem}

\makeatother

\providecommand{\assumptionname}{Assumption}
\providecommand{\claimname}{Claim}
\providecommand{\conditionname}{Condition}
\providecommand{\definitionname}{Definition}
\providecommand{\examplename}{Example}
\providecommand{\lemmaname}{Lemma}
\providecommand{\notationname}{Notation}
\providecommand{\propositionname}{Proposition}
\providecommand{\theoremname}{Theorem}

\begin{document}
\title{\vspace{-60bp}
 Evolutionarily Stable (Mis)specifications: \\
 Theory and Applications\thanks{We thank Cuimin Ba, Thomas Chaney, Sylvain Chassang, In-Koo Cho, Krishna Dasaratha, Andrew Ellis, Ignacio
Esponda, Mira Frick, Drew Fudenberg, Alice Gindin, Ryota Iijima, Yuhta
Ishii, Philippe Jehiel, Pablo Kurlat, Jonny Newton, Filippo Massari,
Andy Postlewaite, Philipp Sadowski, Alvaro Sandroni, Grant Schoenebeck,
Joshua Schwartzstein, Philipp Strack, Carl Veller, and various conference and seminar
participants for helpful comments.  Byunghoon Kim provided excellent research assistance. Kevin He thanks the California
Institute of Technology for hospitality when some of the work on this
paper was completed, and the University Research Foundation Grant at the University of Pennsylvania for financial support. Jonathan Libgober thanks Yale University and the Cowles foundation for their hospitality.}}
\author{Kevin He\thanks{University of Pennsylvania. Email: \texttt{\protect\protect\protect\href{mailto:hesichao\%5C\%5C\%5C\%40gmail.com}{hesichao@gmail.com}}}
\and Jonathan Libgober\thanks{University of Southern California. Email: \texttt{\protect\protect\protect\href{mailto:libgober\%5C\%5C\%5C\%40usc.edu}{libgober@usc.edu}}}}
\date{{\normalsize{}{}{}}%
\begin{tabular}{rl}
First version: & December 20, 2020\tabularnewline
This version: & February 10, 2023\tabularnewline
\end{tabular}}

\maketitle
\vspace*{-20bp}

\begin{abstract}
{\normalsize{}{}{}\thispagestyle{empty} \setcounter{page}{0}}{\normalsize\par}

\noindent Toward explaining the persistence of biased inferences, we propose a framework to evaluate competing (mis)specifications in strategic settings. Agents with heterogeneous (mis)specifications coexist and draw Bayesian inferences about their environment through repeated play. The relative stability of (mis)specifications depends on their adherents' equilibrium payoffs. A key mechanism is the \emph{learning channel}: the endogeneity of perceived best replies due to inference. We characterize when a rational society is only vulnerable to invasion by some misspecification through the learning channel. The learning channel leads to new stability phenomena, and can confer an evolutionary advantage to otherwise detrimental biases in economically relevant applications.   \medskip{}

\end{abstract}

\newpage 

\section{Introduction}

In many economic settings, people draw \emph{misspecified inferences}
about the world: that is, they learn from data but start with a prior
belief that dogmatically precludes the true data-generating process.
For instance, past work has documented a number of prevalent
statistical biases. Reasoning about economic fundamentals
under the spell  of these biases  constitutes  misspecified
learning. Following \cite{esponda2016berk}, a growing literature
has focused on the implications of Bayesian learning under different
misspecifications. Most of the  work in this area look at exogenously given misspecifications.

Compared with many other kinds of errors and mistakes, a distinctive
component of misspecified \emph{learning} is that  biased agents use data
to form beliefs about the world. Why and when might misspecified learning
persist, and does the ability to draw inferences enhance
the viability of such mistakes? We study this question in strategic settings, associating the viability of a particular (mis)specification with the objective payoffs of  individuals who adopt it. Our approach to endogenizing misspecified inference contrasts with those involving subjective expectations of payoffs \citep*{olea2019competing,Levyetal2020,gagnon2018channeled}
or goodness-of-fit tests \citep{cho2015learning,cho2017gresham,ba2020,schwartzstein2020using,Lanzani2022}. It also contrasts with work that has used objective payoffs to endogenize misspecified inference in \emph{single-agent} settings \citep*{FL_mutation,FII_welfare_based} or restricted attention to financial markets \citep{sandroni2000markets,massari2020under}.

Our main message is that the \emph{learning channel} --- i.e., the
ability for agents to learn and draw (possibly wrong) inferences from
data --- strictly expands the scope for misspecifications to invade
rational societies in strategic settings. Central to our approach is articulating ways of distinguishing dogmatic beliefs (which are exogenous and do not depend on observed data) from flexible beliefs (which are  endogenously determined in equilibrium). 
We highlight that a rational
society can be immune to \emph{any} invaders who do not
learn from data, yet be vulnerable to some invaders who undertake \emph{strategically
beneficial misinferences}. Also, the mapping between different matching assortativities and the selected biases may be reversed for agents who do not learn from data relative to those who do. 

We find general conditions under which the learning channel enables more invasions, and we also study applications where the  invading misspecification is encoded in economically meaningful and natural biases.
Along similar lines, we examine some tests that guarantee a rational
society will repel invasion by a given group of invaders, provided these opponents do not undertake inference. We find that passing these tests no longer guarantees immunity to invasion when the opponents are misspecified agents who mislearn. Misspecified learners are \emph{polymorphic}: they can appear weak in one environment and
become stronger in another environment in a way impossible for
biased agents with a dogmatic belief and a fixed best response. Due to the learning channel, the misspecified invaders' equilibrium beliefs and equilibrium best-response function depend on details of the environment (e.g., matching assortativity and   population composition).

In applications, we  show how the persistence
of particular biases depends jointly on the social interaction structure, the possibility of learning, and the stage game's payoff structure. All three factors influence the selection of biases, so studying only one factor in isolation may provide an incomplete understanding.

\subsection{Inference and Selecting Misspecified Beliefs about Correlation} \label{sect:LQNintuitionIntro}

To articulate some intuition for why inference can affect the selection of biases, we informally describe one application of our framework. Consider
a linear-quadratic-normal (LQN) game of incomplete information
as the stage game, interpreted as an incomplete-information version of Cournot duopoly. A population of players (firms) match in pairs
every period to play the stage game. The intercept of the demand curve
is drawn i.i.d. across games, and every pair of matched players receive
correlated information about this intercept in their game.
After  observing this signal, players choose a production quantity. The market price depends on the intercept of demand, the quantity choices of the firms, and a price elasticity parameter (which is fixed across matches). 

We suppose that a small fraction of firms hold a dogmatically wrong belief about the  signal correlation and invade a society which has correct beliefs about all game parameters. An important property of this game is that players gain from strategic commitments, and which commitments are valuable depend on assortativity. If entrants are only paired with each other (perfectly assortative matching), 
then they can improve payoffs by committing to more cooperative strategies. If entrants are
paired with the rational incumbents (uniform matching), then committing to more aggressive strategies
can help them obtain more favorable outcomes compared to when incumbents play
each other.  Our contribution is to show that whether a certain biased belief about signal correlation leads to more cooperative or more aggressive play, and hence
whether it will be selected for a given matching assortativity, depends on whether the learning
channel is present.

When learning is absent, an increase in the subjective perception of correlation makes a player choose \emph{less} aggressive strategies. Intuitively, because production quantities are strategic substitutes, a player who believes signals are excessively correlated will produce relatively less following an optimistic signal about demand, expecting the opponent to produce more. But when inference is present, an exaggerated perception of correlation also leads the player to believe that market price is less elastic relative to the truth. This is because the agent overestimates opponent's production and is thus surprised by how little the price adjusts. Inferring a more inelastic price  makes the player choose \emph{more} aggressive strategies. While these forces move in opposite directions, the second effect dominates. Thus, the presence of inference can reverse the conclusion of which misperception outperforms rationality. 

The presence of the learning channel has an even more stark effect on the selection of
errors when the underlying elasticity parameter can take on multiple possible values. In
such cases, a fixed belief about elasticity can be beneficial for some realizations of the true
elasticity parameter but harmful for others. We use this idea to show that generally,
there is some amount of uncertainty under which no entrant with a fixed misperception
about correlation and elasticity can invade a rational society, but some entrant with a biased
belief making flexible inferences from data can do
strictly better than the incumbents.

\subsection{A Framework of Competing Specifications}
In our general framework, we encode specifications in \emph{models} that
delineate feasible beliefs about the stage game. These models serve as the basic unit of cultural
transmission.  The
model's adherents think that one of the model parameters describes the true stage game. They estimate the best-fitting parameter which determines their subjective preference. Models  rise and fall in prominence
based on the objective welfare of adherents, as higher payoffs
confer greater evolutionary success.

When we allow for inference in the example above, the incumbents and the entrants differ in their
perceptions about the signal correlation structure in the stage game. Every firm learns about an aspect of the environment
(price elasticity) through the lens of its model. Firms that believe
in different correlations interpret the same observation differently
when inferring price elasticity, as they make different estimates
about rival firm's production based on their own demand signal. 

Society consists of the adherents of multiple competing models who match up to play the stage game every period. 
We introduce the concept of a \emph{zeitgeist} to capture the
social interaction structure --- the sizes of the
subpopulations with different models and the matchmaking technology
that pairs up opponents to play the game.  Agents can identify which subpopulation their opponent is from, and (correctly) know that the game they play is orthogonal to the type of opponent.\footnote{If the players think that the stage game can change depending on their opponent, then this would give additional channels for biases to invade a rational society. Our framework focuses on how the learning channel that plays a distinctive role in misspecified learning affects the viability of errors.} Our framework assumes that the agents
might face one of several possible games and
therefore richer models can in principle help as they allow agents to adapt their behavior more. Conditional on the stage game, in equilibrium each agent
forms a Bayesian belief about the  game using data from all
of her interactions, and plays a subjective best response against every type of opponent given this belief. 

We define the \emph{evolutionary stability} of model A against model
B based on whether model A has a weakly higher average equilibrium payoff
than model B when the population share of model A is close to 1, with the average taken over the different stage games. This criterion is familiar from past work that use what is known as the \emph{indirect evolutionary approach}. Under this approach, evolution does not directly act on strategies, but rather acts on some trait that determines best responses. While our stability concepts reduce to standard notions under this approach when inference is absent, our contribution is to apply it to the selection of \emph{models} that contain multiple feasible beliefs about the environment.

Indeed, we show that the ability to draw inferences within a model (as opposed to committing to a fixed belief) may be necessary for misspecifications to defeat rationality. In Section \ref{sec:TheoryOverPref}, we  characterize environments where the correctly specified model is only evolutionarily fragile against invading models that allow for inferences.  Our argument constructs an optimal misspecified model for invading a rational society. This misspecification resembles an ``illusion of control'' bias, where agents think the outcomes they get in a game only depend on their own strategy and not on the opponent's strategy. The model has the property that its adherents end up adopting
the optimal commitment against a correctly specified opponent game-by-game. Misinference thus becomes a channel to tailor commitments to the true game. The correctly specified model  is  evolutionarily fragile   against this misspecified model  with uniform matching, unless the former already gets the Stackelberg payoff in every game.

More generally, one can ask whether misspecified models can exhibit different  stability properties than distorted preferences   in our framework. Our next two results say that misspecified models are more \emph{polymorphic}: they can appear weak against rational incumbents in one environment and yet grow stronger and successfully invade the rational society in another environment, in a way that is impossible for invaders with a fixed subjective preference. The reason is that due to the learning channel, an adherent of a misspecified model may come to hold different beliefs about parameters of the underlying stage game, and thus adopt different best-reply functions, when facing game outcomes generated from different strategy profiles. Thus, changes in the population structure and matching process can influence perceived best replies for adherents of misspsecified models.

Polymorphism enables a new stability phenomenon that we call  \emph{stability reversals}. Two models exhibit stability reversal
if:
\begin{enumerate}
    \item whenever
model  A is dominant,  its adherents strictly outperform model  B's adherents
not only on average, but even conditional on opponent's type; and
    \item whenever model  B is dominant, its adherents strictly outperform
model A's adherents on average
\end{enumerate} In the absence of inference, condition (1) would imply that A outperforms B regardless of the two subpopulations' sizes. But this no longer holds when inference is possible. The reason is that the adherents of model  B might make an evolutionarily advantageous inference only when they are matched up with each other sufficiently often.  Thus, even if condition (1) held, model B might still drive out  model A if model B adherents reach some critical mass.

Polymorphism also manifests in a non-monotonicity of stability with respect to matching assortativity. As discussed in \cite{AlgerWeibull2013}, the assortativity parameter can represent  degree of homophily in the society or frequency of interaction with kin.  Various versions of the idea that high assortativity selects for cooperative agents and low assortativity selects for competitive ones date  back to at least \cite{Hamiltonb,Hamiltona}. But this simple dichotomous perspective becomes complicated with misspecifications. Because the adherents of a misspecified model can draw different misinferences about a fixed game's parameters when facing data  generated by different opponent actions,  one model may be favored over another only at \emph{intermediate} levels of assortativities, but not favored at  either very low or very high levels. Thus, a particular bias might only survive in moderately homophilous societies --- a novel empirical implication of misspecified inference.

\section{Environment and Stability Concept \label{sec:Environment-and-Stability}}

We start with our formal stability
concept, defining \emph{equilibrium
zeitgeist} to determine the evolutionary fitness of specifications that coexist in a society. We consider a separate notion, \emph{equilibrium zeitgeist
with strategic uncertainty}, in Section \ref{sec:ABEE}, when we allow agents to draw inferences about others' strategies in addition to learning about the fundamentals. Online Appendix
\ref{sec:Learning-Foundation} provides a combined learning foundation for both equilibrium concepts, but in the main text we primarily focus on the steady-state characterization.

\subsection{\label{subsec:Objective-Primitives}Objective Primitives}

A population of agents repeatedly match to play a stage game, which is a symmetric two-player game with a common, metrizable strategy
space $\mathbb{A}$. There is a set of possible
states of nature $G\in\mathcal{G}$, called \emph{situations}. The
strategy choices $a_{i},a_{-i}\in\mathbb{A}$ of $i$ and $-i$, together
with the situation, stochastically generate consequences $y_{i},y_{-i}\in\mathbb{Y}$
from a metrizable space $\mathbb{Y}$. Each $i$'s consequence $y_{i}$
determines her utility, according to a common utility function $\pi:\mathbb{Y}\to\mathbb{R}$.
The objective distribution over consequences is $F^{\bullet}(a_{i},a_{-i},G)\in\Delta(\mathbb{Y}),$ with an associated density or probability mass function associated denoted by $f^{\bullet}(a_{i},a_{-i},G),$ where $f^{\bullet}(a_{i},a_{-i},G)(y)\in\mathbb{R}_{+}$
for each $y\in\mathbb{Y}$. We suppress $G$ from $f^{\bullet}$ and
$F^{\bullet}$ when $|\mathcal{G}|=1$.

This setup captures mixed strategies (if $\mathbb{A}$
is the set of mixtures over some pure actions), incomplete-information
games (if $S$ is a space of private signals, $A$ a space of actions,
and $\mathbb{A}=A^{S}$ is the set of signal-contingent actions), and even asymmetric games. For the latter, we consider the ``symmetrized'' version where each player is placed into each role with equal probability (see Section \ref{sec:ABEE} for one application where agents play an asymmetric game).

\subsection{Models and Parameters}

Throughout this paper, we will take the strategy space $\mathbb{A},$
the set of consequences $\mathbb{Y},$ and the utility function over
consequences $\pi$ to be common knowledge among the agents. But,
agents are unsure about how play in the stage game translates into
consequences: that is, they have \emph{fundamental uncertainty} about
the function $(a_{i},a_{-i})\mapsto F^{\bullet}(a_{i},a_{-i},G).$

We focus on the case where society consists of two observably distinguishable groups
of agents, A and B, who may behave differently in the stage game due
to different beliefs about how $y$ is generated. The two groups of agents entertain different \emph{models} of the
world that help resolve  their fundamental uncertainty. A model $\Theta$ is a collection of data-generating processes  $F:\mathbb{A}^{2}\to\Delta(\mathbb{Y})$ about
how strategy profiles translate into consequences for the agent, with different processes corresponding to different \emph{parameters} of the model. Each $F$ has associated with it a density or probability mass function $f(a_{i},a_{-i}):\mathbb{Y}\to\mathbb{R}_{+}$
for every $(a_{i},a_{-i})\in\mathbb{A}^{2}$. We thus view each model as a subset
of $(\Delta(\mathbb{Y}))^{\mathbb{A}^{2}}$ and we   assume it is
metrizable. 

Each agent enters society with a persistent model, which
depends entirely on whether she is from group A or group B. We refer to the agents who are endowed with a given model the \emph{adherents} of that model. Each agent dogmatically believes that in every situation $G \in \mathcal{G}$, one of the parameters of her model accurately represents the stage game. We call
 $\Theta=\{F^{\bullet}(\cdot,\cdot,G):G\in\mathcal{G}\}$  the  \emph{minimal correctly specified }model. A model may exclude the true $F^{\bullet}(\cdot,\cdot,G)$
that produces consequences, at least in some situation $G$. In this case, the model is
\emph{misspecified}.  

\subsection{\label{subsec:ZeitDef}Zeitgeists}

To study competition between two models, we must describe the social
composition and interaction structure in the society where learning
takes place. We have in mind a setting where each agent plays the stage game with a random opponent
in every period and uses her personal experience in these matches
to calibrate the most accurate parameter within her model.  A zeitgeist describes the corresponding landscape.
\begin{defn}
\label{def:Zeitdef}Fix models $\Theta_{A}$ and $\Theta_{B}$.
A \emph{zeitgeist $\mathfrak{Z}=(\mu_{A}(G),\mu_{B}(G),p,\lambda,a(G))_{G\in\mathcal{G}}$
}consists of: (1) for each situation $G,$ a belief over parameters for
each model, $\mu_{A}(G)\in\Delta(\Theta_{A})$ and $\mu_{B}(G)\in\Delta(\Theta_{B})$;
(2) relative sizes of the two groups in the society, $p=(p_{A},p_{B})$
with $p_{A},p_{B}\ge0,$ $p_{A}+p_{B}=1$; (3) a matching assortativity
parameter $\lambda\in[0,1]$; (4) for each situation $G,$ each group's
strategy when matched against each other group, $a=(a_{AA}(G),a_{AB}(G),a_{BA}(G),a_{BB}(G))$
where $a_{g,g^{'}}(G)\in\mathbb{A}$ is the strategy that an adherent
of $\Theta_{g}$ plays against an adherent of $\Theta_{g^{'}}$ in
situation $G$.
\end{defn}
A zeitgeist outlines the beliefs and interactions among agents with
heterogeneous models living in the same society. Part (1) captures the beliefs of each group. Parts (2) and (3) determine
social composition and social interaction---the relative prominence
of each model and the probability of interacting with one's own group
versus with the overall population. In each period, $\lambda$ is the probability an agent's opponent is from her own group,
and $1-\lambda$ is the probability the opponent is drawn uniformly from the population.
Therefore, an agent from group $g$ has probability
$\lambda+(1-\lambda)p_{g}$ of being matched with an opponent from
her own group, and a complementary chance of being matched with an
opponent from the other group. Part (4) describes behavior in the
society. Note that a zeitgeist describes each group's situation-contingent
belief and behavior, since agents may infer different parameters and thus
adopt different subjective best replies in different situations.

\subsection{\label{subsec:EZDef}Equilibrium Zeitgeists}
 A model's fitness corresponds to the equilibrium payoffs of its adherents. An equilibrium zeitgeist (EZ) imposes optimality conditions on 
inference  and behavior in a zeitgeist. Optimality of behavior requires each player to best respond
given her beliefs, and optimality of inference requires that the support
of each player's belief only contains the ``best-fitting'' parameter
from her model in the sense of minimizing Kullback-Leibler (KL) divergence.

We now formalize this criterion. For two distributions over consequences, $\Phi,\Psi\in\Delta(\mathbb{Y})$
with density or probability mass functions $\psi,\phi$,
define the KL divergence from $\Psi$ to $\Phi$ as $D_{KL}(\Phi\parallel\Psi):=\int\phi(y)\ln\left(\frac{\phi(y)}{\psi(y)}\right)dy$.
Recall that every data-generating process $F$, like the true fundamental $F^{\bullet}(\cdot, \cdot, G)$, outputs a distribution over consequences for
every profile of own play and opponent's play, $(a_{i},a_{-i})\in\mathbb{A}^{2}$.
For data-generating process $F,$ let $K(F;a_{i},a_{-i},G):=D_{KL}(F^{\bullet}(a_{i},a_{-i},G)\parallel F(a_{i},a_{-i}))$
be the KL divergence from the expected distribution
$F(a_{i},a_{-i})$ to the objective distribution $F^{\bullet}(a_{i},a_{-i},G)$
under the play $(a_{i},a_{-i})$ and situation $G$. For a distribution $\mu$ over parameters, let $U_{i}(a_{i},a_{-i};\mu)$
represent $i$'s subjective expected utility under the belief that
the true parameter is drawn according to $\mu.$ That is, $U_{i}(a_{i},a_{-i};\mu):=\mathbb{E}_{F\sim\mu}(\mathbb{E}_{y\sim F(a_{i},a_{-i})}[\pi(y)])$.
\begin{defn}
\label{def:EZ}A zeitgeist $\mathfrak{Z}=(\mu_{A}(G),\mu_{B}(G),p,\lambda,a(G))_{G\in\mathcal{G}}$
is an\emph{ equilibrium zeitgeist (EZ)} if, for every $G\in\mathcal{G}$
and $g,g^{'}\in\{A,B\},$ $a_{g,g^{'}}(G)\in\underset{a_{i}\in\mathbb{A}}{\arg\max}\ U_{i}(a_{i},a_{g^{'},g}(G);\mu_{g}(G))$
and, for every $g\in\{A,B\},$ belief $\mu_{g}(G)$ is supported on
\begin{align*}
\underset{F\in\Theta_{g}}{\arg\min}\left\{ (\lambda+(1-\lambda)p_{g})\cdot K(F;a_{g,g}(G),a_{g,g}(G),G)+(1-\lambda)(1-p_{g})\cdot K(F;a_{g,-g}(G),a_{-g,g}(G),G)\right\} 
\end{align*}
where $-g$ means the group other than $g$.
\end{defn}

Plainly, this definition requires agents from group $g$ to choose a subjective best response
against their opponents, given the belief $\mu_{g}$
about the fundamental uncertainty. No matter which group the agent is matched against, these choices are always made to selfishly maximize her individual
subjective utility function. Each agent's belief $\mu_{g}$ is supported on the parameters in her model
that minimize a weighted KL-divergence objective in situation $G$,
with the data from each type of match weighted by the probability
of confronting this type of opponent. The use of KL-divergence minimization as the inference procedure is standard in the misspecified Bayesian learning literature, as in \cite{esponda2016berk}. We note that here we assume inference occurs separately across situations. This reflects situation persistence, with agents having enough data to establish new beliefs and behavior if the situation were to change. Our learning foundation in Online Appendix \ref{sec:Learning-Foundation} justifies this situation-by-situation updating, but we omit the details here as it otherwise plays no role in our results. 

\subsection{Evolutionary Stability of Models}

Given a distribution $q\in\Delta(\mathcal{G})$ and an EZ, we
define the \emph{fitness} of each model as the expected objective payoff of its adherents in the EZ when $G$ is drawn according
to $q$. We have in mind an evolutionary story where the relative success
of the two models depends on their
relative fitness, so that one model is more successful if the objective expected payoffs are higher. Given this criterion, our question of interest is: Can the adherents of a \emph{resident
model} $\Theta_{A}$, starting at a position of social prominence,
always repel an invasion from a small $\epsilon$ mass of agents who
adhere to a \emph{mutant model} $\Theta_{B}$?

Evolutionary stability depends on the fitness of models $\Theta_{A},\Theta_{B}$
in EZs with $p_{A}=1, p_B=0$. But it is motivated by the invasion of
a small but strictly positive population of model $\Theta_{B}$ adherents
into an otherwise homogeneous society of model $\Theta_{A}$ adherents.
Below, we directly analyze EZs with $p=(1,0)$, but
note that these EZs can be written as the limit of EZs where the population
share of $\Theta_{B}$ is positive but approaching 0. Online
Appendix \ref{sec:Existence-and-Continuity} provide conditions for
the existence of an EZ with $p=(1,0)$ and to ensure that any limit
of EZs with positive but diminishing fraction of $\Theta_{B}$ remains
an EZ with $p=(1,0)$. 

\begin{defn}
\label{def:stability} Say $\Theta_{A}$ is \emph{evolutionarily stable
{[}fragile{]} }against $\Theta_{B}$ under $\lambda$-matching if
there exists at least one EZ with models $\Theta_{A},\Theta_{B}$,
$p=(1,0),$ matching assortativity $\lambda$ and, in all such EZs,
$\Theta_{A}$ has a weakly higher {[}strictly lower{]} fitness than
$\Theta_{B}$.
\end{defn}

Evolutionary stability is when $\Theta_{A}$ has higher fitness than
$\Theta_{B}$ in all EZs, and evolutionary fragility is when $\Theta_{A}$
has lower fitness in all EZs.\footnote{If
the set of EZs is empty, then $\Theta_{A}$ is neither evolutionarily stable
nor evolutionarily fragile against $\Theta_{B}.$} These two cases give sharp predictions about whether a small share
of mutant-model invaders might grow in size, across all equilibrium
selections. A third possible case, where $\Theta_{A}$ has lower fitness
than $\Theta_{B}$ in some but not all  EZs, correspond to a situation
where the mutant model may or may not grow in the society, depending
on the equilibrium selection.

\subsection{Discussion} 

Our model applies the ``indirect evolutionary approach'' framework (see \citet{robson2011evolutionary}) to settings where agents can draw inferences (especially misspecified inferences). Suppose $\Theta = \{F\}$ is
a singleton model that only contains one parameter. Then $\Theta$ also determines preferences in the stage game with subjective
utility function $(a_{i},a_{-i})\mapsto\mathbb{E}_{y\sim F(a_{i},a_{-i})}[\pi(y)]$. In this special case, our equilibrium and stability concepts coincide with those used in an existing literature that studies which preferences are selected by evolution (see, for instance, \citet{Alger2019survey} for a survey). Models are more general than preferences in that agents may adapt their beliefs (which determine their subjective preferences) endogenously. The reason we introduce \emph{zeitgeists} is, relative to other evolutionary frameworks, ours requires beliefs about the data generating process, $\mu$, to be incorporated. Allowing for multiple situations is the most direct way for inference itself to be beneficial, although one could alo study settings with multiple situations without inference (e.g., \citet{GuthNapel2006}). 

An important assumption is that agents (correctly)
believe the economic fundamentals (represented by $G$) do not vary
depending on which group they are matched against. That is, the mapping
$(a_{i},a_{-i})\mapsto\Delta(\mathbb{Y})$ describes the stage game
that they are playing, and agents know that they always play the same
stage game even though opponents from different groups may use different
strategies in the game. As a result, the agent's experience in games
against both groups of opponents jointly resolve the same fundamental
uncertainty about the environment.\footnote{We note that play between two groups
$g$ and $g^{'}$ is not a Berk-Nash equilibrium \citep{esponda2016berk}, since adherents from one group draw inferences about the game's
parameters from the matches against the other group, which may adopt a
different strategy. A Berk-Nash equilibrium between groups $g$ and $g'$ would require inferences to \emph{only} be made from data generated in the match between
$g$ and $g'$.} If adherents were able to believe the fundamentals changed depending on their opponent, then this would give a trivial way for in-group preferences to emerge and also trivialize the question of which errors could invade. 

We comment on some other modelling assumptions.  First, our framework assumes that agents can identify which group their matched opponent
belongs to, though we do not assume that agents know the data-generating processes contained
in other models or that they are capable of making inferences using
other models. Observability assumptions are common in the literature on the indirect evolutionary approach; see \cite{Alger2019survey} and \cite{Dekeletal2007} for discussions. While there are a number of ways it can be relaxed, we expect the main insights to carry through given sufficient observability. In our context, one key assumption which makes our approach tractable is that players do not change their inferences in response to seeing their opponents' actions. In other words, players do not necessarily try to ``read into'' what others do when learning. This particular assumption seems plausible in many cases, as the inference problem on its own may be rather complex even before considering such higher-order inferences.  Consider hedge funds that regularly
trade against each other in a variety of settings. Funds hold differing philosophies, with some focusing on fundamental analysis and
others on technical analysis.\footnote{In practice, each fund's model about the financial market is well
known to other market participants, as it is always prominently marketed
to their clients.} But, simply observing another fund's actions would not lead a technical
analyst to embrace efficient markets, or vice versa. Both fundamental analysis and technical analysis are complex
forecasting systems that involve calibrating sophisticated models
and take many years of training and experience to master. In settings such as these,
agents need not know how others' models work even after identifying who they are. 

Second, EZs as presented abstract away from the issues surrounding learning others'
strategies. However, we study
an extension in Section \ref{sec:ABEE} allowing agents to be misspecified
about others' strategies and hold wrong beliefs about these strategies
in equilibrium.

Lastly, even as agents adjust their beliefs and behavior to achieve optimality, population proportions $p_{A}$ and $p_{B}$
remain fixed. We imagine a world where the relative prominence of
models change much more slowly than the rate of convergence to an
EZ. This assumption about
the relative rate of change in the population sizes follows the previous
work on evolutionary game theory (See \citet*{Sandholm2001Preference}
or \citet*{Dekeletal2007}).

\section{Learning Channel and New Stability Phenomena\label{sec:new_stability_phenomena}}

The main novelty of our framework relative to past work on the indirect evolutionary approach is that agents maximize endogenously determined subjective preferences, not exogenously fixed ones. The \emph{learning channel} refers to this endogenous preference formation,  and this section discusses how this mechanism leads to new stability phenomena. 

The idea that agents' personal experiences (and more broadly, the environments that generate these experiences) shape their preferences \emph{beyond} their individual characteristics is empirically well documented. For instance,  recent work studying attitudes toward immigrants (\cite{ThomasPaper}) or attitudes among immigrants (\cite{BrianPaper})  find that variation in a person's environment---plausibly independent from individual characteristics---can considerably influence their political behavior and preferences. In an experiment with Indian men, \cite{Lowe2021} finds that favoritism for one's own caste changes in response to cross-caste contacts, in a way that depends on whether interactions are competitive or cooperative. Our framework derives the implications of these kinds of preference-formation mechanism on the stability of misspecified models.

Misspecified models, unlike correctly specified models or dogmatic preferences, are \emph{polymorphic}: a given model can induce different preferences through the learning channel in different environments. Our framework thus gives a natural setting where environment shapes preference and lets us ask about its implications. We show that this polymorphism strictly expands the possibility of invading rational societies, and it also makes models that seem evolutionarily unfit in one environment surprisingly strong invaders in other environments. We also show how accommodating feedback changes the predictions of the evolutionary framework. We show that the learning channel can suggest invasions under assumptions prohibiting it with exogenous preferences, and lead to greater indeterminacy in stable outcomes. 

\subsection{\label{subsec:When-Is-the}Necessity of the Learning Channel
for Fragility of Rationality}

\label{sec:TheoryOverPref}

Our first result characterizes when a misspecified model can \emph{only} invade a rational society when inference is possible, due to the gain achieved via adapting preferences to the relevant situation. The following example illustrates: 

\begin{example}
\label{exa:only_theory_invades} Suppose there are two situations, $G_{A}$ and $G_{B}$, which are equally likely, and consequences  $\text{\ensuremath{\mathbb{Y}}}=\{g,b\},$
with $u(g)=1$ and $u(b)=0.$ Suppose that the probability a given player obtains $y$ given an action profile and situation is determined by the table below. \begin{center}
\begin{tabular}{|c|c|c|c|}
\hline 
{\small{}$G_{A}$} & {\small{}$a_{1}$} & {\small{}$a_{2}$} & {\small{}$a_{3}$}\tabularnewline
\hline 
\hline 
{\small{}$a_{1}$} & {\small{}0.1, 0.1} & {\small{}0.1, 0.1} & {\small{}0.1, 0.11}\tabularnewline
\hline 
{\small{}$a_{2}$} & {\small{}0.1, 0.1} & 0.3, 0.3 & {\small{}0.1, 0.1}\tabularnewline
\hline 
{\small{}$a_{3}$} & {\small{}0.11, 0.1} & {\small{}0.1, 0.1} & {\small{}0.2, 0.2}\tabularnewline
\hline 
\end{tabular}{\small{} \qquad{}}%
\begin{tabular}{|c|c|c|c|}
\hline 
{\small{}$G_{B}$} & {\small{}$a_{1}$} & {\small{}$a_{2}$} & {\small{}$a_{3}$}\tabularnewline
\hline 
\hline 
{\small{}$a_{1}$} & {\small{}0.11, 0.11} & {\small{}0.5, 0.5} & {\small{}0.12, 0.4}\tabularnewline
\hline 
{\small{}$a_{2}$} & {\small{}0.5, 0.5} & {\small{}0.12, 0.12} & {\small{}0.14, 0.55}\tabularnewline
\hline 
{\small{}$a_{3}$} & {\small{}0.4, 0.12} & {\small{}0.55, 0.14} & {\small{}0.4, 0.4}\tabularnewline
\hline 
\end{tabular}
\par\end{center}
Taking $\lambda=0,$ we show the correctly specified model is not evolutionarily
fragile against any singleton mutant model $\Theta=\{F\}$. Indeed, the minimal correctly specified model obtains objective fitness .35 if  $(a_{2},a_{2})$ in situation $G_{A}$ and  $(a_{3},a_{3})$ in situation $G_{B}$ are played, as these are Nash equilibria. But under the singleton model $\{F\}$, one of the three must hold: 

\begin{itemize} 
\item If $a_{3}$ is a best response to $a_{3}$ under $F$, there is an EZ where $(a_{3},a_{3})$ is always the outcome, and the expected fitness is $.3 < .35$

\item If $a_{2}$ is a best response to $a_{3}$ under $F$, there is an EZ where $(a_{2},a_{3})$ is played by the mutant and resident in $G_{B}$, so the mutant's payoff is at most $\frac{1}{2}.3+\frac{1}{2}.14 < .35$

\item If $a_{1}$ is a best response to $a_{3}$ under $F$, then there is an EZ where $(a_{1},a_{3})$ is played by the mutant and resident in $G_{A}$, so the mutant's payoff is at most $\frac{1}{2}.1+ \frac{1}{2}.55 < .35$.
\end{itemize}

\noindent Thus, the minimal correctly
specified model is not evolutionarily fragile against any singleton. However, consider the misspecified model $\Theta= \{F_{A},F_{B}\}$, where both $F_{A}$ and $F_{B}$ depend only on one's own strategies and not the opponent's. Under $F_{A},$ $a_{1},a_{2},a_{3}$ lead to consequence $g$ with probabilities 0.1, 0.3, and 0.2 respectively. Under $F_{B},$ playing $a_{1},a_{2},a_{3}$ lead to consequence $g$ with probabilities 0.5, 0.14, and 0.4 respectively.

The resident minimal correctly specified model is evolutionarily fragile
against this misspecified model. Note that the mutants never choose $a_{3}$, since this is dominated under both $F_{A}$ and $F_{B}$. Next, note that mutants would play $a_{2}$ when believing $F_{A}$ and $a_{1}$ when believing $F_{B}$.  We show these mutants play $a_{2}$ in $G_{A}$ and $a_{1}$
in $G_{B}$ against the resident. Indeed, if mutants were to play $a_{1}$ in situation $G_{A}$, the  correctly specified residents would
best respond with $a_{3}$ in $G_{A}$. The mutants then learn $F_{A}$ in
$G_{A}$, and would then deviate
to $a_{2}.$ If mutants play $a_{2}$ in situation $G_{B},$ once
again the residents best respond with $a_{3}$ in $G_{B}$, and
the mutants learn $F_{B}$. But under $F_{B},$ the mutants
believe they should deviate to $a_{1}.$ These arguments rule out
all other EZ behavior, so the mutants must play $a_{2}$ in $G_{A}$
and $a_{1}$ in $G_{B}$. In this EZ, mutant fitness is $(1/2).3+(1/2).5=.4>.3$, higher than the resident's fitness

\end{example}

The previous example feature two notable features: (1) A misspecification resembling an ``illusion of control'' whereby individuals believe consequences only depend on their own actions, and 
(2) Inferences leading to a belief that a desirable action is dominant, in each situation. Models of this form allow us to determine when the ability to draw misinferences strictly
expands the scope for invasion against rationality. Intuitively, if mutants can adopt
the optimal commitment situation-by-situation, then the learning
channel allows the mutants to tailor their commitment. But a mutant with only one model (i.e.,
an exogenous subjective preference) lacks the flexibility to play
differently in different situations.

Some notation is needed to state the general result. Consider an arbitrary situation $G$. We let $v_{G}^{\text{NE}}\in\mathbb{R}$ be the highest symmetric
Nash equilibrium payoff in $G$, when agents choose
strategies from $\mathbb{A}$. For each $a_{i}\in\mathbb{A}$,
we let $\underline{\text{BR}}(a_{i},G)$ be a rational best response
against the strategy $a_{i}$ in situation $G,$ breaking ties \emph{against}
the user of $a_{i}$. Let $\bar{v}_{G}\in\mathbb{R}$ be the Stackelberg
equilibrium payoff in situation $G$, breaking ties against the
Stackelberg leader, i.e.,

\begin{equation}
\bar{v}_{G}:=\max_{a_{i}}U_{i}(a_{i},\underline{\text{BR}}(a_{i},G),F^{\bullet}(G)).\label{eq:stackleberg}
\end{equation}

\noindent Call the strategy $\bar{a}_{G}$ that maximizes Equation
(\ref{eq:stackleberg}) the Stackelberg strategy in situation $G$.
We assume the Stackelberg strategy is unique in each situation, and
furthermore there is a unique rational best response to $\bar{a}_{G}$ in each situation $G',$ where possibly $G\ne G'$. Finally, let $v_{G}^{b}$
denote the worst equilibrium payoff of an agent with the subjective best-response
correspondence $b$ when she plays against a rational opponent in
situation $G.$\footnote{More formally, given correspondence $b:\mathbb{A}\rightrightarrows\mathbb{A}$,
let $v_{G}^{b}\in\mathbb{R}$ be defined as $i$'s lowest payoff across
all strategy profiles $(a_{i},a_{-i})$ such that $a_{i}\in b(a_{-i})$
and $a_{-i}$ is a rational response to $a_{i}$ in situation $G.$
If no such profile exists, let $v_{G}^{b}=-\infty.$} 

We impose two identifiability conditions: 
\begin{defn}
\emph{Situation identifiability} is satisfied if for every $a_{i},a_{-i}\in\mathbb{A}$
and $G\ne G',$ we have $F^{\bullet}(a_{i},a_{-i},G)\ne F^{\bullet}(a_{i},a_{-i},G').$
\emph{Stackelberg identifiability} is satisfied if whenever $G\ne G'$
and $a_{-i}$, $a_{-i}'$ are rational best responses to $\bar{a}_{G}$
in situations $G$ and $G'$, we have $F^{\bullet}(\bar{a}_{G},a_{-i},G)\ne F^{\bullet}(\bar{a}_{G},a_{-i}',G')$.
\end{defn}

Under situation identifiability, a minimal correctly specified agent can identify
the true situation. Under Stackelberg identifiability, playing $\bar{a}_{G}$ in situation $G$ leads to different consequences
than playing the same strategy in situation $G'\ne G$,
provided the opponent chooses the rational best response to the strategy. We can now state our result.
\begin{thm}
\label{thm:theoryneeded} Suppose $\lambda=0$, there are finitely
many situations, and there is a symmetric Nash equilibrium in $\mathbb{A}\times\mathbb{A}$
for every situation $G$.
\begin{enumerate}
\item If there is no point $(u_{G})_{G\in\mathcal{G}}$ in the convex hull
of $\{(v_{G}^{b})_{G\in\mathcal{G}}\mid b:\mathbb{A}\rightrightarrows\mathbb{A}\}$
with the property that $u_{G}\ge v_{G}^{\text{NE}}$ for every $G\in\mathcal{G},$
then there exists a full-support distribution $q\in\Delta(\mathcal{G})$ so
that the correctly specified model is not evolutionarily fragile
against any singleton model.
\item If $v_{G}^{\text{NE}}<\bar{v}_{G}$ for some $G$, situation identifiability
and Stackelberg identifiability hold, and there are finitely many
strategies, then there exists a model $\hat{\Theta}$ such that the
correctly specified model is evolutionarily fragile against $\hat{\Theta}$
under any full-support distribution $q\in\Delta(\mathcal{G})$.
\end{enumerate}
\end{thm}

The core of the proof uses a separating hyperplanes argument to determine a distribution $q$ under which the rational model cannot be invaded. One can check that indeed Example \ref{exa:only_theory_invades} satisfies both conditions of Theorem \ref{thm:theoryneeded}. Whenever the conditions are satisfied, the minimal correctly specified model is evolutionarily fragile
against some mutant model, but not evolutionarily fragile against
any \emph{singleton} mutant model. In these environments, the ability  adapt preferences endogenously to the relevant situation (i.e., the learning channel) is a necessary condition for an invading mutant
to displace the rational incumbent.  Hence, this result shows that
mutants with misspecified models cannot in general be represented
simply as mutants with fixed subjective best-response correspondences.

\subsection{\label{subsec:Stability-Reversals}Stability Reversals}

We now show that the learning channel can lead to greater indeterminacy in the emergence of stable biases. For expositional simplicity, we assume that $|\mathcal{G}|=1$ throughout this section. We will refer to  a model's \emph{conditional fitness against group $g$}, i.e., the expected payoff of the model's adherents in matches against
group $g.$
\begin{defn}
Two models $\Theta_{A},\Theta_{B}$
exhibit \emph{stability reversal} if (i) in every EZ with $\lambda=0$
and $(p_{A},p_{B})=(1,0),$ $\Theta_{A}$ has strictly higher conditional
fitness than $\Theta_{B}$ against group A opponents and against group
B opponents, but also (ii) in every EZ with $\lambda=0$ and $(p_{A},p_{B})=(0,1),$
$\Theta_{B}$ has strictly higher fitness than $\Theta_{A}$.
\end{defn}

When $p_B=0$, how $\Theta_A$ performs against $\Theta_B$ does not actually affect group A's fitness. Condition (i) encodes the strong requirement that $\Theta_A$  outperforms $\Theta_B$ even on the zero-probability event of being matched against a $\Theta_B$ opponent.  A stability reversal occurs if this stronger requirement holds (when $\Theta_A$ dominates in society), and yet $\Theta_{B}$ is still stable against $\Theta_{A}$ (if $\Theta_{B}$ starts from a position of prominence). 

We begin with two general results on when stability reversals \emph{cannot} emerge. First, it cannot emerge without the learning channel: 

\begin{prop}
\label{prop:no_reversal}Suppose $|\mathcal{G}|=1$. Two singleton models (i.e., two subjective
preferences in the stage game) cannot exhibit stability reversal. 
\end{prop}

Additionally, stability reversals cannot emerge in decision problems. We show this by introducing a class of games where strategic interactions do not matter: 

\begin{defn}
A model $\Theta$ is \emph{strategically independent} if for all
$\mu\in\Delta(\Theta)$, $\underset{a_{i}\in\mathbb{A}}{\arg\max}\ U_{i}(a_{i},a_{-i};\mu)$
is the same for every $a_{-i}\in\mathbb{A}.$
\end{defn}
\noindent The adherents of a strategically independent model believe that while
opponent's action may affect their utility, it does not affect their
best response.
\begin{prop}
\label{prop:reversal_inference_channel}Suppose $|\mathcal{G}|=1$,
suppose $\Theta_{A},\Theta_{B}$ exhibit stability reversal and $\Theta_{A}$
is the correctly specified singleton model. Then, the beliefs that
the adherents of $\Theta_{B}$ hold in all EZs with $p=(1,0)$ and
the beliefs they hold in all EZs with $p=(0,1)$ form disjoint sets.
Also, $\Theta_{B}$ is not strategically independent.
\end{prop}
\noindent The first claim of Proposition \ref{prop:reversal_inference_channel}
underscores that stability reversal require inference---it cannot happen if group B agents merely have a different subjective preference. The second claim
shows that stability reversal can only happen if the misspecified
agents respond differently to different rival play, immediately implying they cannot emerge in decision problems.

We now show by example that stability reversal can emerge with models that allow for inference. Consider a two-player investment game where player $i$
chooses an investment level $a_{i}\in\{1,2\}.$ A random productivity
level $P$ is realized according to $b^{\bullet}(a_{i}+a_{-i})+\epsilon$
where $\epsilon$ is a zero-mean noise term, $b^{\bullet}>0$. Player
$i$'s payoffs are $a_{i}\cdot P-1_{\{a_{i}=2\}}\cdot c$. Consequences are $y=(a_{i},a_{-i},P).$ We record the payoff matrix of this investment game: 
\begin{center}
\begin{tabular}{|c|c|c|}
\hline 
 & 1 & 2\tabularnewline
\hline 
\hline 
1 & $2b^{\bullet},2b^{\bullet}$ & $3b^{\bullet},6b^{\bullet}-c$\tabularnewline
\hline 
2 & $6b^{\bullet}-c,3b^{\bullet}$ & $8b^{\bullet}-c,8b^{\bullet}-c$\tabularnewline
\hline 
\end{tabular}
\par\end{center}

\begin{condition}
\label{cond:medium_cost}$5b^{\bullet}<c<6b^{\bullet}$.
\end{condition}
In words, we assume that $a_{i}=1$ is a strictly
dominant strategy in the stage game, but the investment profile (2,2)
Pareto dominates the investment profile (1,1).
Consider two models in the society. Take $\Theta_{A}$ to be a correctly
specified singleton (thus knowing the true mapping from actions to payoffs), while $\Theta_{B}$
wrongly stipulates $P=b(x_{i}+x_{-i})-m+\epsilon$, where $m>0$ is fixed, while $b\in\mathbb{R}$ is a parameter
that the adherents infer. We impose a condition on $\Theta_{B}$, which holds whenever $m>0$ is large enough:.
\begin{condition}
\label{cond:large_misspec} $c<4b^{\bullet}+\frac{1}{3}m$ and $c<5b^{\bullet}+\frac{1}{4}m.$
\end{condition}
We show that in this example models $\Theta_{A}$ and $\Theta_{B}$
exhibit stability reversal.
\begin{example}
\label{exa:stability_reversal_example}In the investment game, under
Condition \ref{cond:medium_cost} and Condition \ref{cond:large_misspec},
$\Theta_{A}$ and $\Theta_{B}$ exhibit stability reversal.
\end{example}
The idea is that the adherents of $\Theta_{B}$ overestimate the complementarity
of investments, and this overestimation is more severe when they face
data generated from lower investment profiles. As a result, the match
between $\Theta_{A}$ and $\Theta_{B}$ plays out in a different way
depending on which model is resident: it results in the investment
profile $(1,2)$ when $\Theta_{A}$ is resident, but results in $(1,1)$
when $\Theta_{B}$ is resident. (We relegate the formal argument to Appendix \ref{App:ExDetails}.) Due to Propositions \ref{prop:no_reversal} and \ref{prop:reversal_inference_channel}, we conclude that this example is possible due to the non-trivial strategic interactions and $\Theta_{B}$'s inference about $b$ (i.e., the learning channel). 

Stability reversals provide a clear demonstration of polymorphism in models that permit inference. A mutant model   may appear very weak when present in small proportions, doing worse than the incumbent model conditional on every type of opponent. Yet, if the population share of the mutant model reaches a critical mass, its adherents infer a more evolutionarily advantageous model parameter based on their within-group interactions, change their best-response correspondence, and hence  outperform the adherents of the incumbent model.

\subsection{\label{subsec:Non-Monotonic-Stability-in}Non-Monotonic Stability
in Matching Assortativity}

Our last general result shows another unique prediction of the learning channel: a mutant model might successfully invade \emph{only} when  matching assortativity in the society is intermediate. This non-monotonicity in stability arises because a misspecified agent can draw different inferences about the game's fundamentals depending on the relative frequency of in-group and out-group interactions, as these two groups of opponents choose different actions. The idea that social interaction structure shapes people's beliefs about the world has been empirically documented,\footnote{For example, \cite{BazziEtAl2019} document how ethnic attachment in response to a resettlement policy in Indonesia has varying effects depending on whether a community is ``fractionalized'' (so that most interactions are not with one's own group members, i.e., $\lambda$ is small) versus polarized (so that most interactions are with one's own group, i.e., $\lambda$ is large).} and our framework accommodates this mechanism and shows how it affects the stability of misspecified models. 

We again assume there is only one situation,
for simplicity. Note that without inference (i.e., in the setting of preference evolution), the fitness of a
group is linear in matching assortativity. Thus, for singleton models, $\Theta_{A}$ being evolutionarily stable against $\Theta_{B}$
both when $\lambda=0$ and when $\lambda=1$ implies the same holds for all $\lambda\in(0,1).$
\begin{prop}
\label{prop:no_non_mono_lambda}Suppose $\Theta_{A},\Theta_{B}$ are
singleton models (i.e., subjective preferences in the stage game)
and $\Theta_{A}$ is evolutionarily stable against $\Theta_{B}$ with
$\lambda$-matching for both $\lambda=0$ and $\lambda=1.$ Then,
$\Theta_{A}$ is also evolutionarily stable against $\Theta_{B}$
with $\lambda$-matching for any $\lambda\in[0,1]$.
\end{prop}

Crucially, inference leads to cases where the relevant ``preference'' changes depending on how frequently a model interacts with different types of opponents. This means a model's fitness may be \emph{non-linear} in the matching probabilities. This phenomenon is a distinguishing feature of our framework and we show that the conclusion of Proposition \ref{prop:no_non_mono_lambda} need not hold for models that allow for parameter inferences. 

Consider a stage game where each player chooses an action from $\{a_{1},a_{2},a_{3}\}.$
Every player then receives a random prize, $y\in\{g,b\}$, which are
worth utilities $\pi(g)=1,$ $\pi(b)=0.$ The payoff matrix below
displays the objective expected utilities associated with different
action profiles, which also correspond to the probabilities that the
row and column players receive the good prize $g$.
\begin{center}
\begin{tabular}{|c|c|c|c|}
\hline 
 & $a_{1}$ & $a_{2}$ & $a_{3}$\tabularnewline
\hline 
\hline 
$a_{1}$ & 0.25, 0.25 & 0.50, 0.20 & 0.70, 0.15\tabularnewline
\hline 
$a_{2}$ & 0.20, 0.50 & 0.40, 0.40 & 0.40, 0.20\tabularnewline
\hline 
$a_{3}$ & 0.15, 0.70 & 0.20, 0.40 & 0.20, 0.20\tabularnewline
\hline 
\end{tabular}
\par\end{center}

Let $\Theta_{A}$ be the correctly specified singleton model. The
action $a_{1}$ is strictly dominant under the objective payoffs,
so an adherent of $\Theta_{A}$ always plays $a_{1}$ in all matches.
Let $\Theta_{B}$ be a misspecified model $\Theta_{B}=\{F_{H},F_{L}\}$.
Each model $F_{H},F_{L}$ stipulates that the prize $g$ is generated
the probabilities in the following table, where $b$ and $c$ are
parameters that depend on the model. The model $F_{H}$ has $(b,c)=(0.8,0.2)$
and $F_{L}$ has $(b,c)=(0.1,0.4).$
\begin{center}
\begin{tabular}{|c|c|c|c|}
\hline 
 & $a_{1}$ & $a_{2}$ & $a_{3}$\tabularnewline
\hline 
\hline 
$a_{1}$ & 0.10, 0.10 & 0.10, $c$ & 0.10, 0.15\tabularnewline
\hline 
$a_{2}$ & $c,$ 0.10 & $b,b$ & $b,$ 0.20\tabularnewline
\hline 
$a_{3}$ & 0.15, 0.10 & 0.20, $b$ & 0.20, 0.20\tabularnewline
\hline 
\end{tabular}
\par\end{center}

The learning channel for the biased mutants leads the correctly specified
model to have non-monotonic evolutionarily stability in terms of
matching assortativity.
\begin{example}
\label{exa:non_mono_example}In this stage game, $\Theta_{A}$ is
evolutionarily stable against $\Theta_{B}$ under $\lambda$-matching
when $\lambda=0$ and $\lambda=1,$ but it is also evolutionarily
fragile under $\lambda$-matching when $\lambda\in(\lambda_{l},\lambda_{h})$,
where $0<\lambda_{l}<\lambda_{h}<1$ are $\lambda_{l}=0.25$, $\lambda_{h}\approx0.56$.
\end{example}
Consider the match between two adherents
of $\Theta_{B}.$ If they believe in $F_{H}$, they will play the
action profile $(a_{2},a_{2})$ and payoff
profile $(0.4,0.4)$, a Pareto improvement compared to the correctly
specified outcome $(a_{1},a_{1})$. The problem is that the data
from play of $(a_{2},a_{2})$ fit $F_{L}$ better than
than $F_{H}$, since the objective 40\% probability of getting prize
$g$ is closer to $F_{L}$'s conjecture (10\%) than $F_{H}$'s conjecture (80\%). A belief in $F_{H}$ --- and hence the profile $(a_{2},a_{2})$
--- cannot be sustained if the mutants only play each other. On the
other hand, when an adherent of $\Theta_{B}$ plays a correctly specified
$\Theta_{A}$ adherent, both  $F_{H}$ and $F_{L}$ prescribe
a best response of $a_{2}$ against the $\Theta_{A}$ adherent's play
$a_{1}.$ The data generated from the $(a_{2},a_{1})$ profile lead
biased agents to the parameter $F_{H}$ that enables cooperative behavior
within the mutant community. But, these matches against correctly
specified opponents harm the mutant's welfare, as they only get an
objective payoff of 0.2.

Therefore, the most advantageous interaction structure for the mutants
is one where they can infer $F_{H}$ using the data
from matches against correctly specified opponents, then extrapolate
this optimistic belief about $b$ to coordinate on $(a_{2},a_{2})$
in matches against fellow mutants. This requires the mutants to match
with intermediate assortativity. Figure \ref{fig:3x3} depicts the
equilibrium fitness of the mutant model $\Theta_{B}$ as a function
of assortativity. While payoffs of $\Theta_{B}$ adherents increase
in $\lambda$ at first, eventually they drop when mutant-vs-mutant
matches become sufficiently frequent that a belief in $F_{H}$ can
no longer be sustained. Note that a similar conclusion obtains with fixed $\lambda$ and varying population sizes: what actually matters is the probability $\Theta_{B}$ with which interacts with each model. Non-linearity of fitness in the population shares can emerge here as well, also a unique possibility due to inference.\footnote{See Section \ref{sect:StableShares} for a discussion of stability with intermediate population shares.}

\begin{figure}
\begin{centering}
\includegraphics[scale=0.5]{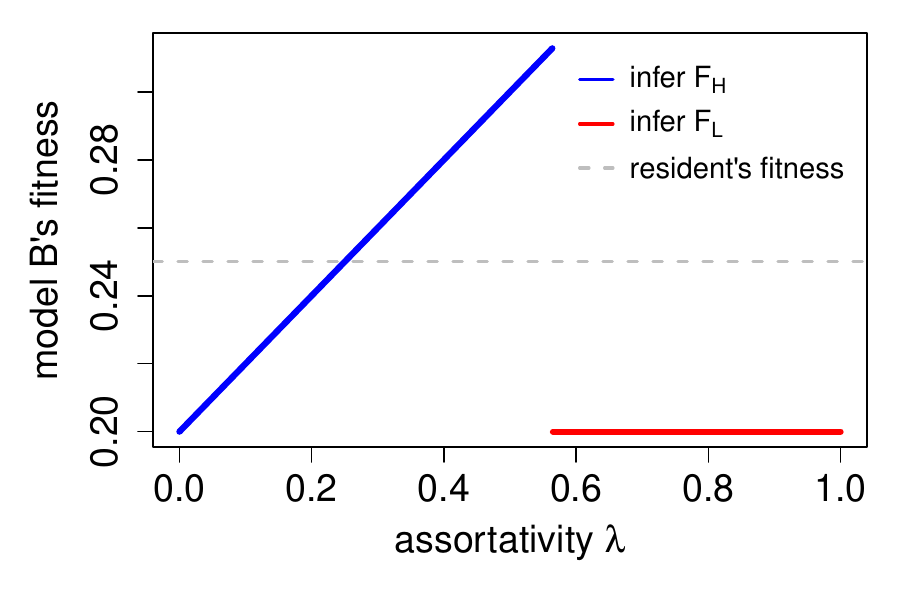}
\par\end{centering}
\vspace{-13bp}

\caption{\label{fig:3x3}The EZ fitness of $\Theta_{B}$ for different values
of $\lambda$ when $p_{B}=0$. (The EZ fitness
of the resident model $\Theta_{A}$ is always 0.25.) In the blue
region, adherents of $\Theta_{B}$
infer $F_{H}$ and receive linearly increasing average payoffs across
all matches as $\lambda$ increases. In the red region, adherents of $\Theta_{B}$ infer $F_{L}$ and receive
payoff 0.2 in all matches, regardless of $\lambda$.}
\end{figure}

\section{\label{sec:LQN}Higher-Order Misspecifications in LQN Games}

Section \ref{sec:new_stability_phenomena} showed that the learning channel can in general lead to new stability phenomena. Next, we illustrate the relevance of the learning channel applied to a specific economically significant bias. This bias relates to how players perceive the correlation in private information in a strategic setting. We work with a class of linear quadratic normal (LQN) games. While prior work has exploited the tractability of this classic framework to derive comparative statics with respect to information (e.g., \citet{bergemann2013robust}), we innovate by accommodating both misspecifications and inference. In the main text we focus
on a Cournot duopoly application, and extend the insights to general LQN games in Appendix \ref{subsec:LQNGeneral}.

\subsection{Stage Game and Misperceptions of Information Structure}

We first describe the stage game. There is a demand state $\omega \sim \mathcal{N}(0,\sigma_{\omega}^{2})$, where
$\mathcal{N}(\mu,\sigma^{2})$ is the normal distribution with mean
$\mu$ and variance $\sigma^{2}.$ Each player is a firm, with firm $i$ receiving a private signal $s_{i}=\omega+\epsilon_{i},$ and then
choosing $q_{i}\in\mathbb{R}$ (i.e., a quantity). The resulting market
price is $P=\omega-r^{\bullet}\cdot\frac{1}{2}(q_{1}+q_{2})+\zeta$,
where $\zeta\sim\mathcal{N}(0,(\sigma_{\zeta}^{\bullet})^{2})$ is
an idiosyncratic independent price shock. Firm $i$'s profit is $q_{i}P-\frac{1}{2}q_{i}^{2}.$

The stage game is parametrized by $\sigma_{\omega}^{2},r^{\bullet},(\sigma_{\zeta}^{\bullet})^{2} > 0$---i.e., variance in
market demand, the elasticity of market price with respect to average
quantity supplied, and the variance of price shocks, respectively. These parameters
remain constant (so $|\mathcal{G}|=1$).
However, demand state $\omega,$ signals $(s_{i})$, and
price shock $\zeta$ are redrawn independently across matches.

Note that market prices and quantity choices may be positive or negative.
To interpret, when $P>0,$ the market pays for each unit of good supplied,
and market price decreases in total supply. When $P<0,$ the market
pays for disposal of the good. The cost $\frac{1}{2}q_{i}^{2}$ represents either a convex
production cost or a convex disposal cost, depending on the sign of
$q_{i}.$

We take the signals within a match to possibly be correlated conditional on $\omega$, and study the perception of this correlation. Recalling that $s_{i}=\omega+\epsilon_{i}$, we assume in particular that
$\epsilon_{i}=\frac{\kappa}{\sqrt{\kappa^{2}+(1-\kappa)^{2}}}z+\frac{1-\kappa}{\sqrt{\kappa^{2}+(1-\kappa)^{2}}}\eta_{i},$
where $\eta_{i}\sim\mathcal{N}(0,\sigma_{\epsilon}^{2})$ is the idiosyncratic
component generated i.i.d. across players and $z\sim\mathcal{N}(0,\sigma_{\epsilon}^{2})$
is the common component. Higher $\kappa$
leads to an information structure with higher conditional correlation.
When $\kappa=0,$ $s_{i}$ and $s_{-i}$ are conditionally uncorrelated
given $\omega$. When $\kappa=1,$ we always have $s_{i}=s_{-i}$.
This functional form for $\epsilon_{i}$ ensures $\text{Var}(s_{i})$ is constant in
$\kappa,$ which facilitates tractability. 

While objectively, $\kappa=\kappa^{\bullet}$, our interest will be in studying misspecifications in $\kappa$. Indeed, this particular bias is common in experiments, many of which show subjects often do not form accurate beliefs about
the beliefs of others. We draw a connection between the misperception
we study and such statistical biases:
\begin{defn}
Let $\tilde{\kappa}$ be a player's perceived $\kappa$. A player
suffers from \emph{correlation neglect} if $\tilde{\kappa}<\kappa^{\bullet}$.
A player suffers from \emph{projection bias} if $\tilde{\kappa}>\kappa^{\bullet}$.
\end{defn}
\noindent Correlation neglect agents believe signals are less correlated relative to the truth, whereas projection
bias agents ``project'' their own information onto others (exaggerating
the similarity between others' signals and their own). We are
agnostic about the origin of these misspecifications,
e.g., cognitive biases or more complex mechanisms,\footnote{For example, \citet*{hansen2021algorithmic} show that multiple agents
simultaneously conducting algorithmic price experiments in the same
market may generate correlated information which get misinterpreted
as independent information, a form of correlation neglect for firms.
\citet{goldfarb2019} structurally estimate a model of thinking cost
and find that bar owners over-extrapolate the effect of today's weather
shock on future profitability.} instead asking whether such misspecifications would persist under selection pressures were they to appear.

\subsection{Formalizing Strategies and Models}

This environment fits into the formalism from Section \ref{sec:Environment-and-Stability}
as follows. A strategy is a function $Q_{i}:\mathbb{R}\to\mathbb{R}$
that assigns a quantity $Q_{i}(s_{i})$ to every signal $s_{i}$, and a strategy is \emph{linear} if $Q_{i}(s_{i})=\alpha_{i}s_{i}$ for every $s_{i}\in\mathbb{R}$ and some $\alpha_{i} \geq 0$.
Since the best response to any linear strategy is linear, regardless
of the agent's belief about the correlation parameter and market price
elasticity (Lemma \ref{lem:BR} in Appendix \ref{subsect:LQNEZs}),
we restrict attention to linear strategies and let $\mathbb{A}=[0,\bar{M}_{\alpha}]$
for $\bar{M}_{\alpha}<\infty,$ with $\alpha_{i}\in\mathbb{A}$ referring
to strategy $Q_{i}(s_{i})=\alpha_{i}s_{i}$.

The stage game is common knowledge except
for $r^{\bullet},\kappa^{\bullet},$ and $\sigma_{\zeta}^{\bullet}$.
Models are dogmatic and possibly wrong about
$\kappa,$ but  allow inferences about $r$ and $\sigma_{\zeta}$.
We set the consequence space for agent $i$ to be $\mathbb{Y}=\mathbb{R}^{3},$ where
$y=(s_{i},q_{i},P)\in\mathbb{Y}$. Consequence $y$ delivers utility $\pi(y):=q_{i}P-\frac{1}{2}q_{i}^{2}.$
Since $\kappa$ indexes models, we write $\Theta(\kappa):=\{F_{r,\kappa,\sigma_{\zeta}}:r\in[0,\bar{M}_{r}],\sigma_{\zeta}\in[0,\bar{M}_{\sigma_{\zeta}}]\}$
for some $\bar{M}_{r},\bar{M}_{\sigma_{\zeta}}<\infty.$ So,
$\Theta(\kappa)$ is a set of parameters which reflect a dogmatic belief in the correlation parameter $\kappa$.
Each $F_{r,\kappa,\sigma_{\zeta}}:\mathbb{A}\times\mathbb{A}\to\Delta(\mathbb{Y})$
is such that $F_{r,\kappa,\sigma_{\zeta}}(\alpha_{i},\alpha_{-i})$
gives the distribution over $i$'s consequences in a stage game with
parameters $(r,\kappa,\sigma_{\zeta})$, when $i$ uses the linear
strategy $\alpha_{i}$ against an opponent using linear strategy $\alpha_{-i}.$ While agents learn about both $r$ and $\sigma_{\zeta}$,
(mis)inferences about $r$ drives the main results.\footnote{Since each firm's profit is linear in the market price,  belief about the variance of the idiosyncratic price shock does not change her expected payoffs or behavior. The parameter $\sigma_{\zeta}$ absorbs changes in the variance of market price, creating significant tractability. To infer $r$, it is only necessary to consider the mean of the market price in the data, not its variance.}

We assumed that the space of feasible linear strategies
$\alpha_{i}\in[0,\bar{M}_{\alpha}]$ and the domain of inference over
game parameters $r\in[0,\bar{M}_{r}],\sigma_{\zeta}\in[0,\bar{M}_{\sigma_{\zeta}}]$
are compact, to guarantee EZ existence. For some of our results, we utilize the following shorthand:
\begin{notation}
A result is said to hold ``\emph{with high enough price volatility
and large enough strategy space and inference space}'' if, whenever
the strategy space $[0,\bar{M}_{\alpha}]$ has $\bar{M}_{\alpha}\ge\frac{1/\sigma_{\epsilon}^{2}}{1/\sigma_{\epsilon}^{2}+1/\sigma_{\omega}^{2}}$,
there exist $0<L_{1},L_{2},L_{3}<\infty$ so that for any objective
game $F^{\bullet}$ with $(\sigma_{\zeta}^{\bullet})^{2}\ge L_{1}$
and with models where $r\in[0,\bar{M}_{r}],$
$\sigma_{\zeta}\in[0,\bar{M}_{\sigma_{\zeta}}]$ are such that $\bar{M}_{\sigma_{\zeta}}^{2}\ge(\sigma_{\zeta}^{\bullet})^{2}+L_{2}$
and $\bar{M}_{r}\ge L_{3}$, the result is true.
\end{notation}

When imposed, these assumptions will ensure behavior and beliefs are interior.  Our analysis relies on a number of technical lemmas, which we defer to Appendix \ref{sec:LQNMore}. We show, for example, that the set
of EZs is non-empty and it is upper hemicontinuous in population sizes.
We also derive there closed-form expressions for the best-fitting
inference and optimal behavior of misspecified agents.

\subsection{The Impact of Misspecification: Some Intuition \label{subsec:CournotIntuition}}

Before presenting our results on the fragility of correct specifications,
we briefly describe what happens when players entertain a dogmatically
misspecified view of $\kappa$.

Most importantly, an agent's inference about $r$ is strictly decreasing
in her belief about the correlation parameter $\kappa.$ To understand
why, assume player $i$ uses the linear strategy $\alpha_{i}$ and
player $-i$ uses the linear strategy $\alpha_{-i}$. After receiving
a private signal $s_{i}$, player $i$ expects to face a price distribution
with a mean that is linearly increasing in $\mathbb{E}[s_{-i}\mid s_{i}]$,
which in turn is linearly increasing in $s_{i}$ (see Appendix \ref{subsect:LQNEZs}
for more details).\footnote{Specifically, Lemma
\ref{lem:kappa} in Appendix \ref{subsect:LQNEZs} shows there exists a
strictly increasing and strictly positive function $\psi(\kappa)$
so that $\mathbb{E}_{\kappa}[s_{-i}\mid s_{i}]=\psi(\kappa)\cdot s_{i}$
for all $s_{i}\in\mathbb{R},\kappa\in[0,1].$} Now, under projection bias $\kappa>\kappa^{\bullet},$
$\mathbb{E}_{\kappa}[s_{-i}\mid s_{i}]$ is excessively steep in $s_{i}$,
since the correlation is higher. For example, following a large and
positive $s_{i},$ the agent overestimates the similarity of $-i$'s
signal and wrongly predicts that $-i$ must also choose a very high
quantity, and thus becomes surprised when market price remains high.
As a result, the agent then wrongly infers that the market price elasticity
must be low. Therefore, in order to rationalize the average market
price conditional on own signal, an agent with projection bias must
infer $r<r^{\bullet}$. For similar reasons, an agent with correlation
neglect infers $r>r^{\bullet}.$

The fact that projection biased players believe the elasticity of
demand is lower than the truth suggests that they will behave more
aggressively than correctly specified players---with the converse
holding for correlation neglect players. Intuitively, if price reacts
less to quantity, then players should price more aggressively. While
this does turn out to be true---and drives many of the results below---things
are more subtle because increasing $\kappa$ has an \emph{a priori}
ambiguous impact on the agent's equilibrium aggressiveness. In fact,
in our characterization, we show that increasing $\kappa$ but holding
fixed the player's belief about price elasticity has the direct effect
of \emph{lowering} aggression. The results below show that the indirect
effect through the learning channel dominates, and the evolutionary
stability of correlational errors are driven by this channel. We show
in Section \ref{subsec:learningkey} that the conclusions are reversed
when we shut down the learning channel.

\subsection{Selecting Biases under Uniform and Assortative Matching}

We now consider the evolutionary instability of correctly specified
beliefs about the information structure. We first take $\lambda=0$; we note that this case requires some technical innovation in order to characterize the \emph{asymmetric} equilibrium
strategy profile in matches between the correctly specified residents
and the projection-biased mutants.
\begin{prop}[Uniform Matching Selects Projection Bias]
\label{prop:LQN_lambda_0}Let $r^{\bullet}>0,$ $\kappa^{\bullet}\in[0,1]$
be given. With high enough price volatility and large enough strategy
space and inference space, there exist $\underline{\kappa}<\kappa^{\bullet}<\bar{\kappa}$
so that taking $(\Theta_{A},\Theta_{B})=(\Theta(\kappa^{\bullet}),\Theta(\kappa))$
for $\kappa\in[\underline{\kappa},\bar{\kappa}]$, there is a unique
 EZ with uniform matching ($\lambda=0$) and $(p_{A},p_{B})=(1,0)$.
The equilibrium fitness of $\Theta(\kappa)$ is strictly higher than
that of $\Theta(\kappa^{\bullet})$ if $\kappa>\kappa^{\bullet}$,
and strictly lower if $\kappa<\kappa^{\bullet}$.
\end{prop}
\noindent Figure \ref{fig:uniformFitness} illustrates how, around $\kappa^{\bullet}$,
mutant payoffs increase in $\kappa$. But misperception only helps to a point---the correct specification becomes evolutionarily
stable for large enough $\kappa$.

\begin{figure}
\hspace{2mm}
    \begin{subfigure}{.44\textwidth}
      \caption{Uniform Matching}
    	\includegraphics[width=.9\linewidth]{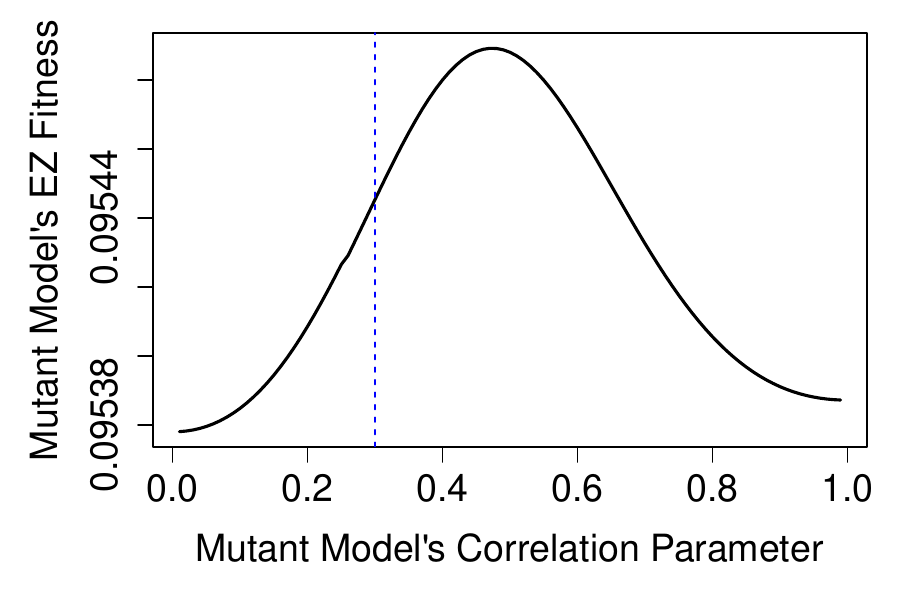}
    	\small
    \label{fig:uniformFitness}
    \end{subfigure}
    \hspace{-2mm}
    \begin{subfigure}{.44\textwidth}
    	\caption{Perfectly Assortative Matching}
    	\includegraphics[width=.9\linewidth]{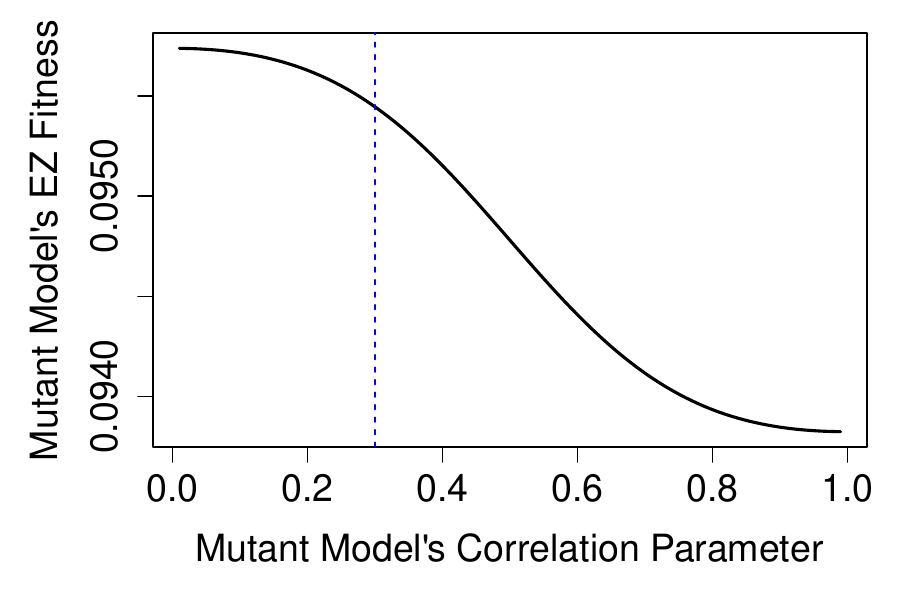}
    	\small
    	\label{fig:assortativeFitness} 
    \end{subfigure}
    \begin{minipage}{1\linewidth}
        \small
        \vspace{-2mm}
        \caption{Fitness of mutant model against a correctly specified resident, as a function
of $\kappa$.}
        \label{figure:distance-beliefs}
        \vspace*{-1em}
        \small
        \singlespacing \emph{Notes}: 
The left panel assumes uniform matching (i.e., $\lambda=0$) and the right panel assumes assortative matching (i.e., $\lambda=1$). 
Both examples take $\kappa^{\bullet}=0.3$, $r^{\bullet}=1$,
$\sigma_{\omega}^{2}=\sigma_{\epsilon}^{2}=1$.
    \end{minipage}
\end{figure}

The intuition for this result follows from the intuition outlined
in Section \ref{subsec:CournotIntuition}--- projection bias generates
a commitment to aggression as it leads the biased agents to under-infer
market price elasticity. It is well known that in Cournot oligopoly
games, such commitment can be beneficial. For instance, if quantities
are chosen sequentially, the first mover obtains a higher payoff compared
to the case where quantities are chosen simultaneously. A similar
force is at work here, but the source of the commitment is different.
Misspecification about signal correlation leads to misinference about
$r^{\bullet}$, which causes the mutants to credibly respond to their
opponents' play in an overly aggressive manner. The rational residents,
who can identify the mutants in the population, back down and yield
a larger share of the surplus. While projection bias is beneficial
in small measure, it is also intuitive that excessive aggression would
be detrimental as well, as overproduction can be individually suboptimal.

By contrast, perfectly assortative matching favors biases which lead to more \emph{cooperative} behavior, and thus the commitment to aggression is detrimental to fitness. Correspondingly, we obtain the opposite result:
evolutionary stability selects correlation neglect.
\begin{prop}[Perfectly Assortative Matching Selects Correlation Neglect]
\label{prop:LQN_lambda_1}Let $r^{\bullet}>0,$ $\kappa^{\bullet}\in[0,1]$
be given. With high enough price volatility and large enough strategy
space and inference space, taking $(\Theta_{A},\Theta_{B})=(\Theta(\kappa_{A}),\Theta(\kappa_{B}))$
where $\kappa_{A}\le\kappa_{B}$, the fitness of $\Theta_{A}$ is
weakly higher than that of $\Theta_{B}$ in every EZ with any population
proportion $p$ and perfectly assortative matching ($\lambda=1$).
\end{prop}

\noindent Correlation neglect leads agents to over-infer market price elasticity, enabling commitment
to more cooperative behavior (i.e., linear strategies with a smaller
coefficient $\alpha_{i}$). Rational opponents would take advantage
of such agents, but biased agents never match up against rational
opponents in a society with perfectly assortative matching. 
The contrast with uniform matching is illustrated in Figure
\ref{fig:assortativeFitness}---when $\lambda=1$, the misspecified
agents' payoffs are \emph{decreasing} in $\kappa$ around the true $\kappa^{\bullet}$.

In fact, the fragility of the correct specification is even starker when $\lambda=1$ compared to $\lambda=0$. Proposition \ref{prop:LQN_lambda_1} implies that mutant fitness is not only locally decreasing in $\kappa$ around $\kappa^{\bullet}$, but monotonic \emph{for all} $\kappa$ (whereas Figure \ref{fig:uniformFitness} illustrated the possibility of non-monotonicity of fitness in $\kappa$). Indeed, letting $\alpha^{TEAM}$
denote the symmetric linear strategy profile that maximizes the sum
of the two firms' expected objective payoffs, we show that among
symmetric strategy profiles, players' payoffs strictly decrease in
their aggressiveness in the region $\alpha>\alpha^{TEAM}$. We
also show that with $\lambda=1$ and any $\kappa\in[0,1],$ the equilibrium
play among two adherents of $\Theta(\kappa)$ strictly increases in
aggression as $\kappa$ grows, always being strictly more aggressive
than $\alpha^{TEAM}$. Lowering perception of $\kappa$ confers an
evolutionary advantage by bringing play monotonically closer to $\alpha^{TEAM}$ in equilibrium.

\subsection{\label{subsec:learningkey}The Necessity of the (Mis)Learning Channel}

In the previous sections, the misinference over $r$ allows agents to commit to behavior which increases their equilibrium payoffs against their typical opponents. We establish two results to emphasize that the statistical biases
may not be beneficial on their own, but only become beneficial due to the learning channel. First, assuming a single situation (as we have been working with so far), we show that if players were instead dogmatically correct about $r=r^{\bullet}$, then the predictions in Propositions \ref{prop:LQN_lambda_0} and \ref{prop:LQN_lambda_1} can be reversed:
\begin{prop}
\label{prop:no_learning_channel}Let $r^{\bullet}>0,$ $\kappa^{\bullet}\in[0,1]$
be given. With high enough price volatility and large enough strategy
space and inference space, there exists $\epsilon>0$ so that for
any $\kappa_{l},\kappa_{h}\in[0,1]$, $\kappa_{l}<\kappa^{\bullet}<\kappa_{h}\le\kappa^{\bullet}+\epsilon$,
the correctly specified model $\Theta(\kappa^{\bullet})$ is evolutionarily
stable against the singleton model $\{F_{r^{\bullet},\kappa_{h},\sigma_{\zeta}^{\bullet}}\}$
under uniform matching ($\lambda=0$), and evolutionarily stable against
the singleton model $\{F_{r^{\bullet},\kappa_{l},\sigma_{\zeta}^{\bullet}}\}$
under perfectly assortative matching ($\lambda=1$).
\end{prop}
Using dogmatic beliefs over $r$ to shut down the learning channel,
misperceptions about $\kappa$ that used to confer an evolutionary
advantage for a $\lambda \in \{0,1\}$ can no longer invade
a society of correctly specified residents. Intuitively, this is because
\emph{an error about $\kappa$ has the direct effect of lowering welfare},
but also causes mislearning about $r$ and hence a stronger, indirect
effect of increasing welfare. In the case of uniform matching, for
instance, the direct effect of an increase in the perceived correlation
$\kappa$ is for players to use less aggressive strategies, anticipating
that any favorable signal about market demand is also shared by the
opponent.

For our second result on the necessity of mislearning, suppose that the environment
features multiple situations given by multiple feasible values of
$r^{\bullet}.$ Theorem \ref{thm:theoryneeded} does not apply directly, but the basic intuition
remains the same. Mistaken agents who do not learn have a fixed belief
about $r$ that cannot be  beneficial in all situations (i.e., for all values of $r^{\bullet}$), and so they
do not end up with  higher fitness than rational agents. But, misspecified
agents who can make different inferences about price elasticity in
different situations can invade a rational society. Even though models with $\kappa \neq \kappa^{\bullet}$ do not obtain the Stackelberg payoff in every situation, they outperform the correctly specified model in every situation, which is impossible for  any fully dogmatic model.
\begin{prop}
\label{prop:LQNLearningNeeded} For every $\overline{r}\ge3$, there
exists a $q\in\Delta([0,\bar{r}])$ such that the correctly specified
model is evolutionarily stable against any singleton model with
a fixed $(r,\kappa)$ when $r^{\bullet}\sim q$. On the other hand,
for every $\bar{r}>0$, there exists a projection bias model with
$\kappa>\kappa^{\bullet}$ so that the corrected specified model
is evolutionarily fragile against it for any $\rho\in\Delta([0,\bar{r}])$.
\end{prop}

\section{\label{sec:ABEE}Evolutionary Stability of Analogy Classes}

We now study a second major application---coarse thinking in games. \citet{jehiel2005analogy} introduced analogy-based
expectation equilibrium (ABEE) in extensive-form games, where agents
group opponents' nodes into \emph{analogy
classes} and only keep track of aggregate statistics of opponents'
average behavior within each analogy class. An ABEE is a strategy
profile where agents best respond to the belief that at all nodes
in every analogy class, opponents behave according to the average
behavior in the analogy class. The ensuing literature typically treats analogy classes as exogenously
given, interpreted as arising from coarse feedback or agents' cognitive
limitations.\footnote{Section 6.2 of \citet{jehiel2005analogy} mentions that if players
could choose their own analogy classes, then the finest analogy classes
need not arise, but also says ``it is beyond the scope of this paper
to analyze the implications of this approach.'' In a different class
of games, \citet{jehiel1995limited} similarly observes that another
form of bounded rationality (having a limited forecast horizon about
opponent's play) can improve welfare.} We use our framework to endogenize them.

\subsection{Relaxing the Observability of Strategies}

To study analogy-based reasoning, we relax the assumption that people correctly know others' strategies in equilibrium.  
We introduce the concepts of extended parameters and extended models:

\begin{defn}
An \emph{extended parameter} is a triplet $(a_{A},a_{B},F)$ with $a_{A},a_{B}\in\mathbb{A}$
and $F:\mathbb{A}^{2}\to\Delta(\mathbb{Y}).$ An\emph{ extended model}
$\overline{\Theta}$ is a collection of extended parameters: i.e.,
a subset of $\mathbb{A}^{2}\times(\Delta(\mathbb{Y}))^{\mathbb{A}^{2}}$.
\end{defn}
\noindent In addition to a conjecture $F$ about
how strategy profiles translate into consequences for the agent, extended models also contain conjectures about how group A and group B opponents will act. We
assume the marginal of the extended model on $(\Delta(\mathbb{Y}))^{\mathbb{A}^{2}}$
is metrizable. As before, we also assume each $F$ is given by a density
or probability mass function $f(a_{i},a_{-i}):\mathbb{Y}\to\mathbb{R}_{+}$
for every $(a_{i},a_{-i})\in\mathbb{A}^{2}$. We say that an extended model $\overline{\Theta}$ is \emph{correctly
specified }if $\overline{\Theta}=\mathbb{A}^{2}\times\{F^{\bullet}(\cdot,\cdot,G)\}$,
so the agent can make unrestricted inferences about others' play and
does not rule out the correct data-generating process $F^{\bullet}(\cdot,\cdot,G)$
for any situation $G$.

Defining zeitgeists for extended models is immediate, as we can simply
replace ``model'' with ``extended model'' in Definition \ref{def:Zeitdef}.
The equilibrium notion, however, is subtly different:
\begin{defn}
A zeitgeist with strategic uncertainty $\overline{\mathfrak{Z}}=(\overline{\Theta}_{A},\overline{\Theta}_{B},\mu_{A}(G),\mu_{B}(G),p,\lambda,a(G))_{G\in\mathcal{G}}$
is an\emph{ equilibrium zeitgeist with strategic uncertainty (EZ-SU)}
if for every $G\in\mathcal{G}$ and $g,g^{'}\in\{A,B\},$ 
 $a_{g,g^{'}}(G)\in\underset{\hat{a}\in\mathbb{A}}{\arg\max}\ \mathbb{E}_{(a_{A},a_{B},F)\sim\mu_{g}(G)}\left[\mathbb{E}_{y\sim F(\hat{a},a_{g^{'}})}(\pi(y))\right]$
and, for every $g\in\{A,B\},$ the belief $\mu_{g}(G)$ is supported
on 
\begin{align*}
\underset{(\hat{a}_{A},\hat{a}_{B},\hat{F})\in\overline{\Theta}_{g}}{\arg\min}\left\{ \begin{array}{c}
(\lambda+(1-\lambda)p_{g})\cdot D_{KL}(F^{\bullet}(a_{g,g}(G),a_{g,g}(G),G)\parallel\hat{F}(a_{g,g}(G),\hat{a}_{g})))\\
+(1-\lambda)(1-p_{g})\cdot D_{KL}(F^{\bullet}(a_{g,-g}(G),a_{-g,g}(G),G)\parallel\hat{F}(a_{g,-g}(G),\hat{a}_{-g})
\end{array}\right\} 
\end{align*}
where $-g$ means the group other than $g$.
\end{defn}

\noindent The only difference with Definition \ref{def:EZ} is that the KL divergence
is now taken with respect to the \emph{conjectured opponent's strategy}, part of the extended model. Conjectures now include others' play, in addition to stage game parameters.

\subsection{Defining Stable Population Shares\label{sect:StableShares}}

In this Section, we will also be interested in stable
population shares in a society that contains both rational and misspecified
players. We briefly introduce the following solution concept.

\begin{defn} \label{def:stability_interior}
Given population share $p\in(0,1)$ and an EZ (or EZ-SU), $p$ is said to be a\emph{ stable population share }given the
EZ (or EZ-SU) if both models have the same fitness.
\end{defn}
\noindent Since EZ(-SU)s are defined with interior population
shares, we can calculate the fitness of a model in terms of its
adherents' objective expected payoff. Whereas Definition \ref{def:stability}'s stability notion reflects performance with $(p_{A},p_{B})=(1,0)$, stability with interior
population shares as in Definition \ref{def:stability_interior} correspond to both models being co-existing with equal fitness.

\subsection{Centipede Games and Analogy-Based Reasoning}

We now analyze analogy-based reasoning in
the centipede game in Figure \ref{fig:The-centipede-game} (there
is only one situation, given by the payoffs in this game). P1 and
P2 take turns choosing Across (A) or Drop (D). The non-terminal nodes
are labeled $n^{k}$, $1\le k\le K$ where $K$ is an even number.
P1 acts at odd nodes and P2 acts at even nodes, where 
choosing Drop at $n^{k}$ leads to the terminal node $z^{k}$.
If Across is always chosen, then the terminal node $z^{end}$ is reached.
Every time a player $i$ chooses Across, the sum of
payoffs grows by $g>0,$ but if the opponent chooses Drop next,
$i$'s payoff is $\ell>0$ smaller than $i$'s payoff had they chosen Drop, with $\ell > g$. Thus, if $z^{end}$ is reached, both get $Kg/2;$ if $z^{k}$ is reached when $k$ is odd, both players obtain $\frac{g(k-1)}{2}$; and if if $z^{k}$ is reached when $k$ is even, P1 obtains $\frac{k-2}{2}g-\ell$, and P2 obtains $\frac{k}{2}g+\ell$.

\begin{figure}[h]
\begin{centering}
\includegraphics[scale=0.35]{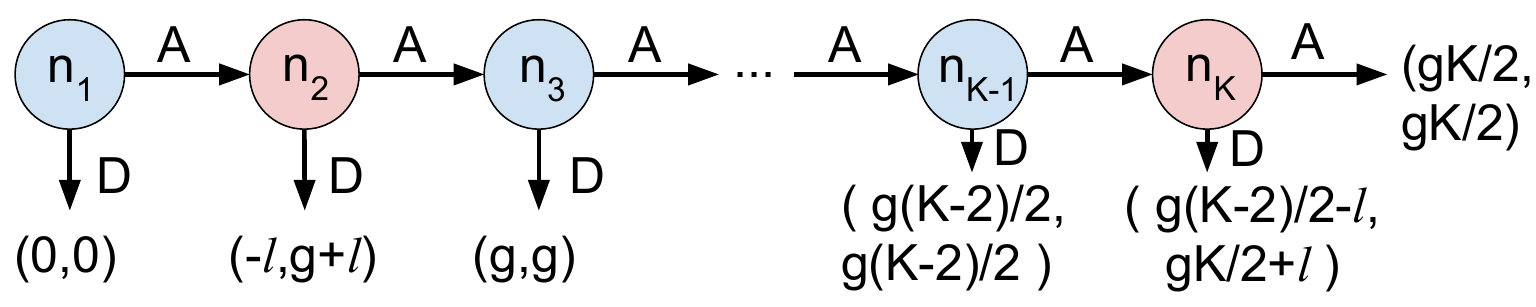}
\par\end{centering}
\vspace{-15bp}

\caption{\label{fig:The-centipede-game}The centipede game. P1 (blue) and P2 (red) alternate in choosing
Across (A) or Drop (D). Payoff profiles are shown at the terminal
nodes.}
\end{figure}

While this is an asymmetric stage game, we study a symmetrized version where two
matched agents are randomly assigned into the roles of P1 and P2.
Let $\mathbb{A}=\{(d^{k})_{k=1}^{K}\in[0,1]^{K}\}$, so each strategy
is characterized by the probabilities of playing Drop at various nodes
in the game tree. When assigned into the role of P1, the strategy
$(d^{k})$ plays Drop with probabilities $d^{1},d^{3},...,d^{K-1}$
at nodes $n^{1},n^{3},...n^{K-1}$. When assigned into the role of
P2, it plays Drop with probabilities $d^{2},d^{4},...,d^{K}$ at nodes
$n^{2},n^{4},...n^{K}$. The set of consequences is $\mathbb{Y}=\{1,2\}\times(\{z_{k}:1\le k\le K\}\cup\{z_{end}\})$,
where the first dimension of the consequence returns the player role
that the agent was assigned into, and the second dimension returns
the terminal node reached. Let $F^{\bullet}:\mathbb{A}^{2}\to\Delta(\mathbb{Y})$ be the objective distribution over consequences.

All agents know the game tree (i.e., $F^{\bullet}$), but some might adhere to a model which mistakenly assumes that their opponent plays Drop
with the same probabilities at all of their nodes. Formally, define
the restricted space of strategies $\mathbb{A}^{An}:=\{(d^{k})\in[0,1]^{K}:d^{k}=d^{k'}\text{ if }k\equiv k^{'}\text{(mod 2)}\}\subseteq\mathbb{A}$.
The correctly specified extended model is $\overline{\Theta}^{\bullet}:=\mathbb{A}\times\mathbb{A}\times\{F^{\bullet}\}.$
The misspecified model of interest is $\overline{\Theta}^{An}:=\mathbb{A}^{An}\times\mathbb{A}^{An}\times\{F^{\bullet}\}$,
reflecting a dogmatic belief that opponents play the same mixed action
at all nodes in the analogy class. We emphasize
these restriction on strategies only exists in the subjective beliefs
of the model $\overline{\Theta}^{An}$ adherents. All agents, regardless
of their model, actually have the strategy space $\mathbb{A}$. 

\subsection{Results}

The next proposition provides a justification for why we might expect
agents with coarse analogy classes given by $\mathbb{A}^{An}$ to
persist in the society.
\begin{prop}
\label{prop:abee}Suppose $K\ge4$ and $g>\frac{2}{K-2}\ell$. For
any matching assortativity $\lambda\in[0,1],$ the correctly specified
extended model $\overline{\Theta}^{\bullet}$ is evolutionarily stable
with strategic uncertainty against itself, but it is not evolutionarily
stable with strategic uncertainty against the misspecified extended
model $\overline{\Theta}^{An}.$ Also, $\overline{\Theta}^{An}$
is not evolutionarily stable against $\overline{\Theta}^{\bullet}$,
unless $\lambda=1$.
\end{prop}
In contrast to the results from Section \ref{sec:LQN},
whereby a misspecified inference over $r$ was harmful for $\lambda=1$ if and only if such an inference were helpful for $\lambda=0$, in this environment the correctly specified
extended model is not evolutionarily stable against a coarse reasoner for \emph{any} level of assortativity. Here, the conditional fitness of
$\overline{\Theta}^{An}$ against both $\overline{\Theta}^{\bullet}$
and $\overline{\Theta}^{An}$ can strictly improve on the correctly
specified residents' equilibrium fitness. This is because the matches
between two adherents of $\overline{\Theta}^{\bullet}$ must result
in Dropping at the first move in equilibrium, while matches where
at least one player is an adherent of $\overline{\Theta}^{An}$ either
lead to the same outcome or lead to a Pareto dominating payoff profile
as the misspecified agent misperceives the opponent's continuation
probability and thus chooses Across at almost all of the decision
nodes. 

However, $\overline{\Theta}^{An}$ is not evolutionarily stable against
$\overline{\Theta}^{\bullet}$ either. The correctly specified agents
can exploit the analogy reasoners' mistake and receive higher payoffs
in matches against them than the misspecified
agents receive in matches against each other. Hence, no homogeneous population can be stable, as the resident model would have lower fitness than the
mutant model in equilibrium. Thus we determine stable shares as defined in Section \ref{sect:StableShares}, focusing on the EZ-SU where Across is played as
often as possible.

We take $\lambda=0$ throughout the remainder of this section. Suppose $K\ge4$ and $g>\frac{2}{K-2}\ell$. Consider the\emph{ maximal
continuation EZ-SU}: (1) misspecified agents always play Across except at node $K$ where they choose
Drop, and (2) correctly
specified agents (i) matched with misspecified agents play Drop at nodes $K-1$ and $K$ and Across otherwise, and (ii) matched with correctly specified agents always play Drop. We verify this indeed forms an EZ-SU.

\begin{prop}
\textup{\label{prop:stable_pop_share} }\emph{Suppose $\lambda=0$, $K\ge4$ and
$g>\frac{2}{K-2}\ell$.}\textup{\emph{ }}\textup{The two models
have the same fitness in the maximal continuation EZ-SU of the centipede
game if and only if $p_{B}^{*}=1-\frac{\ell}{g(K-2)}$, and thus $p_{B}^{*}$
is strictly increasing in $g$ and $K$, and strictly decreasing in $\ell.$}
\end{prop}
Intuitively, $p_{B}^{*}$ reflects the fraction of society expected to be analogy reasoners if long run population changes are determined by fitness.  Under the maintained assumption $g>\frac{2}{K-2}\ell,$ the stable
population share of misspecified agents is strictly more than 50\%,
and the share grows with more periods and a larger increase
in payoffs from contintuation. The main intuition is that the
misspecified model has a higher conditional fitness than the rational
model against rational opponents. The former leads to many periods
of continuation and a high payoff for the biased agent when the rational
agent eventually drops, but the latter leads to 0 payoff from immediate
dropping. On the other hand, the misspecified model has a lower conditional
fitness than the rational model against misspecified opponents. For
the two groups to have the same expected fitness, there must be fewer
rational opponents (i.e., a smaller stable population share $p_{A}^{*}$)
when $g$ and $K$ are higher.

Note that, when payoffs are specified as above, two successive periods of continuation lead
to a strict Pareto improvement in payoffs. Consider instead the dollar
game \citep{reny1993common} in Figure \ref{fig:The_dollar_game},
a variant with a more ``competitive'' payoff
structure, where an agent always gets zero when the opponent
plays Drop, at all parts of the game tree. Assume total payoff
increases by 1 in each round. If the first player stops immediately,payoffs are (1, 0), and if the second player continues at the
final node $n^{K}$, payoffs are $(K+2,0).$

\begin{figure}[h]
\begin{centering}
\includegraphics[scale=0.35]{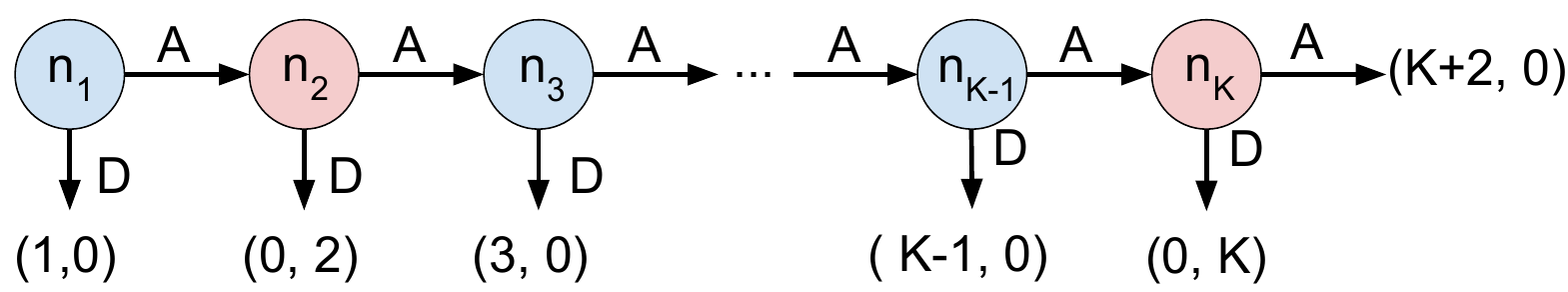}
\par\end{centering}
\vspace{-10bp}

\caption{\label{fig:The_dollar_game}The dollar game. Players 1 (blue) and 2 (red) alternate in choosing Across
(A) or Drop (D). Payoff profiles are shown at the terminal nodes.}
\end{figure}

\begin{prop}
\textup{\label{prop:stable_pop_share_dollar} For $\lambda=0$ and every population
size $(p,1-p)$ with $p\in[0,1],$ }\emph{the maximal continuation
EZ-SU}\textup{ is an EZ-SU where the fitness of $\overline{\Theta}^{\bullet}$
is strictly higher than that of $\overline{\Theta}^{An}$.}
\end{prop}
While maximal continuation remains an EZ-SU, the rational model strictly outperforms the misspecified model for all population shares.
Provided the maximal continuation EZ-SU remains focal, we
should thus expect no analogy reasoners in the long run with this
stage game. Intuitively, the change in the payoffs means one player can only do better \emph{at the expense} of the opponent. Since $\lambda=0$, this implies the less cooperative strategy will be selected. But unlike Section \ref{sec:LQN}, it is the correctly specified model that cannot be exploited.

In a recent survey, \citet{jehiel2020survey} points
out that the misspecified Bayesian learning approach to analogy classes
should aim for ``a better understanding of how the subjective theories
considered by the players may be shaped by the objective characteristics
of the environment.''\footnote{\citet{jehiel2020survey} interprets ABEEs as players adopting the
``simplest'' explanations of observed aggregate statistics of play with coarse feedback.
An objectively coarse feedback structure can lead agents to adopt
the subjective belief that others behave in the same way in all contingencies
in the same coarse analogy class.} 
Taken together, our analysis in this section provides predictions regarding when coarse reasoning should be more prevalent, specifically when the
payoff structure is ``less competitive.'' When this is indeed the case, the bias become more prevalent with
a longer horizon and with faster payoff growth.

\section{\label{subsec:literature}Related Literature}

Our paper contributes to the literature on misspecified Bayesian learning
by proposing a framework to assess which specifications are more likely
to persist based on their objective performance. Most prior work on
misspecified Bayesian learning takes the misspecification as exogenous, studying the subsequent implications in both single-agent decision problems\footnote{See \citet*{nyarko1991learning,fudenberg2017,heidhues2018unrealistic,he_gambler}.}
and multi-agent games.\footnote{See \citet*{bohren2016informational,bohren2017bounded,jehiel2018investment,molavi2019macro,dasaratha2020network,ba2020overconfidence,frick2019misinterpreting,murooka2021multi}.}
A number of papers establish general convergence properties of misspecified
learning.\footnote{See \citet*{esponda2016berk,esponda2019asymptotic,frick2019stability,FLS_general_conv}.} As discussed in the introduction, our work is part of a separate line of research on selecting between multiple specifications for
Bayesian learning, focusing on various
criteria that differ from objective expected payoffs as in our approach.

This paper is closest to two independent and contemporaneous papers,
\citet*{FL_mutation} and \citet*{FII_welfare_based}, who consider
welfare-based criteria for selecting among misspecifications in single-agent
decision problems.\footnote{\citet*{FL_mutation} study a framework where a
continuum of agents with heterogeneous misspecifications arrive each
period and learn from their predecessors' data. \citet*{FII_welfare_based}
assign a \emph{learning efficiency index} to every misspecified signal
structure and conduct a robust comparison of welfare under different
misspecifications.} We differ in highlighting that the learning channel can \emph{strictly} expand the possibility for misspecifications to invade rational societies in strategic settings (relative to biased invaders who do not draw inferences), and we show that misspecifications can lead to different best responses in different environments and thus induce new stability phenomena.

Our framework of competition between different specifications for
Bayesian learning is inspired by the evolutionary game theory literature. Relative to this literature, our contribution is to accommodate misspecified inference. 
We follow past work that also
uses objective payoffs as the selection criterion for subjective preferences in games and decision problems
(e.g., \citet*{Dekeletal2007}, see also the surveys \citet*{robson2011evolutionary}
and \citet*{Alger2019survey}) and the evolution of constrained strategy
spaces \citep*{heller2015three,heller2016rule}. Like us, \citet*{GuthNapel2006} allow for stage-game heterogeneity, studying the ability to discriminate between these games.

When agents entertain
fundamental uncertainty about payoff parameters,  our framework applies evolutionary forces to \emph{sets
of} preferences (i.e., models with multiple possible parameter values).
This allows us to ask our central question: when does the ability to draw inference expand the scope for errors to invade rational societies? Developing a framework which accommodates inference is necessary to answer this question, providing the main point of departure from the literature on the indirect evolutionary approach. Our emphasis on \emph{Bayesian} learning also distinguishes our work from papers that study the evolution of different belief-formation processes \citep{heller2020biased,berman2020naive}, who take a reduced-form (and possibly non-Bayesian) approach
and consider arbitrary inference rules. 

\section{\label{sec:Concluding-Discussion}Concluding Discussion}

We have introduced an evolutionary approach to predict the persistence and emergence of misspecified Bayesian learning. We have emphasized the implications and significance of the learning channel, showing its implications for evolutionary stability and the viability of biases. We showed that the learning channel strictly
expands the possibility for mistakes to invade a rational society, and illustrated how incorporating inference enables the evolutionary approach to speak to new applications and phenomena. 

We acknowledge that our framework does not account for which errors
appear in the first place. It is plausible that some first-stage filter
prevents certain obvious misspecifications from ever reaching the
stage that we study in the evolutionary framework. For this reason, the applications
we focused on reflected misspecifications that seem psychologically plausible.

We have used an otherwise off-the-shelf framework to describe the selection of specifications. The goal of this paper is
not to identify the suitable definition of fitness to justify a particular
error (which is the focus for many of the papers that \citet{robson2011evolutionary}
survey). Rather, our goal has been to determine what evolutionary forces would suggest about the emergence of misspecified learning, and implications thereof. In doing so, we have attempted to describe why it may be important for biases to respond to data, while still departing from rationality in the long run.  

\vspace{-15bp}

\begin{singlespace}
{\small{}{}{} \bibliographystyle{ecta}
\bibliography{misspec_and_welfare}
}{\small\par}
\end{singlespace}

\vspace{-10bp}

\appendix

\begin{center}
\textbf{\Large{}{}{}Appendix}{\Large\par}
\par\end{center}

\vspace{-30bp}

\section{\label{sec:LQNMore} Additional Results for Section \ref{sec:LQN}}

\subsection{Subjective Best Response and Misspecified Inference \label{subsect:LQNEZs}}

In order to determine which models (i.e., perceptions of $\kappa$)
are stable against rival models, we must characterize the relevant
equilibrium zeitgeists. This section develops a number of preliminary
results that relate beliefs about the game parameters to best responses,
and conversely strategy profiles to the KL-divergence minimizing inferences.
The proofs of these results appear in the Online Appendix \ref{sec:omitted_proofs}. 

We begin by proving the result alluded to in Section \ref{sec:LQN}:
every agent's inferences about the state and about opponent's signal
are linear functions of her own signal. The linear coefficient on
the latter increases with the correlation parameter $\kappa$.
\begin{lem}
\label{lem:kappa}There exists a strictly increasing function $\psi(\kappa),$
with $\psi(0)>0$ and $\psi(1)=1,$ so that $\mathbb{E}_{\kappa}[s_{-i}\mid s_{i}]=\psi(\kappa)\cdot s_{i}$
for all $s_{i}\in\mathbb{R},$ $\kappa\in[0,1].$ Also, there exists
a strictly positive $\gamma\in\mathbb{R}$ so that $\mathbb{E}_{\kappa}[\omega\mid s_{i}]=\gamma\cdot s_{i}$
for all $s_{i}\in\mathbb{R}$, $\kappa\in[0,1].$
\end{lem}
Linearity of $\mathbb{E}[\omega\mid s_{i}]$ and $\mathbb{E}[s_{-i}\mid s_{i}]$
in $s_{i}$ allows us explicitly characterize the corresponding linear
best responses, given beliefs about $\kappa$ and elasticity $r$.
For $Q_{i},Q_{-i}$ (not necessarily linear) strategies in the stage
game and $\mu\in\Delta(\Theta(\kappa))$, let $U_{i}(Q_{i},Q_{-i};\mu)$
be $i$'s subjective expected utility from playing $Q_{i}$ against
$Q_{-i},$ under the belief $\mu.$
\begin{lem}
\label{lem:BR}For $\alpha_{-i}$ a linear strategy, $U_{i}(\alpha_{i},\alpha_{-i};\mu)=\mathbb{E}[s_{i}^{2}]\cdot\left(\alpha_{i}\gamma-\frac{1}{2}\hat{r}\alpha_{i}^{2}-\frac{1}{2}\hat{r}\psi(\kappa)\alpha_{i}\alpha_{-i}-\frac{1}{2}\alpha_{i}^{2}\right)$
for every linear strategy $\alpha_{i},$ where $\hat{r}=\int r\ d\mu(r,\kappa,\sigma_{\zeta})$
is the mean of $\mu$'s marginal on elasticity. For $\kappa\in[0,1]$
and $r>0$, $\alpha_{i}^{BR}(\alpha_{-i};\kappa,r):=\frac{\gamma-\frac{1}{2}r\psi(\kappa)\alpha_{-i}}{1+r}$
best responds to $\alpha_{-i}$ among all (possibly non-linear) strategies
$Q_{i}:\mathbb{R}\to\mathbb{R}$ for all $\sigma_{\zeta}>0$.
\end{lem}
Lemma \ref{lem:BR} shows that $\alpha_{i}^{BR}(\alpha_{-i};\kappa,r)$
is not only the best-responding linear strategy when opponent plays
$\alpha_{-i}$ and $i$ believes in correlation parameter $\kappa$
and elasticity $r$, it is also optimal among the class of all strategies
$Q_{i}(s_{i})$ against the same opponent play and under the same
beliefs.

Call a linear strategy more \emph{aggressive} if its coefficient $\alpha_{i}\ge0$
is larger. One implication of Lemma \ref{lem:BR} is that agent $i$'s
subjective best response function becomes more aggressive when $i$
believes in lower $\kappa$ or lower $r$. We have $\frac{\partial\alpha_{i}^{BR}}{\partial\kappa}<0$
because the agent can better capitalize on her private information
about market demand when her rival does not share the same information.
We have $\frac{\partial\alpha_{i}^{BR}}{\partial r}<0$ because the
agent can be more aggressive when facing an inelastic market price.

We now turn to equilibrium inference about the market price elasticity
$r^{\bullet}$. The following lemma shows that any linear strategy
profile generates data whose KL-divergence can be minimized to 0 by
a unique value of $r$. We also characterize how this inference about
elasticity depends on the strategy profile and the agent's belief
about the correlation parameter $\kappa$. As mentioned earlier, we
focus on the case where the bounds on the inferences $r\in[0,\bar{M}_{r}]$,
$\sigma_{\zeta}\in[0,\bar{M}_{\sigma_{\zeta}}]$ are sufficiently
large to ensure that the KL-divergence minimization problem is well-behaved.
\begin{lem}
\label{lem:unique_inference}With high enough price volatility and
large enough strategy space and inference space, for every $\alpha_{i},\alpha_{-i}\in[0,\bar{M}_{\alpha}],$
we have $D_{KL}(F_{r^{\bullet},\kappa^{\bullet},\sigma_{\zeta}^{\bullet}}(\alpha_{i},\alpha_{-i})\parallel F_{\hat{r},\kappa,\hat{\sigma}_{\zeta}}(\alpha_{i},\alpha_{-i}))=0$
for exactly one pair $\hat{r}\in[0,\bar{M}_{r}],\hat{\sigma}_{\zeta}\in[0,\bar{M}_{\sigma_{\zeta}}]$.
This $\hat{r}$ is given by $r_{i}^{INF}(\alpha_{i},\alpha_{-i},;\kappa^{\bullet},\kappa,r^{\bullet}):=r^{\bullet}\frac{\alpha_{i}+\alpha_{-i}\psi(\kappa^{\bullet})}{\alpha_{i}+\alpha_{-i}\psi(\kappa)}$.
\end{lem}
Lemma \ref{lem:unique_inference} implies that an agent's inference
about $r$ is strictly decreasing in her belief about the correlation
parameter $\kappa.$ To understand why, assume player $i$ uses the
linear strategy $\alpha_{i}$ and player $-i$ uses the linear strategy
$\alpha_{-i}$. After receiving a private signal $s_{i}$, player
$i$ expects to face a price distribution with a mean of $\gamma s_{i}-r(\frac{1}{2}\alpha_{i}s_{i}+\frac{1}{2}\alpha_{-i}\mathbb{E}_{\kappa}[s_{-i}\mid s_{i}]).$
Under projection bias $\kappa>\kappa^{\bullet},$ $\mathbb{E}_{\kappa}[s_{-i}\mid s_{i}]$
is excessively steep in $s_{i}$. For example, following a large and
positive $s_{i},$ the agent overestimates the similarity of $-i$'s
signal and wrongly predicts that $-i$ must also choose a very high
quantity, and thus becomes surprised when market price remains high.
The agent then wrongly infers that the market price elasticity must
be low. Therefore, in order to rationalize the average market price
conditional on own signal, an agent with projection bias must infer
$r<r^{\bullet}$. For similar reasons, an agent with correlation neglect
infers $r>r^{\bullet}.$

Combining Lemma \ref{lem:BR} and Lemma \ref{lem:unique_inference},
we find that increasing $\kappa$ has an \emph{a priori} ambiguous
impact on the agent's equilibrium aggressiveness. Increasing $\kappa$
has the direct effect of lowering aggression (by Lemma \ref{lem:BR}),
but it also causes the indirect effect of lowering inference about
$r$ (by Lemma \ref{lem:unique_inference}) and therefore increases
aggression (by Lemma \ref{lem:BR}). 

Lemma \ref{lem:unique_inference} considers the problem of KL-divergence
minimization when all of the data are generated from a single strategy
profile, $(\alpha_{-i},\alpha_{-i}).$ It implies that if $\lambda\in\{0,1\}$
and $(p_{A},p_{B})=(1,0)$, that is matching is either perfectly uniform
or perfectly assortative in a homogeneous society, then every agent
can find a parameter to exactly fit her equilibrium data. This is because
agents only match with opponents from one group in the EZ. The self-confirming
property lends a great deal of tractability and allows us to provide
sharp comparative statics and assess the stability of models.

With interior population shares, agents can observe consequences from
matches against the adherents of both $\Theta_{A}$ and $\Theta_{B}.$
Thus, they must find a single set of parameters for the stage game
that best fits all of their data, and even this best-fitting parameter
will have positive KL divergence in equilibrium. The next lemma shows
the LQN game satisfies the sufficient conditions from Online Appendix
\ref{sec:Existence-and-Continuity} (Assumptions \ref{assu:compact}
through \ref{assu:quasiconcave}) for the existence and upper hemicontinuity
of EZs. So, the tractable analysis in homogeneous societies remains
robust to the introduction of a small but non-zero share of a mutant
model.
\begin{lem}
\label{lem:LQN_conditions_1_to_5}For every $r^{\bullet},\sigma_{\zeta}^{\bullet}\ge0$,
$\lambda\in[0,1],$ $\kappa^{\bullet},\kappa\in[0,1],$ $\bar{M}_{\alpha},\bar{M}_{\sigma_{\zeta}},\bar{M}_{r}<\infty$,
the LQN with objective parameters $(r^{\bullet},\kappa^{\bullet},\sigma_{\zeta}^{\bullet})$,
strategy space $\mathbb{A}=[0,\bar{M}_{\alpha}]$ and models $\Theta(\kappa^{\bullet}),\Theta(\kappa)$
with parameter spaces $[0,\bar{M}_{r}]$, $[0,\bar{M}_{\sigma_{\zeta}}]$
satisfy Assumptions \ref{assu:compact}, \ref{assu:continuous_utility},
\ref{assu:finite_KL}, \ref{assu:continuous_KL}, and \ref{assu:quasiconcave}.
Therefore, EZs in LQN are upper hemicontinuous in population sizes.
\end{lem}

\subsection{More General LQN Games \label{subsec:LQNGeneral}}

We turn to general incomplete-information games and provide a condition
for a model to be evolutionarily fragile against a ``nearby'' misspecified
model. This condition shows how assortativity and the learning channel
shape the evolutionary selection of models for a broader class of
stage games and biases. We also relate the condition to the specific
results studied so far in this application.

Consider a stage game where a state of the world $\omega$ is realized
at the start of the game. Players 1 and 2 observe private signals
$s_{1},s_{2}\in S\subseteq\mathbb{R}$, possibly correlated given
$\omega.$ The objective distribution of $(\omega,s_{1},s_{2})$ is
$\mathbb{P}^{\bullet}$. Based on their signals, players choose actions
$q_{1},q_{2}\in\mathbb{R}$ and receive random consequences $y_{1},y_{2}\in\mathbb{Y}.$
The distribution over consequences as a function of $(\omega,s_{1},s_{2},q_{1},q_{2})$
and the utility over consequences $\pi:\mathbb{Y}\to\mathbb{R}$ are
such that each player $i$'s objective expected utility from taking
action $q_{i}$ against opponent action $q_{-i}$ in state $\omega$
is given by $u_{i}^{\bullet}(q_{i},q_{-i};\omega)$, differentiable
in its first two arguments.

For an interval of real numbers $[\text{\ensuremath{\underbar{\ensuremath{\kappa}}}},\bar{\kappa}]$
with $\text{\ensuremath{\underbar{\ensuremath{\kappa}}}}<\bar{\kappa}$
and $\kappa^{\bullet}\in(\text{\ensuremath{\underbar{\ensuremath{\kappa}}}},\bar{\kappa})$,
suppose there is a family of models $(\Theta(\kappa))_{\kappa\in[\text{\ensuremath{\underbar{\ensuremath{\kappa}}}},\bar{\kappa}]}$.
Fix $\lambda\in[0,1]$ and a strategy space $\mathbb{A}\subseteq\mathbb{R}^{S}$,
representing the feasible signal-contingent strategies. Suppose the
two models in the society are $\Theta_{A}=\Theta(\kappa^{\bullet})$
and $\Theta_{B}=\Theta(\kappa)$ for some $\kappa\in[\text{\ensuremath{\underbar{\ensuremath{\kappa}}}},\bar{\kappa}].$
The next assumption requires there to be a unique EZ with $(p_{A},p_{B})=(1,0)$
in such societies with any $\kappa\in[\text{\ensuremath{\underbar{\ensuremath{\kappa}}}},\bar{\kappa}]$,
and further requires the EZ to feature linear equilibria. Linear equilibria
exist and are unique in a large class of games outside of the duopoly
framework, and in particular in LQN games under some conditions on
the payoff functions (see, e.g., \citet{angeletos2007efficient}).
\begin{assumption}
\label{assu:linear_interim_eqm} Suppose there is a unique EZ under
$\lambda$-matching and population proportions $(p_{A},p_{B})=(1,0)$
with $\Theta_{A}=\Theta(\kappa^{\bullet})$, $\Theta_{B}=\Theta(\kappa)$
for every $\kappa\in[\text{\ensuremath{\underbar{\ensuremath{\kappa}}}},\bar{\kappa}].$
Suppose the $\kappa$-indexed EZ strategy profiles $(\sigma(\kappa))=(\sigma_{AA}(\kappa),\sigma_{AB}(\kappa),\sigma_{BA}(\kappa),\sigma_{BB}(\kappa))$
are linear, i.e., $\sigma_{gg^{'}}(\kappa)(s_{i})=\alpha_{gg^{'}}(\kappa)\cdot s_{i}$
with $\alpha_{gg^{'}}(\kappa)$ differentiable in $\kappa$. Suppose
that in the EZ with $\kappa=\kappa^{\bullet},$ $\alpha_{AA}(\kappa^{\bullet})$
is objectively interim-optimal against itself.\footnote{More precisely, for every $s_{i}\in S,$ $\alpha_{AA}(\kappa^{\bullet})\cdot s_{i}$
maximizes the agent's objective expected utility across all of $\mathbb{R}$
when $-i$ uses the same linear strategy $\alpha_{AA}(\kappa^{\bullet})$.} Finally, assume for every $\kappa$, Assumptions \ref{assu:compact},
\ref{assu:continuous_utility}, \ref{assu:finite_KL}, \ref{assu:continuous_KL},
and \ref{assu:quasiconcave} are satisfied.
\end{assumption}
\begin{prop}
\label{prop:general_incomplete_info_game}Let $\alpha^{\bullet}:=\alpha_{AA}(\kappa^{\bullet}).$
Then, under Assumption \ref{assu:linear_interim_eqm}, if

\[
\mathbb{E}^{\bullet}\left[\mathbb{E}^{\bullet}\left[\frac{\partial u_{1}^{\bullet}}{\partial q_{2}}(\alpha^{\bullet}s_{1},\alpha^{\bullet}s_{2},\omega)\cdot[(1-\lambda)\alpha_{AB}^{'}(\kappa^{\bullet})+\lambda\alpha_{BB}^{'}(\kappa{}^{\bullet})]\cdot s_{2}\mid s_{1}\right]\right]>0,
\]
then there exists some $\epsilon>0$ so that $\Theta(\kappa^{\bullet})$
is evolutionarily fragile against models $\Theta(\kappa)$ with
$\kappa\in(\kappa^{\bullet},\kappa^{\bullet}+\epsilon]\cap[\text{\ensuremath{\underbar{\ensuremath{\kappa}}}},\bar{\kappa}]$.
Also, if

\[
\mathbb{E}^{\bullet}\left[\mathbb{E}^{\bullet}\left[\frac{\partial u_{1}^{\bullet}}{\partial q_{2}}(\alpha^{\bullet}s_{1},\alpha^{\bullet}s_{2},\omega)\cdot[(1-\lambda)\alpha_{AB}^{'}(\kappa^{\bullet})+\lambda\alpha_{BB}^{'}(\kappa{}^{\bullet})]\cdot s_{2}\mid s_{1}\right]\right]<0,
\]
then there exists some $\epsilon>0$ so that $\Theta(\kappa^{\bullet})$
is evolutionarily fragile against models $\Theta(\kappa)$ with
$\kappa\in[\kappa^{\bullet}-\epsilon,\kappa^{\bullet})\cap[\text{\ensuremath{\underbar{\ensuremath{\kappa}}}},\bar{\kappa}]$.
Here $\mathbb{E}^{\bullet}$ is the expectation with respect to the
objective distribution of $(\omega,s_{1},s_{2})$ under $\mathbb{P}^{\bullet}$.
\end{prop}
Proposition \ref{prop:general_incomplete_info_game} describes a general
condition to determine whether a correctly specified model is evolutionarily
fragile against a nearby misspecified mutant model. The condition
asks if a slight change in the mutant model's $\kappa$ leads mutants'
opponents to change their equilibrium actions such that the mutants
become better off on average. These opponents are the residents under
uniform matching $\lambda=0$, so $\alpha_{AB}^{'}(\kappa^{\bullet})$
is relevant. These opponents are other mutants under perfectly assortative
matching $\lambda=1$, so $\alpha_{BB}^{'}(\kappa{}^{\bullet})$ is
relevant.

Proposition \ref{prop:general_incomplete_info_game} implies that
one should only expect the correctly specified model to be stable
against all nearby models in ``special'' cases --- that is, when
the expectation in the statement of Proposition \ref{prop:general_incomplete_info_game}
is exactly equal to 0. One such special case is when the agents face
a decision problem where 2's action does not affect 1's payoffs, that
is $\frac{\partial u_{1}^{\bullet}}{\partial q_{2}}=0$. This sets
the expectation to zero, so the result never implies that the correctly
specified model is evolutionarily fragile against a misspecified
model in such decision problems.

In the duopoly game analyzed previously, we have $\frac{\partial u_{1}^{\bullet}}{\partial q_{2}}(q_{1},q_{2},\omega)=-\frac{1}{2}r^{\bullet}q_{1}$.
Player 1 is harmed by player 2 producing more if $q_{1}>0,$ and helped
if $q_{1}<0.$ From straightforward algebra, the expectation in Proposition
\ref{prop:general_incomplete_info_game} simplifies to 
\[
\mathbb{E}^{\bullet}[s_{1}^{2}]\cdot(-\frac{1}{2}\psi(\kappa^{\bullet})r^{\bullet}\alpha^{\bullet})\cdot[(1-\lambda)\alpha_{AB}^{'}(\kappa^{\bullet})+\lambda\alpha_{BB}^{'}(\kappa{}^{\bullet})].
\]

The proof of Proposition \ref{prop:LQN_lambda_0} shows that when
$\lambda=0$, $\alpha_{AB}^{'}(\kappa^{\bullet})<0$. The proof of
Proposition \ref{prop:LQN_lambda_1} shows that when $\lambda=1$,
$\alpha_{BB}^{'}(\kappa^{\bullet})>0$. The uniqueness of EZ also
follow from these results, for an open interval of $\kappa$ containing
$\kappa^{\bullet}$. We restrict $\mathbb{A}$ to the set of linear
strategies, and Lemma \ref{lem:BR} implies linear strategies played
by two correctly specified firms against each other are interim optimal.
Finally, Lemma \ref{lem:LQN_conditions_1_to_5} verifies that Assumptions
\ref{assu:compact} through \ref{assu:quasiconcave} are satisfied.
So, the conditions of Proposition \ref{prop:general_incomplete_info_game}
hold for $\lambda\in\{0,1\}$, and we deduce the correctly specified
model is evolutionarily fragile against slightly higher $\kappa$
(for $\lambda=0)$ and slightly lower $\kappa$ (for $\lambda=1$).

\section{\label{sec:Proofs}Proofs of Key Results from the Main Text}

\global\long\def\thedefn{A.\arabic{defn}}%
\global\long\def\thecor{A.\arabic{cor}}%
\global\long\def\theprop{A.\arabic{prop}}%
\global\long\def\thelem{A.\arabic{lem}}%
\global\long\def\theclaim{A.\arabic{claim}}%
\global\long\def\theassumption{A.\arabic{assumption}}%
 \setcounter{prop}{0} \setcounter{lem}{0} \setcounter{defn}{0}
\setcounter{assumption}{0}

\subsection{Proof of Theorem \ref{thm:theoryneeded}}

\textit{Part 1:} Let $\mathcal{V}$ be the convex hull of $\{(v_{G}^{b})_{G\in\mathcal{G}}\mid b:\mathbb{A}\rightrightarrows\mathbb{A}\}$,
and let $\mathcal{U}=\{(u_{G})_{G\in\mathcal{G}}:u_{G}\le v_{G}\text{ for all }G\text{ for some }v\in\mathcal{V}\}.$
Note $\mathcal{U}$ is closed and convex (since $\mathcal{V}$ is
convex). By hypothesis, $v^{\text{NE}}$ is not in the interior or
on the boundary of $\mathcal{U}.$ So by the separating hyperplane
theorem, there exists a vector $q\in\mathbb{R}^{|\mathcal{G}|}$ with
$q_{G}\ne0$ for every $G,$ so that $q\cdot v^{\text{NE}}>q\cdot u$
for every $u\in\mathcal{U}.$ Furthermore, $q_{G}\ge0$ for every
$G.$ This is because if $q_{G'}<0$ for some $G',$ then since $\mathcal{U}$
contains vectors with arbitrarily negative values in the $G'$ dimension,
we cannot have $q\cdot v^{\text{NE}}\ge q\cdot u$ for every $u\in\mathcal{U}.$
We may then without loss view $q$ as a distribution on $\mathcal{G}$. In fact, we can take $q$ to be full support.  To see this,  note that since $|\mathcal{G}| < \infty$ and $\mathcal{U}$ is convex, we have
\begin{equation*} 
\lim_{\varepsilon \rightarrow 0} \max_{v \in \mathcal{U}} \left[(1- \varepsilon)q + \frac{\varepsilon}{|\mathcal{G}|}(1,1, \ldots,1) \right] \cdot v =  \max_{v \in \mathcal{U}} q \cdot v,
\end{equation*} 

\noindent by continuity of the support function of convex sets in $\mathbb{R}^{n}$ (given that the support function on $\mathcal{U}$ is bounded for all $q \geq 0$, since $v_{G}^{b}$ is bounded above for every $b$ and every $G$). Thus, setting $\tilde{q}(\varepsilon) = (1- \varepsilon)q + \frac{\varepsilon}{|\mathcal{G}|}(1,1,\ldots, 1)$, we have $\tilde{q}(\varepsilon)$ is a full support distribution with $\tilde{q}(\varepsilon) \cdot v^{NE} > \tilde{q}(\varepsilon) \cdot u$ whenever $\varepsilon$ is sufficiently small, since we have that this inequality holds in the limit.

Now consider any singleton model $\Theta=\{F\}$, and let $b:\mathbb{A}\rightrightarrows\mathbb{A}$
be the subjective best-response correspondence that $F$ induces. If
$v_{G}^{b}\ne-\infty$ for every $G$, then, for each $G$ we can
find a strategy profile $(a_{i}^{G},a_{-i}^{G})$ where $a_{i}^{G}\in b(a_{-i}^{G}),$
$a_{-i}^{G}$ is a rational best response to $a_{i}^{G}$ in situation
$G$, and the strategy pair gives utility $v_{G}^{b}$ to the first
player. There is an  EZ where the resident correctly specified agents
get $v_{G}^{\text{NE}}$ in situation $G$, and the mutants with model
$\Theta$ play $(a_{i},a_{-i})$ in matches against the residents
and get utility $v_{G}^{b}$ in the same situation. Under the distribution
of situations $q$, the residents' fitness is $q\cdot v^{\text{NE}}$
while that of the mutants is $q\cdot v^{b}$, and the former is weakly
larger by construction of $q$ since $v^{b}\in\mathcal{U}$. This
EZ shows the correctly specified model is not evolutionarily fragile
against $\{F\}.$ Otherwise, if we have that $v_{G}^{b}=-\infty$
for some $G,$ then there are no  EZs, so the correctly specified
model is not evolutionarily fragile against $\{F\}$ by the emptiness
of the set of  EZs.

\textit{Part 2:} Suppose the hypotheses hold and let us construct
the misspecified model $\hat{\Theta}=\{F_{G}:G\in\mathcal{G}\}.$
To define the parameters $F_{G},$ first consider $\tilde{F}_{G}$ where
$\tilde{F}_{G}(a_{i},a_{-i}):=F^{\bullet}(a_{i},\underline{\text{BR}}(a_{i},G),G)$
for every $a_{-i}\in\mathbb{A}$. Now for each $(a_{i},a_{-i},G)\in\mathbb{A}\times\mathbb{A}\times\mathcal{G}$,
define the distribution $F_{G}(a_{i},a_{-i})\in\Delta(\mathbb{Y})$
as a sufficiently small perturbation of the $\tilde{F}_{G}(a_{i},a_{-i})$,
such that for every $a_{i},a_{-i}\in\mathbb{A}$ and every $G\in\mathcal{G}$,
$\min_{\hat{G}\in\mathcal{G}}KL(F^{\bullet}(a_{i},a_{-i},G)\parallel F_{\hat{G}}(a_{i},a_{-i}))$
has a unique solution. This can be done because there are finitely
many strategies and situations.

Consider any  EZ $\mathfrak{Z}$ with the correctly specified resident,
$\hat{\Theta}$ as the mutant, $\lambda=0$. By situation identifiability,
in $\mathfrak{Z}$ the correctly specified residents must believe
in the true $F^{\bullet}(\cdot,\cdot,G)$ in every situation $G$.
The mutants cannot hold a mixed belief in any situation $G$, by the
construction of the parameters in $\hat{\Theta}$ to rule out ties in
KL divergence. We show further that mutants must believe in $F_{G}$
in situation $G$. This is because if they instead believed in $F_{G'}$
for some $G'\ne G$, then they must play $\bar{a}_{G'}$ as the Stackelberg
strategy is assumed to be unique. Let $a_{-i}$ be the rational best
response to $\bar{a}_{G'}$ in situation $G$ and $a_{-i}'$ be the
rational best response to $\bar{a}_{G'}$ in situation $G',$ both
unique by assumption. The mutants' expected distribution of consequences
$F_{G'}(\bar{a}_{G'},a_{-i})$ is a perturbed version of $F^{\bullet}(\bar{a}_{G'},a_{-i}',G')$,
while the true distribution of consequences $F^{\bullet}(\bar{a}_{G'},a_{-i},G)$
is a perturbed version of $F_{G}(\bar{a}_{G'},a_{-i})$. We have $F^{\bullet}(\bar{a}_{G'},a_{-i}',G')\ne F^{\bullet}(\bar{a}_{G'},a_{-i},G)$
by Stackelberg identifiability, so $KL(F^{\bullet}(\bar{a}_{G'},a_{-i},G)\parallel F_{G}(\bar{a}_{G'},a_{-i}))<KL(F^{\bullet}(\bar{a}_{G'},a_{-i},G)\parallel F_{G'}(\bar{a}_{G'},a_{-i}))$
when the perturbations are sufficiently small. This contradicts the
mutants believing in $F_{G'}$ in situation $G$ as the parameter $F_{G}$
generates smaller KL divergence. So the mutants get the Stackelberg
payoff in each situation, which means they have higher fitness than
the residents in every  EZ since $\bar{v}_{G}>v_{G}^{\text{NE}}$
for at least one situation and $q$ has full support. Finally, there
exists at least one  EZ: it is an  EZ for the residents to believe
in $F^{\bullet}(\cdot,\cdot,G)$ in every situation $G$, to play
the symmetric Nash profile that results in $v_{G}^{\text{NE}}$ when
matched with other residents (this profile exists by hypothesis of
the theorem), and for the mutants to believe in $F_{G}$ and play
$(\bar{a}_{G},\underline{\text{BR}}(\bar{a}_{G},G))$ in matches against
residents in situation $G.$

\subsection{Proof of Proposition \ref{prop:no_reversal}}
\begin{proof}
Let two singleton models $\Theta_{A},\Theta_{B}$ be given. By contradiction,
suppose they exhibit stability reversal. Let $\mathfrak{Z}=(\mu_{A},\mu_{B},p=(0,1),\lambda=0,(a))$
be any EZ where $\Theta_{B}$ is resident. By the definition of EZ,
$\mathfrak{Z}^{'}=(\mu_{A},\mu_{B},p=(1,0),\lambda=0,(a))$
is also an EZ where $\Theta_{A}$ is resident. Let $u_{g,g^{'}}$
be model $\Theta_{g}$'s conditional fitness against group $g^{'}$
in the EZ $\mathfrak{Z}^{'}$. Part (i) of the definition of stability
reversal requires that $u_{AA}>u_{BA}$ and $u_{AB}>u_{BB}$. These
conditional fitness levels remain the same in $\mathfrak{Z}$. This
means the fitness of $\Theta_{A}$ is strictly higher than that of
$\Theta_{B}$ in $\mathfrak{Z}$, a contradiction.
\end{proof}

\subsection{Proof of Proposition \ref{prop:reversal_inference_channel}}
\begin{proof}
To show the first claim, by way of contradiction, suppose $\mathfrak{Z}=(\mu_{A},\mu_{B},p=(1,0),\lambda=0,(a_{AA},a_{AB},a_{BA},a_{BB}))$
is an EZ, and $\mathfrak{\tilde{Z}}=(\mu_{A},\mu_{B},p=(0,1),\lambda=0,(\tilde{a}_{AA},\tilde{a}_{AB},\tilde{a}_{BA},\tilde{a}_{BB}))$
is another EZ where the adherents of $\Theta_{B}$ hold the same belief
$\mu_{B}$ (group A's belief cannot change as $\Theta_{A}$ is the
correctly specified singleton model). By the optimality of behavior
in $\mathfrak{Z}$, $a_{BA}$ best responds to $a_{AB}$ under the
belief $\mu_{B}$, and $a_{AB}$ best responds to $a_{BA}$ under
the belief $\mu_{A}$, therefore $\mathfrak{\tilde{Z}}^{'}=(\mu_{A},\mu_{B},p=(0,1),\lambda=0,(\tilde{a}_{AA},a_{AB},a_{BA},\tilde{a}_{BB}))$
is another EZ. This holds because the distributions of observations
for the adherents of $\Theta_{B}$ are identical in $\mathfrak{\tilde{Z}}$
and $\mathfrak{\tilde{Z}}^{'}$, since they only face data generated
from the profile $(\tilde{a}_{BB},\tilde{a}_{BB}).$ At the same time,
since $\tilde{a}_{BB}$ best responds to itself under the belief $\mu_{B},$
we have that $\mathfrak{Z^{'}}=(\mu_{A},\mu_{B},p=(1,0),\lambda=0,(a_{AA},a_{AB},a_{BA},\tilde{a}_{BB}))$
is an EZ. Part (i) of the definition of stability reversal applied
to $\mathfrak{Z^{'}}$ requires that $U^{\bullet}(a_{AB},a_{BA})>U^{\bullet}(\tilde{a}_{BB},\tilde{a}_{BB})$
(where $U^{\bullet}$ is the objective expected payoffs), but part
(ii) of the same definition applied to $\mathfrak{\tilde{Z}}^{'}$
requires $U^{\bullet}(\tilde{a}_{BB},\tilde{a}_{BB})\ge U^{\bullet}(a_{AB},a_{BA}),$
a contradiction.

To show the second claim, by way of contradiction suppose $\Theta_{B}$
is strategically independent and $\mathfrak{Z}=(\mu_{A},\mu_{B},p=(0,1),\lambda=0,(a_{AA},a_{AB},a_{BA},a_{BB}))$
is an EZ. By strategic independence, the adherents of $\Theta_{B}$
find it optimal to play $a_{BB}$ against any opponent strategy under
the belief $\mu_{B}$. So, there exists another EZ of the form $\mathfrak{Z^{'}}=(\mu_{A}^{'},\mu_{B},p=(0,1),\lambda=0,(a_{AA},a_{AB}^{'},a_{BB},a_{BB}))$,
where $a_{AB}^{'}$ is an objective best response to $a_{BB}$. The
belief $\mu_{B}$ is sustained because in both $\mathfrak{Z}$ and
$\mathfrak{Z^{'}}$, the adherents of $\Theta_{B}$ have the same
data: from the strategy profile $(a_{BB},a_{BB}).$ In $\mathfrak{Z^{'}}$,
$\Theta_{A}$ 's fitness is $U^{\bullet}(a_{AB}^{'},a_{BB})$ and
$\Theta_{B}$'s fitness is $U^{\bullet}(a_{BB},a_{BB}).$ We have
$U^{\bullet}(a_{AB}^{'},a_{BB})\ge U^{\bullet}(a_{BB},a_{BB})$ since
$a_{AB}^{'}$ is an objective best response to $a_{BB},$ contradicting
the definition of stability reversal.
\end{proof}

\subsection{Proof of Proposition \ref{prop:no_non_mono_lambda}}
\begin{proof}
Let $\lambda\in[0,1]$ be given and let $\mathfrak{Z}=(\mu_{A},\mu_{B},p=(1,0),\lambda,(a))$
be an EZ. Since $\Theta_{A},\Theta_{B}$ are singleton models, $\mathfrak{Z}_{0}=(\mu_{A},\mu_{B},p=(1,0),\lambda=0,(a))$
and $\mathfrak{Z}_{1}=(\mu_{A},\mu_{B},p=(1,0),\lambda=1,(a))$
are also EZs.  Let $u_{g,g^{'}}$ represent model $\Theta_{g}$'s
conditional fitness against group $g^{'}$ in each of these three
EZs. From the hypothesis of the proposition, $u_{A,A}\ge u_{B,A}$
and $u_{A,A}\ge u_{B,B}$. This means the fitness of $\Theta_{A}$
in $\mathfrak{Z},$ which is $u_{A,A}$, is weakly larger than the
fitness of $\Theta_{B}$ in $\mathfrak{Z},$ which is $\lambda u_{B,B}+(1-\lambda)u_{B,A}$.
This shows $\Theta_{A}$ has weakly higher fitness than $\Theta_{B}$
in every  EZ with $\lambda$ and $p=(1,0)$. Also, at least one such
 EZ exists with assortativity $\lambda$, for at least one  EZ exists
when $\lambda=0$, and the same equilibrium belief and behavior also
constitutes an EZ for any other assortativity.
\end{proof}

\subsection{Details Behind Example \ref{exa:stability_reversal_example}} \label{App:ExDetails}

Let $b^{*}(a_{i},a_{-i})$ solve $\min_{b\in\mathbb{R}}D_{KL}(F^{\bullet}(a_{i},a_{-i})\parallel\hat{F}(a_{i},a_{-i};b,m))),$
where $F^{\bullet}(a_{i},a_{-i})$ is the objective distribution over
observations under the investment profile $(a_{i},a_{-i}),$ and $\hat{F}(a_{i},a_{-i};b,m)$
is the distribution under the same investment profile in the model
where productivity is given by $P=b(x_{i}+x_{-i})-m+\epsilon$. We
find that $b^{*}(a_{i},a_{-i})=b^{\bullet}+\frac{m}{a_{i}+a_{-i}}$.
That is, adherents of $\Theta_{B}$ end up with different beliefs
about the game parameter $b$ depending on the behavior of their typical
opponents, which in turn affects how they respond to different rival
investment levels. Stability reversal happens because when $\Theta_{A}$
is resident and the adherents of $\Theta_{B}$ always meet opponents
who play $a_{i}=1,$ they end up with a more distorted belief about
the fundamental than when $\Theta_{B}$ is resident.

\subsection{Proof of Proposition \ref{prop:LQN_lambda_0}}
\begin{proof}
We can take $L_{1},L_{2},L_{3}$ as given by Lemma \ref{lem:unique_inference}.
Suppose there is an EZ with behavior $\alpha=(\alpha_{AA},\alpha_{AB},\alpha_{BA},\alpha_{BB})$
and beliefs over parameters $\mu_{A}\in\Delta(\Theta(\kappa^{\bullet})),$
$\mu_{B}\in\Delta(\Theta(\kappa)).$ By Lemma \ref{lem:unique_inference},
both $\mu_{A}$ and $\mu_{B}$ must be degenerate beliefs that induce
zero KL divergence, since both groups match up with group A with probability
1. Furthermore, since $\Theta_{A}$ is correctly specified, it is
easy to see that the parameter $F_{r^{\bullet},\kappa^{\bullet},\sigma_{\zeta}^{\bullet}}$
generates 0 KL divergence, hence the belief of the adherents of $\Theta_{A}$
must be degenerate on this correct parameter.

In terms of behavior, from Lemma \ref{lem:BR}, $\alpha_{i}^{BR}(\alpha_{-i};\kappa,r)\le\gamma$
for all $\alpha_{-i}\ge0,\kappa\in[0,1],r\ge0.$ Since the upper bound
$\bar{M}_{\alpha}\ge\gamma$, the adherents of each model must be
best responding (across all linear strategies in $[0,\infty)$) in
all matches, given their beliefs about the environment.

Using the equilibrium belief of group A, we must have $\alpha_{AA}=\alpha_{i}^{BR}(\alpha_{AA};\kappa^{\bullet},r^{\bullet}),$
so $\alpha_{AA}=\frac{\gamma-\frac{1}{2}r^{\bullet}\psi(\kappa^{\bullet})\alpha_{AA}}{1+r^{\bullet}}$.
We find the unique solution $\alpha_{AA}=\frac{\gamma}{1+r^{\bullet}+\frac{1}{2}r^{\bullet}\psi(\kappa^{\bullet})}$.
Next we turn to $\alpha_{AB},\alpha_{BA},$ and $\mu_{B}.$ We know
$\mu_{B}$ puts probability 1 on some $r_{B}$. For adherents of groups
A and B to best respond to each others' play and for group B's inference
to have 0 KL divergence (when paired with an appropriate choice of
$\sigma_{\zeta}$ ), we must have $\alpha_{AB}=\frac{\gamma-\frac{1}{2}r^{\bullet}\psi(\kappa^{\bullet})\alpha_{BA}}{1+r^{\bullet}},$
$\alpha_{BA}=\frac{\gamma-\frac{1}{2}r_{B}\psi(\kappa)\alpha_{AB}}{1+r_{B}}$,
and $r_{B}=r^{\bullet}\frac{\alpha_{BA}+\alpha_{AB}\psi(\kappa^{\bullet})}{\alpha_{BA}+\alpha_{AB}\psi(\kappa)}$
from Lemma \ref{lem:unique_inference}. We may rearrange the expression
for $\alpha_{BA}$ to say $\alpha_{BA}=\gamma-r_{B}\alpha_{BA}-\frac{1}{2}r_{B}\psi(\kappa)\alpha_{AB}.$
Substituting the expression of $r_{B}$ into this expression of $\alpha_{BA},$
we get {\small{}
\begin{align*}
\alpha_{BA} & =\gamma-r_{B}\cdot(\alpha_{BA}+\alpha_{AB}\psi(\kappa)-\frac{1}{2}\alpha_{AB}\psi(\kappa))\\
 & =\gamma-\frac{r^{\bullet}\alpha_{BA}+r^{\bullet}\alpha_{AB}\psi(\kappa^{\bullet})}{\alpha_{BA}+\alpha_{AB}\psi(\kappa)}\cdot(\alpha_{BA}+\alpha_{AB}\psi(\kappa)-\frac{1}{2}\alpha_{AB}\psi(\kappa))\\
 & =\gamma-r^{\bullet}\alpha_{BA}-r^{\bullet}\alpha_{AB}\psi(\kappa^{\bullet})+\frac{1}{2}\psi(\kappa)\alpha_{AB}\frac{r^{\bullet}\alpha_{BA}+r^{\bullet}\alpha_{AB}\psi(\kappa^{\bullet})}{\alpha_{BA}+\alpha_{AB}\psi(\kappa)}
\end{align*}
}Multiply by $\alpha_{BA}+\alpha_{AB}\psi(\kappa)$ on both sides
and collect terms by powers of $\alpha$,{\small{}
\[
(\alpha_{BA})^{2}\cdot\left[-1-r^{\bullet}\right]+\left(\alpha_{BA}\alpha_{AB}\right)\cdot[-\psi(\kappa)-\frac{1}{2}r^{\bullet}\psi(\kappa)-r^{\bullet}\psi(\kappa^{\bullet})]-(\alpha_{AB})^{2}\cdot[\frac{1}{2}r^{\bullet}\psi(\kappa^{\bullet})\psi(\kappa)]+\gamma[\alpha_{BA}+\alpha_{AB}\psi(\kappa)]=0.
\]
}Consider the following quadratic function in $x$, {\small{}
\begin{equation}
H(x):=x^{2}\left[-1-r^{\bullet}\right]+\left(x\cdot\ell(x)\right)\cdot[-\psi(\kappa)-\frac{1}{2}r^{\bullet}\psi(\kappa)-r^{\bullet}\psi(\kappa^{\bullet})]-(\ell(x))^{2}\cdot[\frac{1}{2}r^{\bullet}\psi(\kappa^{\bullet})\psi(\kappa)]+\gamma\left[x+\ell(x)\psi(\kappa)\right]=0,\label{eq:Hx}
\end{equation}
} where $\ell(x):=\frac{\gamma-\frac{1}{2}r^{\bullet}\psi(\kappa^{\bullet})x}{1+r^{\bullet}}$
is a linear function in $x.$ In an EZ, $\alpha_{BA}$ is a root of
$H(x)$ in $[0,\frac{\gamma}{\frac{1}{2}r^{\bullet}\psi(\kappa^{\bullet})}]$.
To see why, if we were to have $\alpha_{BA}>\frac{\gamma}{\frac{1}{2}r^{\bullet}\psi(\kappa^{\bullet})}$,
then $\alpha_{AB}=0.$ In that case, $r_{B}=r^{\bullet}$ and so $\alpha_{BA}=\alpha_{i}^{BR}(0;\kappa^{\bullet},r^{\bullet})=\frac{\gamma}{1+r^{\bullet}}.$
Yet $\frac{\gamma}{1+r^{\bullet}}<\frac{\gamma}{\frac{1}{2}r^{\bullet}\psi(\kappa^{\bullet})}$,
contradiction. Conversely, for any root $x^{*}$ of $H(x)$ in $[0,\frac{\gamma}{\frac{1}{2}r^{\bullet}\psi(\kappa^{\bullet})}]$,
there is an EZ where $\alpha_{BA}=x^{*},$ $\alpha_{AB}=\ell(x^{*})\in[0,\gamma],$
and $r_{B}=r^{\bullet}\frac{\alpha_{BA}+\alpha_{AB}\psi(\kappa^{\bullet})}{\alpha_{BA}+\alpha_{AB}\psi(\kappa)}.$
\begin{claim}
\label{claim:H} There exist some $\underline{\kappa}_{1}<\kappa^{\bullet}<\bar{\kappa}_{1}$
so that $H$ has a unique root in $[0,\frac{\gamma}{\frac{1}{2}r^{\bullet}\psi(\kappa^{\bullet})}]$
for all $\kappa\in[\underline{\kappa}_{1},\bar{\kappa}_{1}]\cap[0,1].$
\end{claim}
By Claim \ref{claim:H} (proved in the Online Appendix), for $\kappa\in[\underline{\kappa}_{1},\bar{\kappa}_{1}]\cap[0,1]$,
group B has only one possible belief about elasticity (denoted by
$r_{B}(\kappa)$) in EZ), since there is only one possible outcome
in the match between group A and group B. This means $\alpha_{BB}$
is also pinned down, since there is only one solution to $\alpha_{BB}=\alpha_{i}^{BR}(\alpha_{BB};\kappa,r_{B}(\kappa))$.
So for every $\kappa\in[\underline{\kappa}_{1},\bar{\kappa}_{1}]\cap[0,1]$,
there is a unique EZ, where equilibrium behavior is given as a function
of $\kappa$ by $\alpha(\kappa)=(\alpha_{AA}(\kappa),\alpha_{AB}(\kappa),\alpha_{BA}(\kappa),\alpha_{BB}(\kappa)).$

Recall from Lemma \ref{lem:BR} that the objective expected utility
from playing $\alpha_{i}$ against an opponent who plays $\alpha_{-i}$
is $U_{i}^{\bullet}(\alpha_{i},\alpha_{-i})=\mathbb{E}[s_{i}^{2}]\cdot(\alpha_{i}\gamma-\frac{1}{2}r^{\bullet}\alpha_{i}^{2}-\frac{1}{2}r^{\bullet}\psi(\kappa^{\bullet})\alpha_{i}\alpha_{-i}-\frac{1}{2}\alpha_{i}^{2})$.
If $-i$ plays the rational best response, then the objective expected
utility of choosing $\alpha_{i}$ is $\bar{U}_{i}(\alpha_{i}):=\mathbb{E}[s_{i}^{2}]\cdot(\alpha_{i}\gamma-\frac{1}{2}r^{\bullet}\alpha_{i}^{2}-\frac{1}{2}r^{\bullet}\psi(\kappa^{\bullet})\alpha_{i}\frac{\gamma-\frac{1}{2}r^{\bullet}\psi(\kappa^{\bullet})\alpha_{i}}{1+r^{\bullet}}-\frac{1}{2}\alpha_{i}^{2})$.
The derivative in $\alpha_{i}$ is $\bar{U}_{i}^{'}(\alpha_{i})=\gamma-r^{\bullet}\alpha_{i}-\frac{1}{2}\frac{r^{\bullet}}{1+r^{\bullet}}\gamma\psi(\kappa^{\bullet})+\frac{1}{2}\frac{(r^{\bullet})^{2}\psi(\kappa^{\bullet})^{2}}{1+r^{\bullet}}\alpha_{i}-\alpha_{i}$.
We also know that $\alpha_{AA}=\frac{\gamma}{1+r^{\bullet}+\frac{1}{2}r^{\bullet}\psi(\kappa^{\bullet})}$
satisfies the first-order condition that $\gamma-r^{\bullet}\alpha_{AA}-\frac{1}{2}r^{\bullet}\psi(\kappa^{\bullet})\alpha_{AA}-\alpha_{AA}=0$,
therefore {\small{}
\begin{align*}
\bar{U}_{i}^{'}(\alpha_{AA}) & =-\frac{1}{2}\frac{r^{\bullet}}{1+r^{\bullet}}\gamma\psi(\kappa^{\bullet})+\frac{1}{2}\frac{(r^{\bullet})^{2}\psi(\kappa^{\bullet})^{2}}{1+r^{\bullet}}\alpha_{AA}+\frac{1}{2}r^{\bullet}\psi(\kappa^{\bullet})\alpha_{AA}\\
 & =\left[\frac{r^{\bullet}\psi(\kappa^{\bullet})}{2}\right]\left(\frac{-\gamma}{1+r^{\bullet}}+\frac{\alpha_{AA}\psi(\kappa^{\bullet})r^{\bullet}}{1+r^{\bullet}}+\alpha_{AA}\right).
\end{align*}
}Making the substitution $\alpha_{AA}=\frac{\gamma}{1+r^{\bullet}+\frac{1}{2}r^{\bullet}\psi(\kappa^{\bullet})}$,
{\small{}
\begin{align*}
\frac{-\gamma}{1+r^{\bullet}}+\frac{\alpha_{AA}\psi(\kappa^{\bullet})r^{\bullet}}{1+r^{\bullet}}+\alpha_{AA} & =\frac{-\gamma(1+r^{\bullet}+\frac{1}{2}\psi(\kappa^{\bullet})r^{\bullet})+\gamma\psi(\kappa^{\bullet})r^{\bullet}+\gamma(1+r^{\bullet})}{(1+r^{\bullet})(1+r^{\bullet}+\frac{1}{2}\psi(\kappa^{\bullet})r^{\bullet})}\\
 & =\frac{\frac{1}{2}\gamma\psi(\kappa^{\bullet})r^{\bullet}}{(1+r^{\bullet})(1+r^{\bullet}+\frac{1}{2}\psi(\kappa^{\bullet})r^{\bullet})}>0.
\end{align*}
} Therefore, if we can show that $\alpha_{BA}^{'}(\kappa^{\bullet})>0,$
then there exists some $\underline{\kappa}_{1}\le\underline{\kappa}<\kappa^{\bullet}<\bar{\kappa}\le\bar{\kappa}_{1}$
so that for every $\kappa\in[\underline{\kappa},\bar{\kappa}]\cap[0,1]$,
$\kappa\ne\kappa^{\bullet}$ adherents of $\Theta_{B}$ have strictly
higher or strictly lower equilibrium fitness in the unique EZ than
adherents of $\Theta_{A}$, depending on the sign of $\kappa-\kappa^{\bullet}$.
Consider again the quadratic function $H(x)$ in Equation (\ref{eq:Hx})
and implicitly characterize the unique root $x$ in $[0,\frac{\gamma}{\frac{1}{2}r^{\bullet}\psi(\kappa^{\bullet})}]$
as a function of $\kappa$ in a neighborhood around $\kappa^{\bullet}$.
Denote this root by $\alpha^{M}$, let $D:=\frac{d\alpha^{M}}{d\psi(\kappa)}$
and also note $\frac{d\ell(\alpha^{M})}{d\psi(\kappa)}=\frac{-r^{\bullet}}{2(1+r^{\bullet})}\psi(k^{\bullet})\cdot D$.
We have {\small{}
\begin{align*}
 & (-1-r^{\bullet})\cdot(2\alpha^{M})\cdot D+(\alpha^{M}\ell(\alpha^{M}))(-1-\frac{1}{2}r^{\bullet})\\
 & +(\ell(\alpha^{M})D+\alpha^{M}\frac{-r^{\bullet}}{2(1+r^{\bullet})}\psi(\kappa^{\bullet})D)\cdot(-\psi(\kappa)-\frac{1}{2}r^{\bullet}\psi(\kappa)-r^{\bullet}\psi(\kappa^{\bullet}))+(\ell(\alpha^{M}))^{2}\cdot(-\frac{1}{2}r^{\bullet}\psi(\kappa^{\bullet}))\\
 & +(2\ell(\alpha^{M})\frac{-r^{\bullet}}{2(1+r^{\bullet})}\psi(\kappa^{\bullet})D)\cdot(-\frac{1}{2}r^{\bullet}\psi(\kappa^{\bullet})\psi(\kappa))+\gamma(D+\ell(\alpha^{M})+\psi(\kappa)\frac{-r^{\bullet}}{2(1+r^{\bullet})}\psi(\kappa^{\bullet})D)=0
\end{align*}
}Evaluate at $\kappa=\kappa^{\bullet},$ noting that $\alpha^{M}(\kappa^{\bullet})=\ell(\alpha^{M}(\kappa^{\bullet}))=x^{*}:=\frac{\gamma}{1+r^{\bullet}+\frac{1}{2}\psi(\kappa^{\bullet})r^{\bullet}}$.
The terms without $D$ are: {\small{}{}{} 
\begin{align*}
(x^{*})^{2}(-1-\frac{1}{2}r^{\bullet})+(x^{*})^{2}(\frac{1}{2}r^{\bullet}\psi(\kappa^{\bullet}))+\gamma x^{*} & =x^{*}\cdot\left[-x^{*}\cdot\left(1+r^{\bullet}+\frac{1}{2}\psi(\kappa^{\bullet})r^{\bullet}-\frac{1}{2}r^{\bullet}\right)+\gamma\right]\\
 & =x^{*}\cdot\left[-\gamma+\frac{1}{2}x^{*}r^{\bullet}+\gamma\right]=\frac{1}{2}r^{\bullet}(x^{*})^{2}>0.
\end{align*}
}The coefficient in front of $D$ is: {\small{}{}{} 
\[
(-1-r^{\bullet})(2x^{*})+(x^{*}+x^{*}\frac{-r^{\bullet}}{2(1+r^{\bullet})}\psi(\kappa^{\bullet}))\cdot(-\psi(\kappa^{\bullet})-\frac{3}{2}r^{\bullet}\psi(\kappa^{\bullet}))+\frac{1}{2}x^{*}\frac{(r^{\bullet})^{2}}{(1+r^{\bullet})}\psi(\kappa^{\bullet})^{3}+\gamma+\gamma\psi(\kappa^{\bullet})^{2}\cdot\frac{-r^{\bullet}}{2(1+r^{\bullet})}.
\]
} Make the substitution $\gamma=x^{*}\cdot\left(1+r^{\bullet}+\frac{1}{2}\psi(\kappa^{\bullet})r^{\bullet}\right)$,
{\small{}{}{} 
\begin{align*}
 & x^{*}\cdot\left\{ -2-2r^{\bullet}+\left(1-\frac{r^{\bullet}}{2(1+r^{\bullet})}\psi(\kappa^{\bullet})\right)\cdot\psi(\kappa^{\bullet})(-\frac{3}{2}r^{\bullet}-1)+\frac{(r^{\bullet})^{2}}{2(1+r^{\bullet})}\psi(\kappa^{\bullet})^{3}\right\} \\
+ & x^{*}\cdot\left\{ \left(1+r^{\bullet}+\frac{1}{2}\psi(\kappa^{\bullet})r^{\bullet}\right)\cdot(1-\psi(\kappa^{\bullet})^{2}\frac{r^{\bullet}}{2(1+r^{\bullet})})\right\} .
\end{align*}
} Collect terms inside the parenthesis based on powers of $\psi(\kappa^{\bullet}),$
we get {\small{}{}{} 
\begin{align*}
 & x^{*}\cdot\left\{ \psi(\kappa^{\bullet})^{3}\frac{(r^{\bullet})^{2}}{2(1+r^{\bullet})}-\frac{\psi(\kappa^{\bullet})^{2}r^{\bullet}}{2(1+r^{\bullet})}(-\frac{3}{2}r^{\bullet}-1)+\psi(\kappa^{\bullet})(-\frac{3}{2}r^{\bullet}-1)-2r^{\bullet}-2\right\} \\
+ & x^{*}\cdot\left\{ -\psi(\kappa^{\bullet})^{3}\frac{(r^{\bullet})^{2}}{4(1+r^{\bullet})}-\frac{\psi(\kappa^{\bullet})^{2}r^{\bullet}}{2(1+r^{\bullet})}\cdot(1+r^{\bullet})+1+r^{\bullet}+\frac{1}{2}\psi(\kappa^{\bullet})r^{\bullet}\right\} .
\end{align*}
}Combine to get: {\small{}{}{}$x^{*}\cdot\left[\psi(\kappa{}^{\bullet})^{3}\frac{(r^{\bullet})^{2}}{4(1+r^{\bullet})}+\frac{\psi(\kappa^{\bullet})^{2}(r^{\bullet})^{2}}{4(1+r^{\bullet})}-\psi(\kappa^{\bullet})r^{\bullet}-\psi(\kappa^{\bullet})-r^{\bullet}-1\right].$}
Here $\psi(\kappa{}^{\bullet})^{3}\frac{(r^{\bullet})^{2}}{4(1+r^{\bullet})}$
and $\frac{\psi(\kappa^{\bullet})^{2}(r^{\bullet})^{2}}{4(1+r^{\bullet})}$
are positive terms with {\small{}{}{}$\psi(\kappa{}^{\bullet})^{3}\frac{(r^{\bullet})^{2}}{4(1+r^{\bullet})}+\frac{\psi(\kappa^{\bullet})^{2}(r^{\bullet})^{2}}{4(1+r^{\bullet})}\le\frac{(r^{\bullet})^{2}}{4(1+r^{\bullet})}+\frac{(r^{\bullet})^{2}}{4(1+r^{\bullet})}\le\frac{1}{2}\cdot r^{\bullet}\cdot\frac{r^{\bullet}}{1+r^{\bullet}}\le\frac{1}{2}r^{\bullet}.$}
Now $-r^{\bullet}+\frac{1}{2}\cdot r^{\bullet}<0$, and also $-\psi(\kappa^{\bullet})r^{\bullet}-\psi(\kappa^{\bullet})-1<0.$
Thus the coefficient in front of $D$ is strictly negative. This shows
$D(\kappa^{\bullet})>0.$ Finally, $\frac{d\alpha^{M}}{d\psi(\kappa)}$
has the same sign as $\frac{d\alpha^{M}}{d\kappa}$ since $\psi(\kappa)$
is strictly increasing in $\kappa.$
\end{proof}

\subsection{Proof of Proposition \ref{prop:LQN_lambda_1}}
\begin{proof}
We will show that in every EZ: (i) for each $g\in\{A,B\},$ $\mu_{g}$
puts probability 1 on $\frac{1+\psi(\kappa^{\bullet})}{1+\psi(\kappa_{g})}r^{\bullet}$;
(ii) for each $g\in\{A,B\}$, $\alpha_{gg}=\frac{\gamma}{1+\frac{r^{\bullet}}{2}(1+\psi(\kappa^{\bullet}))+\frac{r^{\bullet}}{2}(\frac{1+\psi(\kappa^{\bullet})}{1+\psi(\kappa_{g})})}$;
(iii) the equilibrium fitness of group A is weakly higher than that
of group B if and only if $\kappa_{A}\le\kappa_{B}$.

Choose $L_{1},L_{2},L_{3}$ as in Lemma \ref{lem:unique_inference},
given $r^{\bullet}$ and $\bar{M}_{\alpha}.$ In any EZ with behavior
$(\alpha_{AA},\alpha_{AB},\alpha_{BA},\alpha_{BB}),$ since the adherents
of each model matches with their own group with probability 1 under
perfectly assortatively matching, we conclude that each of $\mu_{g}$
for $g\in\{A,B\}$ must put full weight on $r_{i}^{INF}(\alpha_{gg},\alpha_{gg};\kappa^{\bullet},\kappa_{g},r^{\bullet})=\frac{\alpha_{gg}+\alpha_{gg}\psi(\kappa^{\bullet})}{\alpha_{gg}+\alpha_{gg}\psi(\kappa_{g})}r^{\bullet}=\frac{1+\psi(\kappa^{\bullet})}{1+\psi(\kappa_{g})}r^{\bullet}$,
proving (i).

Given this belief, we must have $\alpha_{gg}=\frac{\gamma-\frac{1}{2}\frac{1+\psi(\kappa^{\bullet})}{1+\psi(\kappa_{g})}r^{\bullet}\psi(\kappa_{g})\alpha_{gg}}{1+\frac{1+\psi(\kappa^{\bullet})}{1+\psi(\kappa_{g})}r^{\bullet}}$
by Lemma \ref{lem:BR}. Rearranging yields $\alpha_{gg}=\frac{\gamma}{1+\frac{r^{\bullet}}{2}(1+\psi(\kappa^{\bullet}))+\frac{r^{\bullet}}{2}(\frac{1+\psi(\kappa^{\bullet})}{1+\psi(\kappa)})},$
proving (ii).

From Lemma \ref{lem:BR}, the objective expected utility of each player
when both play the strategy profile $\alpha_{symm}$ is $\mathbb{E}[s_{i}^{2}]\cdot\left(\alpha_{symm}\gamma-\frac{1}{2}r^{\bullet}\alpha_{symm}^{2}-\frac{1}{2}r^{\bullet}\psi(\kappa^{\bullet})\alpha_{symm}^{2}-\frac{1}{2}\alpha_{symm}^{2}\right)$.
This is a strictly concave quadratic function in $\alpha_{symm}$
that is 0 at $\alpha_{symm}=0.$ Therefore, it is strictly decreasing
in $\alpha_{symm}$ for $\alpha_{symm}$ larger than the team solution
$\alpha_{TEAM}$ that maximizes this expression, given by the first-order
condition 
\[
\gamma-r^{\bullet}\alpha_{TEAM}-r^{\bullet}\psi(\kappa^{\bullet})\alpha_{TEAM}-\alpha_{TEAM}=0\Rightarrow\alpha_{TEAM}=\frac{\gamma}{1+r^{\bullet}+r^{\bullet}\psi(\kappa^{\bullet})}.
\]
For any value of $\kappa\in[0,1],$ using the fact that $\psi(0)>0$
and $\psi$ is strictly increasing, 
\[
\frac{\gamma}{1+\frac{r^{\bullet}}{2}(1+\psi(\kappa^{\bullet}))+\frac{r^{\bullet}}{2}(\frac{1+\psi(\kappa^{\bullet})}{1+\psi(\kappa)})}>\frac{\gamma}{1+\frac{r^{\bullet}}{2}(1+\psi(\kappa^{\bullet}))+\frac{r^{\bullet}}{2}(1+\psi(\kappa^{\bullet}))}=\alpha_{TEAM}.
\]
Also, $\frac{\gamma}{1+\frac{r^{\bullet}}{2}(1+\psi(\kappa^{\bullet}))+\frac{r^{\bullet}}{2}(\frac{1+\psi(\kappa^{\bullet})}{1+\psi(\kappa)})}$
is a strictly increasing function in $\kappa$, since $\psi$ is strictly
increasing. We therefore conclude that each player's utility when
they play $\frac{\gamma}{1+\frac{r^{\bullet}}{2}(1+\psi(\kappa^{\bullet}))+\frac{r^{\bullet}}{2}(\frac{1+\psi(\kappa^{\bullet})}{1+\psi(\kappa)})}$
against each other is strictly decreasing in $\kappa,$ proving (iii).
\end{proof}

\subsection{Proof of Proposition \ref{prop:LQNLearningNeeded}}

We consider a distribution $q$ over two situations that have different
true values of $r^{\bullet}$, where $q(r^{\bullet}=0)=1-\varepsilon$
and $q(r^{\bullet}=\bar{r})=\varepsilon$, for some $\overline{r}\ge3$.
Suppose $p=(1,0)$ with the rational model as the resident. We claim
that there are some $0<r_{0}<r_{1}$ such that the following three
conditions hold.
\begin{itemize}
\item For $0\le r<r_{0}$, in every EZ, every singleton model $(r,\kappa)$
obtains negative payoff when $r^{\bullet}=\overline{r}$, and no more
than the rational model's payoff when $r^{\bullet}=0$.
\item For $r_{0}\le r\le r_{1}$, in every EZ, every singleton model $(r,\kappa)$
obtains strictly less than the rational payoff when $r^{\bullet}=0$,
and no more than the Stackelberg payoff against a rational opponent
when $r^{\bullet}=\overline{r}$. Furthermore, the singleton model's
highest EZ payoff when $r^{\bullet}=0$ is given by a continuous function
$\xi(r)$.
\item For $r>r_{1}$, in every EZ, every singleton model $(r,\kappa)$
obtains payoff less than half that of the rational payoff when $r^{\bullet}=0$,
and no more than the Stackelberg payoff against a rational opponent
when $r^{\bullet}=\overline{r}$.
\end{itemize}
We show that if these conditions hold, then the correctly specified
model is evolutionarily stable against any singleton model when
$\varepsilon$ is sufficiently small. Let $c_{0}>0$ be the rational
model's payoff when $r^{\bullet}=0$, let $c_{\overline{r}}>0$ be
the rational model's payoff when $r^{\bullet}=\overline{r}$, and
let $c_{s}>0$ be the Stackelberg payoff against the rational model
when $r^{\bullet}=\overline{r}$. For every $r\in[r_{0},r_{1}]$,
there exists some $\epsilon_{r}>0$ so that $\epsilon_{r}\cdot c_{s}+(1-\epsilon_{r})\cdot\xi(r)=\epsilon_{r}\cdot c_{\overline{r}}+(1-\epsilon_{r})\cdot c_{0}$.
Since $\xi(r)<c_{0}$ for every $r$, we get that if $\varepsilon<\epsilon_{r}$,
then the rational model is evolutionarily stable against the singleton
model with $r$. We have that $\min_{r\in[r_{0},r_{1}]}\epsilon_{r}>0$
since $\xi(r)$ is continuous. Finally, there is some $\epsilon'>0$
so that $\epsilon'\cdot c_{s}+(1-\epsilon')\cdot(c_{0}/2)<\epsilon'\cdot c_{\overline{r}}+(1-\epsilon')\cdot c_{0}$.
Whenever $\varepsilon<\min\{\min_{r\in[r_{0},r_{1}]}\epsilon_{r},\epsilon'\}$,
the rational model is evolutionarily stable against the singleton
model with any $r\ge0$.

As $\varepsilon\rightarrow0$, by linearity of expectations the expected
payoff converges to the payoff when $r^{\bullet}=0$ with probability
1; for any $\hat{r}>0$, a mutant who believes $r=\hat{r}$ obtains
less than the correctly specified resident when $r^{\bullet}=0$.
Thus, a mutant with $r\in(0,\tilde{r}]$ does worse than the correctly
specified resident. On the other hand, while the misspecified resident
may do better than the correctly specified resident when $r>\tilde{r}$,
they do significantly worse when $r=0$, and uniformly so over all
such $r$; as $\varepsilon\rightarrow0$, the benefit vanishes uniformly
and we have that again that the the rational model is stable against
all such $r$.

Recall that Lemma \ref{lem:BR} says the best replies are $\alpha_{i}^{BR}(\alpha_{-i};\kappa,r):=\frac{\gamma-\frac{1}{2}r\psi(\kappa)\alpha_{-i}}{1+r}$.
Suppose $r^{\bullet}=0$. In this case, the rational player chooses
$q(s_{i})=\gamma s_{i}$, and therefore any other $\hat{r}$ chooses
$q(s_{i})=\gamma\left(\frac{1-\frac{1}{2}\hat{r}\psi(\kappa)}{1+\hat{r}}\right)s_{i}$.
The rational player's expected payoff is $\mathbb{E}[\mathbb{E}[\omega q_{i}(s_{i})-\frac{1}{2}q_{i}(s_{i})^{2}\mid s_{i}]]=\mathbb{E}[s_{i}^{2}]\cdot\left(\frac{\gamma^{2}}{2}\right)$;
the mutant playing strategy $q(s_{i})=\alpha_{i}s_{i}$ obtains $\mathbb{E}[s_{i}^{2}](\gamma\alpha_{i}-\frac{1}{2}\alpha_{i}^{2})$,
which is quadratic in $\alpha_{i}$ and maximized at $\alpha_{i}=\gamma$.
Therefore, the correctly specified resident obtains the highest payoff.

If $r^{\bullet}=\overline{r}$, then a mutant who believes $\hat{r}=0$
uses strategy with slope $\alpha_{i}=\gamma$; the mutant obtains
$\mathbb{E}[\mathbb{E}[\omega\alpha_{i}s_{i}-\overline{r}\left(\frac{1}{2}(\alpha_{i}+\alpha_{-i})\right)\alpha_{i}s_{i}^{2}-\frac{\alpha_{i}^{2}s_{i}^{2}}{2}]=\mathbb{E}[s_{i}^{2}]\left(\frac{\gamma^{2}}{2}-\overline{r}\gamma^{2}\left(\frac{1}{2}(1+\frac{1-\frac{1}{2}\overline{r}\psi(\kappa)}{1+\overline{r}})\right)\right)$.
Note that since $\frac{1}{2}\psi(\kappa)$ is bounded away from 1,
$\frac{1-\frac{1}{2}\overline{r}\psi(\kappa)}{1+\overline{r}}$ is
bounded away from 0. Therefore, as long as $\overline{r}\geq1$, we
have that the mutant's payoff will be negative. Since payoffs are
continuous, taking $r_{0}\rightarrow0$, we can find some sufficiently
small $r_{0}$ such that any mutant with $r<r_{0}$ obtains a negative
payoff when $r^{\bullet}=\overline{r}$.

From Lemma \ref{lem:BR}, we know that the rational resident always
chooses the linear strategy with $\alpha_{-i}=\gamma$ when $r^{\bullet}=0$.
Thus, an adherent of the singleton model with $r_{0}\le r\le r_{1}$
chooses the linear coefficient $\frac{\gamma-\frac{1}{2}r\psi(\kappa)\gamma}{1+r}<\frac{\gamma}{1+r}<\gamma$
in every EZ when $r^{\bullet}=0.$ But the game with $r^{\bullet}=0$
has $\alpha_{i}=\gamma$ as the strictly dominant strategy, so the
mutant gets strictly lower payoff than the resident. The mutant's
EZ strategy is a continuous function of $r,$ so their payoff as a
function of $r$ must also be continuous. When $r^{\bullet}=\overline{r}$,
because the resident must best respond to the mutant's strategy in
an EZ, the mutant cannot get more than the Stackelberg payoff.

Find a small enough $x>0$ so that $x\gamma-\frac{1}{2}x^{2}<\frac{1}{4}\gamma^{2}.$
By the same argument as before, an adherent of the singleton model
with $r$ chooses the linear coefficient $\frac{\gamma-\frac{1}{2}r\psi(\kappa)\gamma}{1+r}.$
Set $r_{1}$ so that $\frac{\gamma}{1+r_{1}}=x.$ For any $r\ge r_{1},$
we get the mutant's EZ strategy has a linear coefficient of $\frac{\gamma-\frac{1}{2}r\psi(\kappa)\gamma}{1+r}\le\frac{\gamma}{1+r}\le\frac{\gamma}{1+r_{1}}=x,$
so their payoff is no larger than $x\gamma-\frac{1}{2}x^{2}<\frac{1}{4}\gamma^{4}$.
This is less than half of the payoff of the rational residents, who
choose the linear coefficient $\gamma$ and get $\frac{1}{2}\gamma^{2}.$

\newpage{}
\begin{center}
{\Large{}{}{}Online Appendix for ``Evolutionarily Stable (Mis)specifications:
Theory and Applications''}{\Large\par}
\par\end{center}

\begin{center}
{\large{}{}{}Kevin He and Jonathan Libgober}{\large\par}
\par\end{center}

\global\long\def\thesection{OA \arabic{section}}%
\global\long\def\thedefn{OA\arabic{defn}}%
\global\long\def\thelem{OA\arabic{lem}}%
\global\long\def\theprop{OA\arabic{prop}}%
\global\long\def\thecor{OA\arabic{cor}}%
\global\long\def\theassumption{OA\arabic{assumption}}%
\setcounter{section}{0} \setcounter{prop}{0} \setcounter{lem}{0}
\setcounter{defn}{0} \setcounter{assumption}{0} \setcounter{page}{1}

\section{\label{sec:omitted_proofs}Proofs Omitted from the Appendix}

\subsection{Proof of Example \ref{exa:stability_reversal_example}}
\begin{proof}
Define $b^{*}(a_{i},a_{-i}):=b^{\bullet}+\frac{m}{a_{i}+a_{-i}}$.
It is clear that $D_{KL}(F^{\bullet}(a_{i},a_{-i})\parallel\hat{F}(a_{i},a_{-i};b^{*}(a_{i},a_{-i}),m)))=0$,
while this KL divergence is strictly positive for any other choice
of $b.$

In every EZ with $\lambda=0$ and $p=(1,0),$ we must have $a_{AA}=a_{AB}=1.$
If $a_{BA}=2,$ then the adherents of $\Theta_{B}$ infer $b^{*}(1,2)=b^{\bullet}+\frac{m}{3}$.
With this inference, the biased agents expect $1\cdot(2(b^{\bullet}+\frac{m}{3})-m)=2b^{\bullet}-\frac{m}{3}$
from playing 1 against rival investment 1, and expect $2\cdot(3(b^{\bullet}+\frac{m}{3})-m)-c=6b^{\bullet}-c$
from playing 2 against rival investment 1. Since $4b^{\bullet}+\frac{m}{3}-c>0$
from Condition \ref{cond:large_misspec}, there is an EZ with $a_{BA}=2$
and $\mu_{B}$ puts probability 1 on $b^{\bullet}+\frac{m}{3}$. It
is impossible to have $a_{BA}=1$ in EZ. This is because $b^{*}(1,1)>b^{*}(1,2),$
and under the inference $b^{*}(1,2)$ we already have that the best
response to 1 is 2, so the same also holds under any higher belief
about complementarity. Also, we have $a_{BB}=2$, since 2 must best
respond to both 1 and 2. So in every such EZ, $\Theta_{A}$'s conditional
fitness against group A is $2b^{\bullet}$ and $\Theta_{B}$'s conditional
fitness against group A is $6b^{\bullet}-c$, with $2b^{\bullet}>6b^{\bullet}-c$
by Condition \ref{cond:medium_cost}. Also, $\Theta_{A}$'s conditional
fitness against group B is $3b^{\bullet}$, while $\Theta_{B}$'s
conditional fitness against group B is $8b^{\bullet}-c$. Again, $3b^{\bullet}>8b^{\bullet}-c$
by Condition \ref{cond:medium_cost}.

Next, we show $\Theta_{B}$ has strictly higher fitness than $\Theta_{A}$
in every EZ with $\lambda=0,p_{B}=1.$ There is no EZ with $a_{BB}=1.$
This is because $b^{*}(1,1)=b^{\bullet}+\frac{m}{2}$. As discussed
before, under this inference the best response to 1 is 2, not 1. Now
suppose $a_{BB}=2.$ Then $\mu_{B}$ puts probability 1 on $b^{*}(2,2)=b^{\bullet}+\frac{m}{4}.$
With this inference, the biased agents expect $1\cdot(3(b^{\bullet}+\frac{m}{4})-m)=3b^{\bullet}-\frac{m}{4}$
from playing 1 against rival investment 2, and expect $2\cdot(4(b^{\bullet}+\frac{m}{4})-m)-c=8b^{\bullet}-c$
from playing 2 against rival investment 2. We have $5b^{\bullet}+\frac{m}{4}-c>0$
from Condition \ref{cond:large_misspec}, so 2 best responds to 2.
We must have $a_{AA}=a_{AB}=1.$ We conclude the unique EZ behavior
is $(a_{AA},a_{AB},a_{BA},a_{BB})=(1,1,1,2)$, since the biased agents
expect $1\cdot(2(b^{\bullet}+\frac{m}{4})-m)=2b^{\bullet}-\frac{m}{2}$
from playing 1 against rival investment 1, and expect $2\cdot(3(b^{\bullet}+\frac{m}{4})-m)-c=6b^{\bullet}-\frac{m}{2}-c$
from playing 2 against rival investment 1. We have $4b^{\bullet}-c<0$
from Condition \ref{cond:medium_cost}, so 1 best responds to 1. In
the unique EZ with $\lambda=0$ and $p=(0,1),$ the fitness of $\Theta_{A}$
is $2b^{\bullet}$ and the fitness of $\Theta_{B}$ is $8b^{\bullet}-c,$
where $8b^{\bullet}-c>2b^{\bullet}$ by Condition \ref{cond:medium_cost}.
\end{proof}

\subsection{Proof of Example \ref{exa:non_mono_example}}
\begin{proof}
Let $KL_{4,1}:=0.4\cdot\ln\frac{0.4}{0.1}+0.6\cdot\ln\frac{0.6}{0.9}\approx0.3112,$
$KL_{4,8}:=0.4\cdot\ln\frac{0.4}{0.8}+0.6\cdot\ln\frac{0.6}{0.2}\approx0.3819,$
and $KL_{2,4}:=0.2\cdot\ln\frac{0.2}{0.4}+0.8\cdot\ln\frac{0.8}{0.6}\approx0.0915$.
Let $\lambda_{h}$ be the unique solution to $(1-\lambda)KL_{2,4}-\lambda(KL_{4,8}-KL_{4,1})=0,$
so $\lambda_{h}\approx0.564.$

We show for any $\lambda\in[0,\lambda_{h})$, there exists a unique
EZ $\mathfrak{Z}=(\Theta_{A},\Theta_{B},\mu_{A},\mu_{B},p=(1,0),\lambda,(a))$,
and that this EZ has $\mu_{B}$ putting probability 1 on $F_{H}$,
$a_{AA}=a_{1},$ $a_{AB}=a_{1},$ $a_{BA}=a_{2},$ $a_{BB}=a_{2}$.
First, we may verify that under $F_{H},$ $a_{2}$ best responds to
both $a_{1}$ and $a_{2}.$ Also, the KL divergence of $F_{H}$ is
$\lambda\cdot KL_{4,8}$ while that of $F_{L}$ is $\lambda\cdot KL_{4,1}+(1-\lambda)\cdot KL_{2,4}$.
Since $\lambda<\lambda_{h},$ we see that $F_{H}$ has strictly lower
KL divergence. Finally, to check that there are no other EZs, note
we must have $a_{AA}=a_{1},$ $a_{AB}=a_{1},$ $a_{BA}=a_{2}$ in
every EZ. In an EZ where $a_{BB}$ puts probability $q\in[0,1]$ on
$a_{2},$ the KL divergence of $F_{H}$ is $\lambda p\cdot KL_{4,8}$and
the KL divergence of $F_{L}$ is $\lambda p\cdot KL_{4,1}+(1-\lambda)\cdot KL_{2,4}.$
We have {\small{}{}{} 
\[
\lambda q\cdot KL_{4,1}+(1-\lambda)\cdot KL_{2,4}-\lambda q\cdot KL_{4,8}=\lambda q\cdot(KL_{4,1}-KL_{4.8})+(1-\lambda)KL_{2,4}\ge(1-\lambda)KL_{2,4}-\lambda(KL_{4,8}-KL_{4,1}).
\]
} Since $\lambda<\lambda_{h},$ this is strictly positive. Therefore
we must have $\mu_{B}$ put probability 1 on $F_{H},$ which in turn
implies $q=1.$

When $\Theta_{A}$ is dominant, the equilibrium fitness of $\Theta_{A}$
is always 0.25 for every $\lambda$. The equilibrium fitness of $\Theta_{B}$,
as a function of $\lambda$, is $0.4\lambda+0.2(1-\lambda).$ Let
$\lambda_{l}$ solve $0.25=0.4\lambda+0.2(1-\lambda),$ that is $\lambda_{l}=0.25.$
This shows $\Theta_{A}$ is evolutionarily fragile against $\Theta_{B}$
for $\lambda\in(\lambda_{l},\lambda_{h}),$ and it is evolutionarily
stable against $\Theta_{B}$ for $\lambda=0$.

Now suppose $\lambda=1.$ If there is an EZ with $p_{A}=1$ where
$a_{BB}$ plays $a_{2}$ with positive probability, then $\mu_{B}$
must put probability 1 on $F_{L},$ since $KL_{4,1}<KL_{4,8}.$ This
is a contradiction, since $a_{2}$ does not best respond to itself
under $F_{L}.$ So the unique EZ involves $a_{AA}=a_{1},$ $a_{AB}=a_{1},$
$a_{BA}=a_{2},$ $a_{BB}=a_{3}.$  In the EZ, the fitness of $\Theta_{A}$
is 0.25, and the fitness of $\Theta_{B}$ is 0.2. This shows $\Theta_{A}$
is evolutionarily stable against $\Theta_{B}$ for $\lambda=1.$
\end{proof}

\subsection{Proof of Claim \ref{claim:H}}
\begin{proof}
We show that $H(x)$ (i) has a unique root in $[0,\frac{\gamma}{\frac{1}{2}r^{\bullet}\psi(\kappa^{\bullet})}]$
when $\kappa=\kappa^{\bullet}$; (ii) does not have a root at $x=0$
or $x=\frac{\gamma}{\frac{1}{2}r^{\bullet}\psi(\kappa^{\bullet})}$,
and (iii) the root in the interval is not a double root. By these
three statements, since $H(x)$ is a continuous function of $\kappa,$
there must exist some $\underline{\kappa}_{1}<\kappa^{\bullet}<\bar{\kappa}_{1}$
so that it continues to have a unique root in $[0,\frac{\gamma}{\frac{1}{2}r^{\bullet}\psi(\kappa^{\bullet})}]$
for all $\kappa\in[\underline{\kappa}_{1},\bar{\kappa}_{1}]\cap[0,1].$

Statement (i) has to do with the fact that if $\kappa=\kappa^{\bullet},$
then we need $\alpha_{AB}=\frac{\gamma-\frac{1}{2}r^{\bullet}\psi(\kappa^{\bullet})\alpha_{BA}}{1+r^{\bullet}}$
and $\alpha_{BA}=\frac{\gamma-\frac{1}{2}r^{\bullet}\psi(\kappa^{\bullet})\alpha_{AB}}{1+r^{\bullet}}$.
These are linear best response functions with a slope of $-\frac{1}{2}\frac{r^{\bullet}}{1+r^{\bullet}}\psi(\kappa^{\bullet})$,
which falls in $(-\frac{1}{2},0).$ So there can only be one solution
to $H$ in that region (even when we allow $\alpha_{AB}\ne\alpha_{BA})$,
which is the symmetric equilibrium found before $\alpha_{AB}=\alpha_{BA}=\frac{\gamma}{1+r^{\bullet}+\frac{1}{2}r^{\bullet}\psi(\kappa^{\bullet})}$.

For Statement (ii), we evaluate $H(0)=-(\frac{\gamma}{1+r^{\bullet}})^{2}\frac{1}{2}r^{\bullet}\psi(\kappa^{\bullet})^{2}+\frac{\gamma^{2}\psi(\kappa^{\bullet})}{1+r^{\bullet}}=\frac{\psi(\kappa^{\bullet})\gamma^{2}}{1+r^{\bullet}}(1-\frac{(1/2)r^{\bullet}\psi(\kappa^{\bullet})}{1+r^{\bullet}})\ne0$
because $1+r^{\bullet}>(1/2)r^{\bullet}\psi(\kappa^{\bullet}).$ Finally,
we evaluate $H(\frac{\gamma}{\frac{1}{2}r^{\bullet}\psi(\kappa^{\bullet})})=(\frac{\gamma}{\frac{1}{2}r^{\bullet}\psi(\kappa^{\bullet})})^{2}(-1-r^{\bullet})+\gamma\frac{\gamma}{\frac{1}{2}r^{\bullet}\psi(\kappa^{\bullet})}=\frac{\gamma^{2}}{\frac{1}{2}r^{\bullet}\psi(\kappa^{\bullet})}(1-\frac{1+r^{\bullet}}{\frac{1}{2}r^{\bullet}\psi(\kappa^{\bullet})}).$
This is once again not 0 because $1+r^{\bullet}>(1/2)r^{\bullet}\psi(\kappa^{\bullet}).$

For Statement (iii), we show that $H^{'}(x^{*})<0$ where $x^{*}=\frac{\gamma}{1+r^{\bullet}+\frac{1}{2}r^{\bullet}\psi(\kappa^{\bullet})}.$
We find that 
\begin{align*}
H^{'}(x)= & 2x(-1-r^{\bullet})+\left(\frac{\gamma-r^{\bullet}\psi(\kappa^{\bullet})x}{1+r^{\bullet}}\right)(-\psi(\kappa^{\bullet})-\frac{1}{2}r^{\bullet}\psi(\kappa^{\bullet})-r^{\bullet}\psi(\kappa^{\bullet}))\\
 & -2\left(\frac{\gamma-\frac{1}{2}r^{\bullet}\psi(\kappa^{\bullet})x}{1+r^{\bullet}}\right)\left(\frac{-\frac{1}{2}r^{\bullet}\psi(\kappa^{\bullet})}{1+r^{\bullet}}\right)\left(\frac{1}{2}r^{\bullet}\psi(\kappa^{\bullet})^{2}\right)+\gamma-\frac{\frac{1}{2}r^{\bullet}\psi(\kappa^{\bullet})}{1+r^{\bullet}}\gamma\psi(\kappa^{\bullet}).
\end{align*}
Collecting terms, the coefficient on $x$ is 
\[
-2-2r^{\bullet}+\frac{\psi(\kappa^{\bullet})^{2}r^{\bullet}}{1+r^{\bullet}}\left(\frac{3}{2}r^{\bullet}+1-\frac{1}{4}(\frac{(r^{\bullet})^{2}\psi(\kappa^{\bullet})^{2}}{1+r^{\bullet}})\right),
\]
while the coefficient on the constant is 
\[
\frac{\gamma\psi(\kappa^{\bullet})}{1+r^{\bullet}}\left(-\frac{3}{2}r^{\bullet}-1+\frac{1}{2}\frac{(r^{\bullet})^{2}\psi(\kappa^{\bullet})^{2}}{1+r^{\bullet}}-\frac{1}{2}r^{\bullet}\psi(\kappa^{\bullet})\right)+\gamma.
\]
Therefore, we may calculate $H^{'}(x^{*})\cdot\frac{1}{x^{*}}(1+r^{\bullet})^{2},$
which has the same sign as $H^{'}(x^{*}),$ to be: 
\begin{align*}
 & -(1+r^{\bullet})^{2}(2+2r^{\bullet})+\psi(\kappa^{\bullet})^{2}r^{\bullet}((1+r^{\bullet})(\frac{3}{2}r^{\bullet}+1)-\frac{1}{4}(r^{\bullet})^{2}\psi(\kappa^{\bullet})^{2})\\
 & +(1+r^{\bullet}+\frac{1}{2}r^{\bullet}\psi(\kappa^{\bullet}))\left[\psi(\kappa^{\bullet})((1+r^{\bullet})[-\frac{3}{2}r^{\bullet}-1-\frac{1}{2}r^{\bullet}\psi(\kappa^{\bullet})]+\frac{1}{2}(r^{\bullet})^{2}\psi(\kappa^{\bullet})^{2})+(1+r^{\bullet})^{2}\right].
\end{align*}
We have 
\[
-(1+r^{\bullet})^{2}(2+2r^{\bullet})+(1+r^{\bullet}+\frac{1}{2}r^{\bullet}\psi(\kappa^{\bullet}))(1+r^{\bullet})^{2}\le(1+r^{\bullet})^{2}(-1-\frac{1}{2}r^{\bullet})<0,
\]
since $0\le\psi(\kappa^{\bullet})\le1.$ Also, for the same reason,
\[
(1+r^{\bullet})[-\frac{1}{2}r^{\bullet}\psi(\kappa^{\bullet})]+\frac{1}{2}(r^{\bullet})^{2}\psi(\kappa^{\bullet})^{2}\le-\frac{1}{2}(r^{\bullet})^{2}\psi(\kappa^{\bullet})+\frac{1}{2}(r^{\bullet})^{2}\psi(\kappa^{\bullet})^{2}\le0.
\]
Finally, $\psi(\kappa^{\bullet})^{2}r^{\bullet}(1+r^{\bullet})(\frac{3}{2}r^{\bullet}+1)+(1+r^{\bullet}+\frac{1}{2}r^{\bullet}\psi(\kappa^{\bullet}))\psi(\kappa^{\bullet})(1+r^{\bullet})(-\frac{3}{2}r^{\bullet}-1)$
is no larger than 
\begin{align*}
 & \psi(\kappa^{\bullet})^{2}r^{\bullet}(\frac{3}{2}(r^{\bullet})^{2}+\frac{5}{2}r^{\bullet}+1)+\left[r^{\bullet}\psi(\kappa^{\bullet})r^{\bullet}(-(3/2)r^{\bullet})\right]\\
 & +[r^{\bullet}\psi(\kappa^{\bullet})r^{\bullet}(-1)+1\cdot\psi(\kappa^{\bullet})r^{\bullet}(-(3/2)r^{\bullet})]+[r^{\bullet}\psi(\kappa^{\bullet})\cdot1\cdot(-1)]
\end{align*}
where the negative terms in the first, second, and third square brackets
are respectively larger in absolute value than the first, second and
third parts in the expansion of the first summand. Therefore, we conclude
$H^{'}(x^{*})<0.$
\end{proof}

\subsection{Proof of Lemma \ref{lem:kappa}}
\begin{proof}
For $i\ne j,$ rewrite $s_{i}=\left(\omega+\frac{\kappa}{\sqrt{\kappa^{2}+(1-\kappa)^{2}}}z\right)+\frac{1-\kappa}{\sqrt{\kappa^{2}+(1-\kappa)^{2}}}\eta_{i}$
and $s_{j}=\left(\omega+\frac{\kappa}{\sqrt{\kappa^{2}+(1-\kappa)^{2}}}z\right)+\frac{1-\kappa}{\sqrt{\kappa^{2}+(1-\kappa)^{2}}}\eta_{j}.$
Note that $\omega+\frac{\kappa}{\sqrt{\kappa^{2}+(1-\kappa)^{2}}}z$
has a normal distribution with mean 0 and variance $\sigma_{\omega}^{2}+\frac{\kappa^{2}}{\kappa^{2}+(1-\kappa)^{2}}\sigma_{\epsilon}^{2}$.
The posterior distribution of $\left(\omega+\frac{\kappa}{\sqrt{\kappa^{2}+(1-\kappa)^{2}}}z\right)$
given $s_{i}$ is therefore normal with a mean of $\frac{1/(\frac{(1-\kappa)^{2}}{\kappa^{2}+(1-\kappa)^{2}}\sigma_{\epsilon}^{2})}{1/(\sigma_{\omega}^{2}+\frac{\kappa^{2}}{\kappa^{2}+(1-\kappa)^{2}}\sigma_{\epsilon}^{2})+1/(\frac{(1-\kappa)^{2}}{\kappa^{2}+(1-\kappa)^{2}}\sigma_{\epsilon}^{2})}s_{i}$
and a variance of $\frac{1}{1/(\sigma_{\omega}^{2}+\frac{\kappa^{2}}{\kappa^{2}+(1-\kappa)^{2}}\sigma_{\epsilon}^{2})+1/(\frac{(1-\kappa)^{2}}{\kappa^{2}+(1-\kappa)^{2}}\sigma_{\epsilon}^{2})}.$

Since $\eta_{j}$ is mean-zero and independent of $i$'s signal, the
posterior distribution of $s_{j}\mid s_{i}$ under the correlation
parameter $\kappa$ is normal with a mean of 
\[
\frac{1/(\frac{(1-\kappa)^{2}}{\kappa^{2}+(1-\kappa)^{2}}\sigma_{\epsilon}^{2})}{1/(\sigma_{\omega}^{2}+\frac{\kappa^{2}}{\kappa^{2}+(1-\kappa)^{2}}\sigma_{\epsilon}^{2})+1/(\frac{(1-\kappa)^{2}}{\kappa^{2}+(1-\kappa)^{2}}\sigma_{\epsilon}^{2})}s_{i}
\]
and a variance of $\frac{1}{1/(\sigma_{\omega}^{2}+\frac{\kappa^{2}}{\kappa^{2}+(1-\kappa)^{2}}\sigma_{\epsilon}^{2})+1/(\frac{(1-\kappa)^{2}}{\kappa^{2}+(1-\kappa)^{2}}\sigma_{\epsilon}^{2})}+\frac{(1-\kappa)^{2}}{\kappa^{2}+(1-\kappa)^{2}}\sigma_{\epsilon}^{2}$.
We thus define \\
 ${\normalcolor \psi(\kappa):=\frac{1/(\frac{(1-\kappa)^{2}}{\kappa^{2}+(1-\kappa)^{2}}\sigma_{\epsilon}^{2})}{1/(\sigma_{\omega}^{2}+\frac{\kappa^{2}}{\kappa^{2}+(1-\kappa)^{2}}\sigma_{\epsilon}^{2})+1/(\frac{(1-\kappa)^{2}}{\kappa^{2}+(1-\kappa)^{2}}\sigma_{\epsilon}^{2})}}$
for $\kappa\in[0,1),$ and $\psi(1):=1$. To see that $\psi(\kappa)$
is strictly increasing in $\kappa,$ we have 
\begin{align*}
1/\psi(\kappa) & =1+\frac{\frac{(1-\kappa)^{2}}{\kappa^{2}+(1-\kappa)^{2}}\sigma_{\epsilon}^{2}}{\sigma_{\omega}^{2}+\frac{\kappa^{2}}{\kappa^{2}+(1-\kappa)^{2}}\sigma_{\epsilon}^{2}}\\
 & =1+\frac{(1-\kappa)^{2}\sigma_{\epsilon}^{2}}{(\kappa^{2}+(1-\kappa)^{2})\sigma_{\omega}^{2}+\kappa^{2}\sigma_{\epsilon}^{2}}
\end{align*}
and then we can verify that the second term is decreasing in $\kappa.$

As $\kappa\to1,$ the term $1/(\frac{(1-\kappa)^{2}}{\kappa^{2}+(1-\kappa)^{2}}\sigma_{\epsilon}^{2})$
tends to $\infty,$ so $\frac{1/(\frac{(1-\kappa)^{2}}{\kappa^{2}+(1-\kappa)^{2}}\sigma_{\epsilon}^{2})}{1/(\sigma_{\omega}^{2}+\frac{\kappa^{2}}{\kappa^{2}+(1-\kappa)^{2}}\sigma_{\epsilon}^{2})+1/(\frac{(1-\kappa)^{2}}{\kappa^{2}+(1-\kappa)^{2}}\sigma_{\epsilon}^{2})}$
approaches $\frac{1/(\frac{(1-\kappa)^{2}}{\kappa^{2}+(1-\kappa)^{2}}\sigma_{\epsilon}^{2})}{1/(\frac{(1-\kappa)^{2}}{\kappa^{2}+(1-\kappa)^{2}}\sigma_{\epsilon}^{2})}=1$.
We also verify that $\psi(0)=\frac{1/\sigma_{\epsilon}^{2}}{(1/\sigma_{\omega}^{2})+(1/\sigma_{\epsilon}^{2})}>0.$

Finally, for any $\kappa\in[0,1]$, $\frac{\kappa}{\sqrt{\kappa^{2}+(1-\kappa)^{2}}}z+\frac{1-\kappa}{\sqrt{\kappa^{2}+(1-\kappa)^{2}}}\eta_{i}$
has variance $\sigma_{\epsilon}^{2}$ and mean 0, so $\mathbb{E}_{\kappa}[\omega\mid s_{i}]=\frac{1/\sigma_{\epsilon}^{2}}{1/\sigma_{\epsilon}^{2}+1/\sigma_{\omega}^{2}}s_{i}$.
We then define $\gamma$ as the strictly positive constant $\frac{1/\sigma_{\epsilon}^{2}}{1/\sigma_{\epsilon}^{2}+1/\sigma_{\omega}^{2}}.$
\end{proof}

\subsection{Proof of Lemma \ref{lem:BR}}
\begin{proof}
Player $i$'s conditional expected utility given signal $s_{i}$ is
\[
\alpha_{i}s_{i}\cdot\mathbb{E}_{\kappa}[\mathbb{E}_{r\sim\text{marg}_{r}(\mu)}[\omega-\frac{1}{2}r\alpha_{i}s_{i}-\frac{1}{2}r\alpha_{-i}s_{-i}+\zeta]\mid s_{i}]-\frac{1}{2}(\alpha_{i}s_{i})^{2}
\]
by linearity, expectation over $r$ is equivalent to evaluating the
inner expectation with $r=\hat{r}$, which gives 
\begin{align*}
 & \alpha_{i}s_{i}\cdot\mathbb{E}_{\kappa}[\omega-\frac{1}{2}\hat{r}\alpha_{i}s_{i}-\frac{1}{2}\hat{r}\alpha_{-i}s_{-i}+\zeta|s_{i}]-\frac{1}{2}(\alpha_{i}s_{i})^{2}\\
= & \alpha_{i}s_{i}\cdot(\gamma s_{i}-\frac{1}{2}\hat{r}\alpha_{i}s_{i}-\frac{1}{2}\hat{r}\psi(\kappa)s_{i}\alpha_{-i})-\frac{1}{2}(\alpha_{i}s_{i})^{2}\\
= & s_{i}^{2}\cdot(\alpha_{i}\gamma-\frac{1}{2}\hat{r}\alpha_{i}^{2}-\frac{1}{2}\hat{r}\psi(\kappa)\alpha_{i}\alpha_{-i}-\frac{1}{2}\alpha_{i}^{2}).
\end{align*}
The term in parenthesis does not depend on $s_{i},$ and the second
moment of $s_{i}$ is the same for all values of $\kappa.$ Therefore
this expectation is $\mathbb{E}[s_{i}^{2}]\cdot\left(\alpha_{i}\gamma-\frac{1}{2}\hat{r}\alpha_{i}^{2}-\frac{1}{2}\hat{r}\psi(\kappa)\alpha_{i}\alpha_{-i}-\frac{1}{2}\alpha_{i}^{2}\right).$
The expression for $\alpha_{i}^{BR}(\alpha_{-i};\kappa,r)$ follows
from simple algebra, noting that $\mathbb{E}[s_{i}^{2}]>0$ while
the second derivative with respect to $\alpha_{i}$ for the term in
the parenthesis is $-\frac{1}{2}\hat{r}-\frac{1}{2}<0.$

To see that the said linear strategy is optimal among all strategies,
suppose $i$ instead chooses any $q_{i}$ after $s_{i}.$ By above
arguments, the objective to maximize is 
\[
q_{i}\cdot(\gamma s_{i}-\frac{1}{2}\hat{r}q_{i}-\frac{1}{2}\hat{r}\psi(\kappa)s_{i}\alpha_{-i})-\frac{1}{2}q_{i}^{2}.
\]
This objective is a strictly concave function in $q_{i},$ as $-\frac{1}{2}\hat{r}-\frac{1}{2}<0.$
First-order condition finds the maximizer $q_{i}^{*}=\alpha_{i}^{BR}(\alpha_{-i};\kappa,\hat{r})$.
Therefore, the linear strategy also maximizes interim expected utility
after every signal $s_{i}$, and so it cannot be improved on by any
other strategy.
\end{proof}

\subsection{Proof of Lemma \ref{lem:unique_inference}}
\begin{proof}
Note that $\frac{\alpha_{i}+\alpha_{-i}\psi(\kappa^{\bullet})}{\alpha_{i}+\alpha_{-i}\psi(\kappa)}\ge0$
and $\frac{\alpha_{i}+\alpha_{-i}\psi(\kappa^{\bullet})}{\alpha_{i}+\alpha_{-i}\psi(\kappa)}=1+\frac{\alpha_{-i}(\psi(\kappa^{\bullet})-\psi(\kappa))}{\alpha_{i}+\alpha_{-i}\psi(\kappa)}\le1+\frac{1}{\psi(0)}$
(recalling $\psi(0)>0)$. Hence let $L_{3}=r^{\bullet}\cdot(1+\frac{1}{\psi(0)}).$
When $\bar{M}_{r}\ge L_{3},$ we always have $r_{i}^{INF}(\alpha_{i},\alpha_{-i},;\kappa^{\bullet},\kappa,r^{\bullet})\le\bar{M}_{r}$
for all $\alpha_{i},\alpha_{-i}\ge0$ and $\kappa^{\bullet},\kappa\in[0,1].$

Conditional on the signal $s_{i},$ the distribution of market price
under the model $F_{\hat{r},\kappa,\hat{\sigma}_{\zeta}}$ is normal
with a mean of 
\[
\mathbb{E}[\omega\mid s_{i}]-\frac{1}{2}\hat{r}\alpha_{i}s_{i}-\frac{1}{2}\hat{r}\alpha_{-i}\cdot\mathbb{E}_{\kappa}[s_{-i}\mid s_{i}]=\gamma s_{i}-\frac{1}{2}\hat{r}\alpha_{i}s_{i}-\frac{1}{2}\hat{r}\alpha_{-i}\psi(\kappa)s_{i},
\]
while the distribution of market price under the parameter $F_{r^{\bullet},\kappa^{\bullet},\sigma_{\zeta}^{\bullet}}$
is normal with a mean of 
\[
\mathbb{E}[\omega\mid s_{i}]-\frac{1}{2}r^{\bullet}\alpha_{i}s_{i}-\frac{1}{2}r^{\bullet}\alpha_{-i}\cdot\mathbb{E}_{\kappa^{\bullet}}[s_{-i}\mid s_{i}]=\gamma s_{i}-\frac{1}{2}r^{\bullet}\alpha_{i}s_{i}-\frac{1}{2}r^{\bullet}\alpha_{-i}\psi(\kappa^{\bullet})s_{i}.
\]
Matching coefficients on $s_{i},$ we find that if $\hat{r}=r^{\bullet}\frac{\alpha_{i}+\alpha_{-i}\psi(\kappa^{\bullet})}{\alpha_{i}+\alpha_{-i}\psi(\kappa)}$,
then these means match after every $s_{i}.$ On the other hand, for
any other value of $\hat{r},$ these means will not match for any
$s_{i}$ and thus $D_{KL}(F_{r^{\bullet},\kappa^{\bullet},\sigma_{\zeta}^{\bullet}}(\alpha_{i},\alpha_{-i})\parallel F_{\hat{r},\kappa,\hat{\sigma}_{\zeta}}(\alpha_{i},\alpha_{-i}))>0$
for any $\hat{r}\ne r^{\bullet}\frac{\alpha_{i}+\alpha_{-i}\psi(\kappa^{\bullet})}{\alpha_{i}+\alpha_{-i}\psi(\kappa)}.$

Let $L_{1}=\max_{\kappa\in[0,1]}\left\{ \text{Var}_{\kappa}[\omega\mid s_{i}]+\text{Var}_{\kappa}\left[\frac{1}{2}r^{\bullet}\cdot(1+\frac{1}{\psi(0)})B_{\alpha}\cdot s_{-i}\mid s_{i}\right]\right\} $.
This maximum exists and is finite, since the expression is a continuous
function of $\kappa$ on the compact domain $[0,1].$ Also, let $L_{2}=\max_{\kappa\in[0,1]}\left\{ \text{Var}_{\kappa}[\omega\mid s_{i}]+\text{Var}_{\kappa}\left[\frac{1}{2}r^{\bullet}B_{\alpha}\cdot s_{-i}\mid s_{i}\right]\right\} ,$where
the maximum exists for the same reason. Conditional on the signal
$s_{i},$ the variance of market price under the parameter $F_{r^{\bullet}\frac{\alpha_{i}+\alpha_{-i}\psi(\kappa^{\bullet})}{\alpha_{i}+\alpha_{-i}\psi(\kappa)},\kappa,\hat{\sigma}_{\zeta}}$
is 
\[
\text{Var}_{\kappa}\left[\omega-\frac{1}{2}r^{\bullet}\frac{\alpha_{i}+\alpha_{-i}\psi(\kappa^{\bullet})}{\alpha_{i}+\alpha_{-i}\psi(\kappa)}\alpha_{-i}s_{-i}\mid s_{i}\right]+\hat{\sigma}_{\zeta}^{2}.
\]
Since $\omega$ and $s_{-i}$ are positively correlated given $s_{i},$
and using the fact $r^{\bullet}\frac{\alpha_{i}+\alpha_{-i}\psi(\kappa^{\bullet})}{\alpha_{i}+\alpha_{-i}\psi(\kappa)}\le r^{\bullet}\cdot(1+\frac{1}{\psi(0)})$
and $\alpha_{-i}\le B_{\alpha},$ this variance is no larger than
\[
\text{Var}_{\kappa}\left[\omega\mid s_{i}\right]+\text{Var}_{\kappa}\left[\frac{1}{2}r^{\bullet}\cdot(1+\frac{1}{\psi(0)})B_{\alpha}\cdot s_{-i}\mid s_{i}\right]+\hat{\sigma}_{\zeta}^{2}=L_{1}+\hat{\sigma}_{\zeta}^{2}.
\]
On the other hand, the variance of market price under the parameter $F_{r^{\bullet},\kappa^{\bullet},\sigma_{\zeta}^{\bullet}}$
is 
\[
\text{Var}_{\kappa^{\bullet}}\left[\omega-\frac{1}{2}r^{\bullet}\alpha_{-i}s_{-i}\mid s_{i}\right]+(\sigma_{\zeta}^{\bullet})^{2}\le\text{Var}_{\kappa^{\bullet}}[\omega\mid s_{i}]+\text{Var}_{\kappa^{\bullet}}\left[\frac{1}{2}r^{\bullet}B_{\alpha}\cdot s_{-i}\mid s_{i}\right]+(\sigma_{\zeta}^{\bullet})^{2}\le L_{2}+(\sigma_{\zeta}^{\bullet})^{2}.
\]
At the same time, since $(\sigma_{\zeta}^{\bullet})^{2}\ge L_{1},$
this conditional variance is at least $L_{1}.$ Among values of $\hat{\sigma}_{\zeta}^{2}\in[0,\bar{M}_{\sigma_{\zeta}}^{2}],$
there exists exactly one such that the conditional variance under
$F_{r^{\bullet}\frac{\alpha_{i}+\alpha_{-i}\psi(\kappa^{\bullet})}{\alpha_{i}+\alpha_{-i}\psi(\kappa)},\kappa,\hat{\sigma}_{\zeta}}$
is the same as that under $F_{r^{\bullet},\kappa^{\bullet},\sigma_{\zeta}^{\bullet}}$,
since we have let $\bar{M}_{\sigma_{\zeta}}^{2}\ge(\sigma_{\zeta}^{\bullet})^{2}+L_{2}$.
Thus there is one choice of $\hat{\sigma}_{\zeta}\in[0,\bar{M}_{\sigma_{\zeta}}]$
with such that $D_{KL}(F_{r^{\bullet},\kappa^{\bullet},\sigma_{\zeta}^{\bullet}}(\alpha_{i},\alpha_{-i})\parallel F_{r^{\bullet}\frac{\alpha_{i}+\alpha_{-i}\psi(\kappa^{\bullet})}{\alpha_{i}+\alpha_{-i}\psi(\kappa)},\kappa,\hat{\sigma}_{\zeta}}(\alpha_{i},\alpha_{-i}))=0$.
For any other choice of $\tilde{\sigma}_{\zeta}$, we conclude that
$D_{KL}(F_{r^{\bullet},\kappa^{\bullet},\sigma_{\zeta}^{\bullet}}(\alpha_{i},\alpha_{-i})\parallel F_{r^{\bullet}\frac{\alpha_{i}+\alpha_{-i}\psi(\kappa^{\bullet})}{\alpha_{i}+\alpha_{-i}\psi(\kappa)},\kappa,\tilde{\sigma}_{\zeta}}(\alpha_{i},\alpha_{-i}))>0$.
\end{proof}

\subsection{Proof of Lemma \ref{lem:LQN_conditions_1_to_5}}
\begin{proof}
Assumption \ref{assu:compact} holds as $\mathbb{A}$, $\Theta_{A},\Theta_{B}$
are compact due to the finite bounds $\bar{M}_{\alpha},\bar{M}_{r},\bar{M}_{\sigma_{\zeta}}.$
Also, from Lemma \ref{lem:BR}, the expected utility from playing
$\alpha_{i}$ against $\alpha_{-i}$ in a model with parameters $(\hat{r},\kappa,\sigma_{\zeta})$
is $\mathbb{E}[s_{i}^{2}]\cdot\left(\alpha_{i}\gamma-\frac{1}{2}\hat{r}\alpha_{i}^{2}-\frac{1}{2}\hat{r}\psi(\kappa)\alpha_{i}\alpha_{-i}-\frac{1}{2}\alpha_{i}^{2}\right)$.
This is a continuous function in $(\alpha_{i},\alpha_{-i},\hat{r})$
and strictly concave in $\alpha_{i}.$ Therefore Assumptions \ref{assu:continuous_utility}
and \ref{assu:quasiconcave} are satisfied.

To see the finiteness and continuity of the $K$ functions, first
recall that the KL divergence from a true distribution $\mathcal{N}(\mu_{1},\sigma_{1}^{2})$
to a different distribution $\mathcal{N}(\mu_{2},\sigma_{2}^{2})$
is given by $\ln(\sigma_{2}/\sigma_{1})+\frac{\sigma_{1}^{2}+(\mu_{1}-\mu_{2})^{2}}{2\sigma_{2}^{2}}-\frac{1}{2}$.
Under own play $\alpha_{i},$ opponent play $\alpha_{-i},$ correlation
parameter $\kappa,$ elasticity $\hat{r}$ and price idiosyncratic
variance $\sigma_{\zeta}^{2}$, the expected distribution of price
after signal $s_{i}$ is 
\[
-\frac{1}{2}\hat{r}\alpha_{i}s_{i}+(\omega-\frac{1}{2}\hat{r}\alpha_{-i}s_{-i}\mid s_{i},\kappa)+\hat{\zeta}
\]
where the first term is not random, the middle term is the conditional
distribution of $\omega-\frac{1}{2}\hat{r}\alpha_{-i}s_{-i}$ given
$s_{i}$, based on the joint distribution of $(\omega,s_{i},s_{-i})$
with correlation parameter $\kappa.$ The final term is an independent
random variable with mean 0, variance $\sigma_{\zeta}^{2}.$ The analogous
true distribution of price is 
\[
-\frac{1}{2}r^{\bullet}\alpha_{i}s_{i}+(\omega-\frac{1}{2}r^{\bullet}\alpha_{-i}s_{-i}\mid s_{i},\kappa^{\bullet})+\zeta^{\bullet}
\]
where $\zeta^{\bullet}$ is an independent random variable with mean
0, variance $(\sigma_{\zeta}^{\bullet})^{2}.$ For a fixed $\kappa,$
we may find $0<\underline{\sigma}^{2}<\bar{\sigma}^{2}<\infty$ so
that the variances of both distributions lie in $[\underline{\sigma}^{2},\bar{\sigma}^{2}]$
for all $s_{i}\in\mathbb{R},$ $\alpha_{i},\alpha_{-i}\in[0,\bar{M}_{\alpha}],$
$\hat{r}\in[0,\bar{M}_{r}].$ First note that as a consequence of
the multivariate normality, the variances of these two expressions
do not change with the realization of $s_{i}.$ The lower bound comes
from the fact that $\text{Var}_{\kappa}(\omega-\frac{1}{2}\hat{r}\alpha_{-i}s_{-i}\mid s_{i})$
is nonzero for all $\alpha_{-i},\hat{r}$ in the compact domains and
it is a continuous function of these two arguments, so it must have
some positive lower bound $\underline{\sigma}^{2}>0.$ For a similar
reason, the variance of the middle term has a upper bound for choices
of the parameters $\alpha_{-i},\hat{r}$ in the compact domains, and
the inference about $\sigma_{\zeta}^{2}$ is also bounded.

The difference in the means of the two distributions is no larger
than $s_{i}\cdot[\frac{1}{2}(\bar{M}_{r}+r^{\bullet})\cdot1+\frac{1}{2}(\bar{M}_{r}+r^{\bullet})\cdot1\cdot(\psi(\kappa)+\psi(\kappa^{\bullet}))].$
Thus consider the function 
\[
h(s_{i}):=\ln(\bar{\sigma}/\underline{\sigma})+\frac{1}{2}(\bar{\sigma}^{2}/\underline{\sigma}^{2})+\frac{[\frac{1}{2}(\bar{M}_{r}+r^{\bullet})\cdot1+\frac{1}{2}(\bar{M}_{r}+r^{\bullet})\cdot1\cdot(\psi(\kappa)+\psi(\kappa^{\bullet}))]^{2}}{2\underline{\sigma}^{2}}s_{i}^{2}-\frac{1}{2}.
\]
That is $h(s_{i})$ has the form $h(s_{i})=C_{1}+C_{2}s_{i}^{2}$
for constants $C_{1},C_{2}.$ It is absolutely integrable against
the distribution of $s_{i}$, and it dominates the KL divergence between
the true and expected price distributions at every $s_{i}$ and for
any choices of $\alpha_{i},\alpha_{-i}\in[0,\bar{M}_{\alpha}],\hat{r}\in[0,\bar{M}_{r}],\sigma_{\zeta}^{2}\in[0,\bar{M}_{\zeta}].$
This shows $K_{A},K_{B}$ are finite, so Assumption \ref{assu:finite_KL}
holds. Further, since the KL divergence is a continuous function of
the means and variances of the price distributions, and since these
mean and variance parameters are continuous functions of $\alpha_{i},\alpha_{-i},\hat{r},\sigma_{\zeta}^{2},$
the existence of the absolutely integrable dominating function $h$
also proves $K_{A},K_{B}$ (as integrals of KL divergences across
different $s_{i})$ are continuous, so Assumption \ref{assu:continuous_KL}
holds.
\end{proof}

\subsection{Proof of Proposition \ref{prop:no_learning_channel}}
\begin{proof}
Find $L_{1},L_{2},L_{3}$ as given by Lemma \ref{lem:unique_inference}.
Suppose $\Theta_{A}=\Theta(\kappa^{\bullet})$, $\Theta_{B}=\{F_{r^{\bullet},\kappa,\sigma_{\zeta}^{\bullet}}\}$
for any $\kappa\in[0,1],$ $(p_{A},p_{B})=(1,0)$, and $\lambda\in[0,1],$
then arguments similar to those in the proof of Lemma \ref{lem:unique_inference}
imply there exists exactly one EZ, and it involves the adherents of
$\Theta_{A}$ holding correct beliefs and playing $\frac{\gamma}{1+r^{\bullet}+\frac{1}{2}r^{\bullet}\psi(\kappa^{\bullet})}$
against each other.

We now analyze $\alpha_{BA}(\kappa)$ in such EZ. In the proof of
Proposition \ref{prop:LQN_lambda_0}, we defined $\bar{U}_{i}(\alpha_{i})$
as $i$'s objective expected utility of choosing $\alpha_{i}$ when
$-i$ plays the rational best response. We showed that $\bar{U}_{i}^{'}(\frac{\gamma}{1+r^{\bullet}+\frac{1}{2}r^{\bullet}\psi(\kappa^{\bullet})})>0.$
In an EZ where $i$ believes in the parameter $F_{r^{\bullet},\kappa,\sigma_{\zeta}^{\bullet}}$
and $-i$ believes in the parameter $F_{r^{\bullet},\kappa^{\bullet},\sigma_{\zeta}^{\bullet}},$
using the expression for $\alpha_{i}^{BR}$ from Lemma \ref{lem:BR},
the play of $i$ solves $x=\frac{\gamma-\frac{1}{2}r^{\bullet}\psi(\kappa)\left(\frac{\gamma-\frac{1}{2}r^{\bullet}\psi(\kappa^{\bullet})x}{1+r^{\bullet}}\right)}{1+r^{\bullet}}$,
which implies $\alpha_{BA}(\kappa)=\frac{\gamma(1+r^{\bullet}-\frac{1}{2}\psi(\kappa)r^{\bullet})}{1+2r^{\bullet}+(r^{\bullet})^{2}-\frac{1}{4}\psi(\kappa)\psi(\kappa^{\bullet})(r^{\bullet})^{2}}$.
Taking the derivative and evaluating at $\kappa=\kappa^{\bullet},$
we find an expression with the same sign as $\frac{1}{4}\psi^{'}(\kappa^{\bullet})r^{\bullet}(1+r^{\bullet})\gamma(-2(1+r^{\bullet})+\psi(\kappa^{\bullet})r^{\bullet}),$
which is strictly negative because $\psi^{'}(\kappa^{\bullet})>0,$
$r^{\bullet}>0,$ $\gamma>0,$ and $\psi(\kappa^{\bullet})\le1$.
This shows there exists $\epsilon>0$ so that for every $\kappa_{h}\in(\kappa^{\bullet},\kappa^{\bullet}+\epsilon]$,
we have $\bar{U}_{i}(\alpha_{BA}(\kappa_{h}))<\bar{U}_{i}(\frac{\gamma}{1+r^{\bullet}+\frac{1}{2}r^{\bullet}\psi(\kappa^{\bullet})})$,
that is the adherents of $\{F_{r^{\bullet},\kappa_{h},\sigma_{\zeta}^{\bullet}}\}$
have strictly lower fitness than the adherents of $\Theta(\kappa^{\bullet})$
with $\lambda=0$ in the unique EZ. Finally, existence and upper-hemicontinuity
of EZ in population proportion in such societies can be established
using arguments similar to the proof of Propositions \ref{prop:existence}
and \ref{prop:uhc}. This establishes the first claim to be proved.

Next, we turn to $\alpha_{BB}(\kappa).$ Using the expressing for
$\alpha_{i}^{BR}$ in Lemma \ref{lem:BR}, we find that $\alpha_{BB}(\kappa)=\frac{\gamma}{1+r^{\bullet}+\frac{1}{2}r^{\bullet}\psi(\kappa)}.$
Since $\psi^{'}>0,$ we have $\alpha_{BB}(\kappa)$ is strictly larger
than $\alpha_{AA}=\frac{\gamma}{1+r^{\bullet}+\frac{1}{2}r^{\bullet}\psi(\kappa^{\bullet})}$
when $\kappa<\kappa^{\bullet}.$ From the proof of Proposition \ref{prop:LQN_lambda_1},
we know that objective payoffs in the stage game is strictly decreasing
in linear strategies larger than the team solution $\alpha_{TEAM}=\frac{\gamma}{1+r^{\bullet}+r^{\bullet}\psi(\kappa^{\bullet})}.$
Since $\alpha_{BB}(\kappa)>\alpha_{AA}>\alpha_{TEAM},$ we conclude
the adherents of $\{F_{r^{\bullet},\kappa_{l},\sigma_{\zeta}^{\bullet}}\}$
have strictly lower fitness than the adherents of $\Theta(\kappa^{\bullet})$
with $\lambda=1$ in the unique EZ, for any $\kappa_{l}<\kappa^{\bullet}.$
Again , existence and upper-hemicontinuity of EZ in population proportion
in such societies can be established using arguments similar to the
proof of Propositions \ref{prop:existence} and \ref{prop:uhc}. This
establishes the second claim to be proved.
\end{proof}

\subsection{Proof of Proposition \ref{prop:general_incomplete_info_game}}
\begin{proof}
Consider the society where $\Theta_{A}=\Theta_{B}=\Theta(\kappa^{\bullet})$,
$(p_{A},p_{B})=(1,0).$ For any EZ with behavior $(\sigma_{AA},\sigma_{AB},\sigma_{BA},\sigma_{BB})$
and beliefs $(\mu_{A},\mu_{B})$, there exists another EZ $(\sigma_{AA}^{'},\sigma_{AB}^{'},\sigma_{BA}^{'},\sigma_{BB}^{'})$
where $\sigma_{g,g^{'}}^{'}=\sigma_{AA}$ for all $g,g^{'}\in\{A,B\}$
and all agents hold the belief $\mu_{A}$. The uniqueness of EZ from
Assumption \ref{assu:linear_interim_eqm} implies $\alpha_{AB}(\kappa^{\bullet})=\alpha_{BA}(\kappa^{\bullet})=\alpha_{BB}(\kappa^{\bullet})=\alpha^{\bullet}.$

Now consider the society where $\Theta_{B}=\Theta(\kappa)$, $(p_{A},p_{B})=(1,0).$
By the same arguments as the existence arguments in Proposition \ref{prop:existence},
there exists an EZ where $\alpha_{AA}(\kappa)=\alpha_{AA}(\kappa^{\bullet}).$
By the uniqueness of EZ from Assumption \ref{assu:linear_interim_eqm},
we must in fact have $\alpha_{AA}(\kappa)=\alpha_{AA}(\kappa^{\bullet})$
for all $\kappa$, so the fitness of model $\Theta(\kappa^{\bullet})$
in the unique EZ is 
\[
\mathbb{E}^{\bullet}\left[\mathbb{E}^{\bullet}\left[u_{1}^{\bullet}(\alpha^{\bullet}s_{1},\alpha^{\bullet}s_{2},\omega)\mid s_{1}\right]\right].
\]
Under $\lambda$ matching with mutant model $\Theta(\kappa)$, the
mutant's fitness in the unique EZ is 
\[
\mathbb{E}^{\bullet}\left[\mathbb{E}^{\bullet}\left[(1-\lambda)u_{1}^{\bullet}(\alpha_{BA}(\kappa)s_{1},\alpha_{AB}(\kappa)s_{2},\omega)+(\lambda)u_{1}^{\bullet}(\alpha_{BB}(\kappa)s_{1},\alpha_{BB}(\kappa)s_{2},\omega)\mid s_{1}\right]\right].
\]
Differentiate and evaluate at $\kappa=\kappa^{\bullet}$. At $\kappa=\kappa^{\bullet},$
adherents of $\Theta_{A}$ and $\Theta_{B}$ have the same fitness
since they play the same strategies. So, a non-zero sign on the derivative
would give the desired evolutionary fragility against either models
with slightly higher or slightly lower $\kappa.$ This derivative
is:

\[
\mathbb{E}^{\bullet}\left[\mathbb{E}^{\bullet}\left[\left.\begin{array}{c}
\frac{\partial u_{1}^{\bullet}}{\partial q_{1}}(\alpha^{\bullet}s_{1},\alpha^{\bullet}s_{2},\omega)\cdot[(1-\lambda)\alpha_{BA}^{'}(\kappa^{\bullet})+\lambda\alpha_{BB}^{'}(\kappa^{\bullet})]\cdot s_{1}\\
+\frac{\partial u_{1}^{\bullet}}{\partial q_{2}}(\alpha^{\bullet}s_{1},\alpha^{\bullet}s_{2},\omega)\cdot[(1-\lambda)\alpha_{AB}^{'}(\kappa^{\bullet})+\lambda\alpha_{BB}^{'}(\kappa^{\bullet})]\cdot s_{2}
\end{array}\right|s_{1}\right]\right].
\]
Using the interim optimality part of Assumption \ref{assu:linear_interim_eqm},
$\mathbb{E}^{\bullet}\left[\frac{\partial u_{1}^{\bullet}}{\partial q_{1}}(\alpha^{\bullet}s_{1},\alpha^{\bullet}s_{2},\omega)\mid s_{1}\right]=0$
for every $s_{1}\in S$, using the necessity of the first-order condition.
The derivative thus simplifies as claimed.
\end{proof}

\subsection{Proof of Proposition \ref{prop:abee}}
\begin{proof}
When $\Theta_{A}=\Theta_{B}=\Theta^{\bullet}$, for any matching assortativity
$\lambda$ and with $(p_{A},p_{B})=(1,0),$ we show adherents of both
models have 0 fitness in every  EZ. Suppose instead that the match
between groups $g$ and $g^{'}$ reach a terminal node other than
$z_{1}$ with positive probability. Let $n_{L}$ be the last non-terminal
node reached with positive probability, so we must have $L\ge2$,
and also that nodes $n_{1},...,n_{L-1}$ are also reached with positive
probability. So Drop must be played with probability 1 at $n_{L}.$
Since $n_{L}$ is reached with positive probability, correctly specified
agents hold correct beliefs about opponent's play at $n_{L}$, which
means at $n_{L-1}$ it cannot be optimal to play Across with positive
probability since this results in a loss of $\ell$ compared to playing
Drop, a contradiction.

Now let $\Theta_{A}=\Theta^{\bullet}$, $\Theta_{B}=\Theta^{An}$.
Suppose $\lambda\in[0,1]$ and let $p_{B}\in(0,1).$ We claim there
is an EZ where $d_{AA}^{k}=1$ for every $k$, $d_{AB}^{k}=0$ for
every even $k$ with $k<K$, $d_{AB}^{k}=1$ for every other $k$,
$d_{BA}^{k}=0$ for every odd $k$ and $d_{BA}^{k}=1$ for every even
$k$, and $d_{BB}^{k}=0$ for every $k$ with $k<K,$ $d_{BB}^{K}=1.$
It is easy to see that the behavior $(d_{AA})$ is optimal under correct
belief about opponent's play. In the $\Theta_{A}$ vs. $\Theta_{B}$
matches, the conjecture about A's play $\hat{d}_{AB}^{k}=2/K$ for
$k$ even, $\hat{d}_{AB}^{k}=1$ for $k$ odd minimizes KL divergence
among all strategies in $\mathbb{A}^{An}$, given B's play. To see
this, note that when B has the role of P2, opponent Drops immediately.
When B has the role of P1, the outcome is always $z_{K}.$ So a conjecture
with $\hat{d}_{AB}^{k}=x$ for every even $k$ has the conditional
KL divergence of: 
\begin{align*}
 & \sum_{k\le K-1\text{ odd}}\underset{(1,z_{k})\text{ for }k\le K-1\text{ odd}}{\underbrace{0\cdot\ln\left(\frac{0}{0}\right)}}+\sum_{k\le K-1\text{ even}}\underset{(1,z_{k})\text{ for }k\le K-1\text{ even}}{\underbrace{0\cdot\ln\left(\frac{0}{(1/2)\cdot(1-x)^{(k/2)-1}\cdot x}\right)}}\\
 & +\underset{(1,z_{K})}{\underbrace{\frac{1}{2}\ln\left(\frac{1/2}{(1/2)\cdot(1-x)^{(K/2)-1}\cdot x}\right)}}+\underset{(1,z_{end})}{\underbrace{0\cdot\ln\left(\frac{0}{(1-x)^{(K/2)}}\right)}}
\end{align*}
when matched with an opponent from $\Theta_{A}$. Using $0\cdot\ln(0)=0,$
the expression simplifies to $\frac{1}{2}\ln\left(\frac{1}{(1-x)^{(K/2)-1}\cdot x}\right)$,
which is minimized among $x\in[0,1]$ by $x=2/K.$ Against this conjecture,
the difference in expected payoff at node $n_{K-1}$ from Across versus
Drop is $(1-2/K)(g)+(2/K)(-\ell).$ This is strictly positive when
$g>\frac{2}{K-2}\ell.$ This means the continuation value at $n_{K-1}$
is at least $g$ larger than the payoff of Dropping at $n_{K-3},$
so again Across has strictly higher expected payoff than Drop. Inductively,
$(d_{BA}^{k})$ is optimal given the belief $(\hat{d}_{AB}^{k}).$
Also, $(d_{AB}^{k})$ is optimal as it results in the highest possible
payoff. We can similarly show that the conjecture $\hat{d}_{BB}^{k}$
with $\hat{d}_{BB}^{k}=2/K$ for $k$ even, $\hat{d}_{BB}^{k}=0$
for $k$ odd minimizes KL divergence conditional on $\Theta_{B}$
opponent, and $(d_{BB}^{k})$ is optimal given this conjecture.

As $p_{B}\to0,$ we find an  EZ where adherents of A have fitness
0, whereas the adherents of B have fitness at least $\frac{1}{2}(((K/2)-1)$$g-\ell)>0$
since $g>\frac{2}{K-2}\ell.$ This shows $\Theta_{A}$ is not evolutionarily
stable against $\Theta_{B}$.

But consider the same $(d_{AA},d_{AB},d_{BA})$ and suppose $d_{BB}^{k}=1$
for every $k$. Taking $p_{B}\to1,$ with $\lambda<1$, we find an
 EZ where adherents of B have fitness 0, adherents of A have fitness
$(1-\lambda)\cdot\frac{1}{2}\cdot((K/2)g+\ell)>0.$ This shows $\Theta_{B}$
is not evolutionarily stable against $\Theta_{A}$.
\end{proof}

\subsection{Proof of Proposition \ref{prop:stable_pop_share}}
\begin{proof}
In the centipede game, suppose $g>\frac{2}{K-2}\ell$. the misspecified
agent thinks a group B agent in the role of P2 and a group A agent
in either role has a probability $2/K$ of stopping at every node.
Under this belief, choosing to continue instead of drop means there
is a $(K-2)/K$ chance of gaining $g$, but a $2/K$ chance of losing
$\ell.$ Since we assume $g>\frac{2}{K-2}\ell$, it is strictly better
to continue. When $p$ fraction of the agents are correctly specified,
the fitness of $\Theta^{\bullet}$ is $p\cdot0+(1-p)\cdot(\frac{1}{2}\frac{g(K-2)}{2}+\frac{1}{2}(\frac{gK}{2}+\ell))$,
while the fitness of $\Theta^{An}$ is $p\cdot[\frac{1}{2}(\frac{g(K-2)}{2}-\ell)+\frac{1}{2}\frac{g(K-2)}{2}]+(1-p)[\frac{1}{2}(\frac{g(K-2)}{2}-\ell)+\frac{1}{2}(\frac{gK}{2}+\ell)]$.
The difference in fitness is 
\[
-p[\frac{1}{2}(\frac{g(K-2)}{2}-\ell)+\frac{1}{2}\frac{g(K-2)}{2}]+(1-p)\frac{1}{2}\ell.
\]
Simplifying, this is $\frac{1}{2}\ell-p\cdot\frac{g(K-2)}{2}$, a
strictly decreasing function in $p.$ When $p=\frac{\ell}{g(K-2)},$
which is a number strictly between 0 and 1/2 from the assumption $g>\frac{2}{K-2}\ell$
in the centipede game, the two models have the same fitness.
\end{proof}

\subsection{Proof of Proposition \ref{prop:stable_pop_share_dollar}}
\begin{proof}
In the $\overline{\Theta}^{An}$ vs. $\overline{\Theta}^{An}$ match,
the adherents of $\overline{\Theta}^{An}$ hold the belief that $\hat{d}_{BB}^{k}=2/K$
for every even $k$. In the role of P1, at node $k$ for $k\le K-3,$
stopping gives them $k$ but continuing gives them a $(K-2)/K$ chance
to get at least $k+2$, and we have $k\le\frac{K-2}{K}(k+2)\iff2k\le2K-4\iff k\le K-2$.
At node $K-1,$ the agent gets $K-1$ from dropping but expects $(K+2)\cdot\frac{K-2}{K}$
from continuing, and $(K+2)\cdot\frac{K-2}{K}-(K-1)=\frac{K^{2}-4-K^{2}+K}{K}=\frac{K-4}{K}>0$
since $K\ge6.$

In the $\overline{\Theta}^{\bullet}$ vs. $\overline{\Theta}^{An}$
match, the adherents of $\Theta^{An}$ hold the belief that $\hat{d}_{AB}^{k}=2/K$
for every $k.$ By the same arguments as before, the behavior of the
adherents of $\Theta^{An}$ are optimal given these beliefs. Also,
the adherents of $\Theta^{\bullet}$ have no profitable deviations
since they are best responding both as P1 and P2.

When $p$ fraction of the agents are correctly specified, in the dollar
game the fitness of $\overline{\Theta}^{\bullet}$ is $p\cdot0.5+(1-p)\cdot(\frac{1}{2}(K-1)+\frac{1}{2}K)$,
while the fitness of $\overline{\Theta}^{An}$ is $p\cdot0+(1-p)\cdot(\frac{1}{2}\cdot0+\frac{1}{2}K)$.
For any $p$, the fitness of $\overline{\Theta}^{\bullet}$ is strictly
higher than that of $\overline{\Theta}^{An}$.
\end{proof}

\section{\label{sec:Existence-and-Continuity}Existence and Continuity of
EZ}

We provide a few technical results about the existence of EZ and the
upper-hemicontinuity of the set of EZs with respect to population
share.  We suppose that $|\mathcal{G}|=1$ for simplicity, but analogous
results would hold for environments with multiple situations. Note
that the same learning channel that generates new stability phenomena
in Section \ref{sec:new_stability_phenomena} also leads to some difficulty
in establishing existence and continuity results, as agents draw different
inferences with different interaction structures.

Let two models, $\Theta_{A},\Theta_{B}$ be fixed. Also fix population
shares $p$ and matching assortativity $\lambda.$ Let $U_{A}:\mathbb{A}^{2}\times\Theta_{A}\to\mathbb{R}$
be such that $U_{A}(a_{i},a_{-i};F)=U_{i}(a_{i},a_{-i};\delta_{F})$
and let $U_{B}:\mathbb{A}^{2}\times\Theta_{B}\to\mathbb{R}$ be such
that $U_{B}(a_{i},a_{-i};F)=U_{i}(a_{i},a_{-i};\delta_{F})$.
\begin{assumption}
\label{assu:compact}$\mathbb{A},\Theta_{A},\Theta_{B}$ are compact
metrizable spaces.
\begin{assumption}
\label{assu:continuous_utility}$U_{A},U_{B}$ are continuous.
\begin{assumption}
\label{assu:finite_KL}For every $F\in\Theta_{A}\cup\Theta_{B}$ and
$a_{i},a_{-i}\in\mathbb{A},$ $K(F;a_{i},a_{-i})$ is well-defined
and finite.
\end{assumption}
\end{assumption}
\end{assumption}
Under Assumption \ref{assu:finite_KL}, we have the well-defined functions
$K_{A}:\Theta_{A}\times\mathbb{A}^{2}\to\mathbb{R}_{+}$ and $K_{B}:\Theta_{B}\times\mathbb{A}^{2}\to\mathbb{R}_{+}$,
where $K_{g}(F;a_{i},a_{-i}):=D_{KL}(F^{\bullet}(a_{i},a_{-i})\parallel F(a_{i},a_{-i}))$.
\begin{assumption}
\label{assu:continuous_KL} $K_{A}$ and $K_{B}$ are continuous.
\begin{assumption}
\label{assu:quasiconcave}$\mathbb{A}$ is convex and, for all $a_{-i}\in\mathbb{A}$
and $\mu\in\Delta(\Theta_{A})\cup\Delta(\Theta_{B})$, $a_{i}\mapsto U_{i}(a_{i},a_{-i};\mu)$
is quasiconcave.
\end{assumption}
\end{assumption}
We show existence of EZ using the Kakutani-Fan-Glicksberg fixed point
theorem, applied to the correspondence which maps strategy profiles
and beliefs over parameters into best replies and beliefs over KL-divergence
minimizing parameter. We start with a lemma.
\begin{lem}
\label{lem:inference_uhc}For $g\in\{A,B\}$, $a=(a_{AA},a_{AB},a_{BA},a_{BB})\in\mathbb{A}^{4},$
and $0\le m_{g}\le1$, let 
\[
\Theta_{g}^{*}(a,m_{g}):=\underset{\hat{F}\in\Theta_{g}}{\arg\min}\left\{ \begin{array}{c}
m_{g}\cdot K(\hat{F};a_{g,g},a_{g,g})+(1-m_{g})\cdot K(\hat{F};a_{g,-g},a_{-g,g})\end{array}\right\} .
\]
Then, $\Theta_{g}^{*}$ is upper hemicontinuous in its arguments.
\end{lem}
This lemma says the set of KL-minimizing parameters is upper hemicontinuous
in strategy profile and matching assortativity. This leads to the
existence result.
\begin{prop}
\label{prop:existence}Under Assumptions \ref{assu:compact}, \ref{assu:continuous_utility},
\ref{assu:finite_KL}, \ref{assu:continuous_KL}, and \ref{assu:quasiconcave},
an EZ exists.
\end{prop}
Next, upper hemicontinuity in $m_{g}$ in Lemma \ref{lem:inference_uhc}
allows us to deduce the upper hemicontinuity of the EZ correspondence
in population shares. 
\begin{prop}
\label{prop:uhc}Fix two models $\Theta_{A},\Theta_{B}$. Also fix
matching assortativity $\lambda\in[0,1].$ The set of EZ is an upper
hemicontinuous correspondence in $p_{B}$ under Assumptions \ref{assu:compact},
\ref{assu:continuous_utility}, \ref{assu:finite_KL}, and \ref{assu:continuous_KL}.
\end{prop}

\subsection{Proofs of Results in Appendix \ref{sec:Existence-and-Continuity}}

\subsubsection{Proof of Lemma \ref{lem:inference_uhc}}
\begin{proof}
Write the minimization objective as 
\[
W(a,F,m_{g}):=m_{g}K_{g}(F;a_{g,g},a_{g,g})+(1-m_{g})K_{g}(F;a_{g,-g},a_{-g,g}),
\]
a continuous function of $(a,F,m_{g})$ by Assumption \ref{assu:continuous_KL}.
Suppose we have a sequence $(a^{(n)},m_{g}^{(n)})\to(a^{*},m_{g}^{*})\in\mathbb{A}^{4}\times[0,1]$
and let $F^{(n)}\in\Theta_{g}^{*}(a^{(n)},m_{g}^{(n)})$ for each
$n,$ with $F^{(n)}\to F^{*}\in\Theta_{g}.$ For any other $\hat{F}\in\Theta_{g},$
note that $W(a^{*},m_{g}^{*},\hat{F})=\lim_{n\to\infty}W(a^{(n)},m_{g}^{(n)},\hat{F})$
by continuity. But also by continuity, $W(a^{*},m_{g}^{*},F^{*})=\lim_{n\to\infty}W(a^{(n)},m_{g}^{(n)},F^{(n)})$
and $W(a^{(n)},m_{g}^{(n)},F^{(n)})\le W(a^{(n)},m_{g}^{(n)},\hat{F})$
for every $n.$ It therefore follows $W(a^{*},m_{g}^{*},F^{*})\le W(a^{*},m_{g}^{*},\hat{F}).$
\end{proof}

\subsubsection{Proof of Proposition \ref{prop:existence}}
\begin{proof}
Consider the correspondence $\Gamma:\mathbb{A}^{4}\times\Delta(\Theta_{A})\times\Delta(\Theta_{B})\rightrightarrows\mathbb{A}^{4}\times\Delta(\Theta_{A})\times\Delta(\Theta_{B}),$
\begin{align*}
 & \Gamma(a_{AA},a_{AB},a_{BA},a_{BB},\mu_{A},\mu_{B}):=\\
 & (\text{BR}(a_{AA},\mu_{A}),\text{BR}(a_{BA},\mu_{A}),\text{BR}(a_{AB},\mu_{B}),\text{BR}(a_{BB},\mu_{B}),\Delta(\Theta_{A}^{*}(a)),\Delta(\Theta_{B}^{*}(a))),
\end{align*}
where $\text{BR}(a_{-i},\mu_{g}):=\underset{\hat{a}_{i}\in\mathbb{A}}{\arg\max}U_{g}(\hat{a}_{i},a_{-i};\mu_{g})$
and, for each $g\in\{A,B\},$ the correspondence $\Theta_{g}^{*}$
is defined with $m_{g}=\lambda+(1-\lambda)p_{g},$ $m_{-g}=1-m_{g}.$
It is clear that fixed points of $\Gamma$ are EZ.

We apply the Kakutani-Fan-Glicksberg theorem (see, e.g, Corollary
17.55 in \citet{aliprantis_infinite_2006}). By Assumptions \ref{assu:compact}
and \ref{assu:quasiconcave}, $\mathbb{A}$ is acompact and convex
metric space, and each $\Theta_{g}$ is a compact metric space, so
it follows the domain of $\Gamma$ is a nonempty, compact and convex
metric space. We need only verify that $\Gamma$ has closed graph,
non-empty values, and convex values.

To see that $\Gamma$ has closed graph, the previous lemma shows the
upper hemicontinuity of $\Theta_{A}^{*}(a)$ and $\Theta_{B}^{*}(a)$
in $a,$ and Theorem 17.13 of \citet{aliprantis_infinite_2006} then
implies $\Delta(\Theta_{A}^{*}(a))$ and $\Delta(\Theta_{B}^{*}(a))$
are also upper hemicontinuous in $a.$ It is a standard argument that
since Assumption \ref{assu:continuous_utility} supposes $U_{A},U_{B}$
are continuous, it implies the best-response correspondences $\text{BR}(a_{AA},\mu_{A}),$
$\text{BR}(a_{BA},\mu_{A}),$ $\text{BR}(a_{AB},\mu_{B}),$ $\text{BR}(a_{BB},\mu_{B})$
have closed graphs.

To see that $\Gamma$ is non-empty, recall that each $\hat{a}_{i}\mapsto U_{g}(\hat{a}_{i},a_{-i};\mu_{g})$
is a continuous function on a compact domain, so it must attain a
maximum on $\mathbb{A}.$ Similarly, the minimization problem that
defines each $\Theta_{g}^{*}(a)$ is a continuous function of $F$
over a compact domain of possible $F$'s, so it attains a minimum.
Thus each $\Delta(\Theta_{g}^{*}(a))$ is the set of distributions
over a non-empty set.

To see that $\Gamma$ is convex valued, clearly $\Delta(\Theta_{A}^{*}(a))$
and $\Delta(\Theta_{B}^{*}(a))$ are convex valued by definition.
Also, $\hat{a}_{i}\mapsto U_{A}(\hat{a}_{i},a_{AA};\mu_{A})$ is quasiconcave
by Assumption \ref{assu:quasiconcave}. That means if $a_{i}^{'},a_{i}^{''}\in\text{BR}(a_{AA},\mu_{A}),$
then for any convex combination $\tilde{a}_{i}$ of $a_{i}^{'},a_{i}^{''},$
we have{\small{}{}{} $U_{A}(\tilde{a}_{i},a_{AA};\mu_{A})\ge\min(U_{A}(a_{i}^{'},a_{AA};\mu_{A}),$}
$U_{A}(a_{i}^{''},a_{AA};\mu_{A}))=\max_{\hat{a}_{i}\in\mathbb{A}}U_{A}(\hat{a}_{i},a_{AA};\mu_{A})$.
Therefore, $\text{BR}(a_{AA},\mu_{A})$ is convex. For similar reasons,
$\text{BR}(a_{BA},\mu_{A}),$ $\text{BR}(a_{AB},\mu_{B}),$ $\text{BR}(a_{BB},\mu_{B})$
are convex.
\end{proof}

\subsubsection{Proof of Proposition \ref{prop:uhc}}
\begin{proof}
Since $\mathbb{A}^{4}\times\Delta(\Theta_{A})\times\Delta(\Theta_{B})$
is compact by Assumption \ref{assu:compact}, we need only show that
for every sequence $(p_{B}^{(k)})_{k\ge1}$ and $(a^{(k)},\mu^{(k)})_{k\ge1}=(a_{AA}^{(k)},a_{AB}^{(k)},a_{BA}^{(k)},a_{BB}^{(k)},\mu_{A}^{(k)},\mu_{B}^{(k)})_{k\ge1}$
such that for every $k$, $(a^{(k)},\mu^{(k)})$ is an EZ with $p=(1-p_{B}^{(k)},p_{B}^{(k)})$,
$p_{B}^{(k)}\to p_{B}^{*}$, and $(a^{(k)},\mu^{(k)})\to(a^{*},\mu^{*})$,
then $(a^{*},\mu^{*})$ is an EZ with $p=(1-p_{B}^{*},p_{B}^{*})$.

We first show for all $g,g^{'}\in\{A,B\},$ $a_{g,g^{'}}^{*}$ is
optimal against $a_{g^{'},g}^{*}$ under the belief $\mu_{g}^{*}.$
Assortativity does not matter here, since optimality applies within
all type match-ups. By Assumption \ref{assu:continuous_utility},
$U_{g}(a_{i},a_{-i};F)$ is continuous, so by property of convergence
in distribution, $U_{g}(a_{g,g^{'}}^{(k)},a_{g^{'},g}^{(k)};\mu_{g}^{(k)})\to U_{g}(a_{g,g^{'}}^{*},a_{g^{'},g}^{*};\mu_{g}^{*})$.
For any other $\hat{a}_{i}\in\mathbb{A},$ $U_{g}(\hat{a}_{i},a_{g^{'},g}^{(k)};\mu_{g}^{(k)})\to U_{g}(\hat{a}_{i},a_{g^{'},g}^{*};\mu_{g}^{*})$
and for every $k,$ $U_{g}(a_{g,g^{'}}^{(k)},a_{g^{'},g}^{(k)};\mu_{g}^{(k)})\ge U_{g}(\hat{a}_{i},a_{g^{'},g}^{(k)};\mu_{g}^{(k)}).$
Therefore $a_{g,g^{'}}^{*}$ best responds to $a_{g^{'},g}^{*}$ under
belief $\mu_{g}^{*}.$

Next, we show parameters in the support of $\mu_{g}^{*}$ minimize weighted
KL divergence for group $g.$ First consider the correspondence $H:\mathbb{A}^{4}\times[0,1]\rightrightarrows\Theta_{g}$
where $H(a,p_{g}):=\Theta_{g}^{*}(a,\lambda+(1-\lambda)(p_{g}))$.
Then $H$ is upper hemicontinuous by Lemma \ref{lem:inference_uhc}.
Since $H(a,p_{g})$ represents the minimizers of a continuous function
on a compact domain, it is non-empty and closed. By Theorem 17.13
of \citet{aliprantis_infinite_2006}, the correspondence $\tilde{H}:\mathbb{A}^{4}\times[0,1]\rightrightarrows\Delta(\Theta_{g})$
defined so that $\tilde{H}(a,p_{g}):=\Delta(H(a,p_{g}))$ is also
upper hemicontinuous. For every $k,$ $\mu_{g}^{(k)}\in\tilde{H}(a^{(k)},p_{g}^{(k)})$,
and $\mu_{g}^{(k)}\to\mu_{g}^{*}$, $a^{(k)}\to a^{*},p_{g}^{(k)}\to p_{g}^{*}.$
Therefore, $\mu_{g}^{*}\in\tilde{H}(a^{*},p_{g}^{*}),$ that is to
say $\mu_{g}^{*}$ is supported on the minimizers of weighted KL divergence.
\end{proof}

\section{Learning Foundation of EZ and EZ-SU\label{sec:Learning-Foundation}}

We provide a unified foundation for EZ and EZ-SU as the steady state
of a learning system. This foundation considers a world where agents
have prior beliefs over extended parameters in an extended models, as
in Section \ref{sec:ABEE}. At the end of every match, each agent
observes her consequence and a noisy signal about the matched opponent's
strategy. We show that under any asymptotically myopic policy, if
behavior and beliefs converge, then the limit steady state must be
an EZ-SU when the noisy signals about opponent's strategy are uninformative.
Sufficiently accurate signals about opponent's play cause the steady
states to be EZs, if the extended models allow agents to make rich
enough inferences about opponents' strategies. Finally, if the true
situation is redrawn every $T$ periods and the agents reset their
beliefs over extended parameters to their prior belief when the situation
is redrawn, then their average payoffs approach their fitness in the
EZ or EZ-SU when $T$ is large.

\subsection{\label{subsec:Regularity-Assumptions-learning}Regularity Assumptions}

We make some regularity assumptions on the objective environments
and on the extended models $\overline{\Theta}_{A},\overline{\Theta}_{B}$.
These are similar to the regularity assumptions from Section \ref{sec:Existence-and-Continuity}.

Suppose the strategy set $\mathbb{A}$ is finite. Suppose the marginals
of the extended models $\overline{\Theta}_{A},\overline{\Theta}_{B}$
on the dimension of fundamental uncertainty, denoted as $\Theta_{A},\Theta_{B}$,
are compact and metrizable spaces. Endow $\overline{\Theta}_{A}$
and $\overline{\Theta}_{B}$ with the product metric. Suppose that
every $(a_{A},a_{B},F)\in\overline{\Theta}_{A}\cup\overline{\Theta}_{B}$
is so that for every $(a_{i},a_{-i})\in\mathbb{A}^{2}$ and every
situation $G,$ whenever $f^{\bullet}(a_{i},a_{-i},G)(y)>0$, we also
get $f(a_{i},a_{A})(y)>0$ and $f(a_{i},a_{B})(y)>0$, where $f$
is the density or probability mass function for $F$.

For each $g,g^{'}\in\{A,B\},$ define $K_{g,g^{'}}:\mathbb{A}^{2}\times\mathcal{G}\times\overline{\Theta}_{g}\to\mathbb{R}$
by $K_{g,g^{'}}(a_{i},a_{-i},G;(a_{A},a_{B},F))=D_{KL}(F^{\bullet}(a_{i},a_{-i},G)\parallel F(a_{i},a_{g^{'}})).$
This is the KL divergence of the parameter $(a_{A},a_{B},F)\in\overline{\Theta}_{g}$
in situation $G$ based on the data generated from the strategy profile
$(a_{i},a_{-i})$. Suppose each $K_{g,g^{'}}$ is well defined and
a continuous function of the extended parameter $(a_{A},a_{B},F)$.

For $g\in\{A,B\}$, $F\in\Theta_{g}$, let $U_{g}(a_{i},a_{-i};F)$
be the expected payoffs of the strategy profile $(a_{i},a_{-i})$
for $i$ when consequences are drawn according to $F.$ Assume $U_{A},U_{B}$
are continuous.

Suppose for every extended model $\overline{\Theta}_{g}$ and every
$(a_{A},a_{B},F)\in\overline{\Theta}_{g}$ and $\epsilon>0,$ there
exists an open neighborhood $V\subseteq\overline{\Theta}_{g}$ of
$(a_{A},a_{B},F)$, so that for every $(\hat{a}_{A},\hat{a}_{B},\hat{F})\in V$,
$1-\epsilon\le f(a_{i},a_{A})(y)/\hat{f}(a_{i},\hat{a}_{A})(y)\le1+\epsilon$
and $1-\epsilon\le f(a_{i},a_{B})(y)/\hat{f}(a_{i},\hat{a}_{B})(y)\le1+\epsilon$
for all $a_{i}\in\mathbb{A},y\in\mathbb{Y}$. Also suppose there is
some $M>0$ so that $\ln(f(a_{i},a_{A})(y))$ and $\ln(f(a_{i},a_{B})(y))$
are bounded in $[-M,M]$ for all $(a_{A},a_{B},F)\in\overline{\Theta}_{g}$,
$a_{i},a_{-i}\in\mathbb{A},y\in\mathbb{Y}$.

\subsection{Learning Environment}

We first consider an environment with only one true situation, $|\mathcal{G}|=1.$
Time is discrete and infinite, $t=0,1,2,...$ A unit mass of agents,
$i\in[0,1]$, enter the society at time 0. A $p_{A}\in(0,1)$ measure
of them are assigned to model $A$ and the rest are assigned to model
$B$. Each agent born into model $g$ starts with the same full support
prior over the extended model, $\mu_{g}^{(0)}\in\Delta(\overline{\Theta}_{g})$,
and believes there is some $(a_{A},a_{B},F)\in\overline{\Theta}_{g}$
so that every group $g$ opponent always plays $a_{g}$ and the consequences
are always generated by $F$.

In each period $t$, agents are matched up partially assortatively
to play the stage game. Assortativity is $\lambda\in(0,1).$ Each
person in group $g$ has $\lambda+(1-\lambda)p_{g}$ chance of matching
with someone from group $g,$ and matches with someone from group
$-g$ with the complementary chance. Each agent $i$ observes their
opponent's group membership and chooses a strategy $a_{i}^{(t)}\in\mathbb{A}$.
At the end of the match, the agent observes own consequence $y_{i}^{(t)}$
and a signal $x_{i}^{(t)}\in\mathbb{A}$ about the opponent's play,
where $x_{i}^{(t)}$ equals the matched opponent's strategy $a_{-i}$
with probability $\tau\in[0,1),$ and it is uniformly random on $\mathbb{A}$
with the complementary probability. To give a foundation for a EZ-SU,
we consider $\tau=0$, so the signal $x_{i}$ is uninformative. To
give a foundation for EZ, we consider $\tau$ close to 1.

Thus, the space of histories from one period is $\{A,B\}\times\mathbb{A}\times\mathbb{Y}\times\mathbb{A}$,
with typical element $(g_{i}^{(t)},a_{i}^{(t)},y_{i}^{(t)},x_{i}^{(t)})$.
It records the group membership of $i$'s opponent $g_{i}^{(t)}$,
$i$'s strategy $a_{i}^{(t)},$, $i$'s consequence $y_{i}^{(t)}$,
and $i$'s ex-post signal about the matched opponent's play, $x_{i}^{(t)}$.
Let $\mathbb{H}$ denote the space of all finite-length histories.

Given the assumption on the two models, there is a well-defined
Bayesian belief operator for each model $g,$ $\mu_{g}:\mathbb{H}\to\Delta(\overline{\Theta}_{g}),$
mapping every finite-length history into a belief over extended parameters
in $\overline{\Theta}_{g}$, starting with the prior $\mu_{g}^{(0)}.$

We also take as exogenously given policy functions for choosing strategies
after each history. That is, $\mathfrak{a}_{g,g^{'}}:\mathbb{H}\to\mathbb{A}$
for every $g,g^{'}\in\{A,B\}$ gives the strategy that a group $g$
agent uses against a group $g^{'}$ opponent after every history.
Assume these policy functions are asymptotically myopic.
\begin{assumption}
\label{assu:asymptotic_myopia}For every $\epsilon>0,$ there exists
$N$ so that for any history $h$ containing at least $N$ matches
against opponents of each group, $\mathfrak{a}_{g,g^{'}}(h)$ is an
$\epsilon$-best response to the Bayesian belief $\mu_{g}(h)$.
\end{assumption}
From the perspective of each agent $i$ in group $g,$ $i$'s play
against groups A and B, as well as $i$'s belief over $\overline{\Theta}_{g},$
is a stochastic process $(\tilde{a}_{iA}^{(t)},\tilde{a}_{iB}^{(t)},\tilde{\mu}_{i}^{(t)})_{t\ge0}$
valued in $\mathbb{A}\times\mathbb{A}\times\Delta(\overline{\Theta}_{g}).$
The randomness is over the groups of opponents matched with in different
periods, the strategies they play, and the random consequences and
ex-post signals drawn at the end of the matches. At the same time,
since there is a continuum of agents, the distribution over histories
within each population in each period is deterministic. As such, there
is a deterministic sequence $(\alpha_{AA}^{(t)},\alpha_{AB}^{(t)},\alpha_{BA}^{(t)},\alpha_{BA}^{(t)},\nu_{A}^{(t)},\nu_{B}^{(t)})\in\Delta(\mathbb{A})^{4}\times\Delta(\Delta(\overline{\Theta}_{A}))\times\Delta(\Delta(\overline{\Theta}_{B}))$
that describes the distributions of play and beliefs that prevail
in the two sub-populations in every period $t.$

\subsection{Steady State Limits are EZ-SUs and EZs}

We state and prove the learning foundation of EZ-SU and EZ. For $(\alpha^{(t)})_{t}$
a sequence valued in $\Delta(\mathbb{A})$ and $a^{*}\in\mathbb{A},$
$\alpha^{(t)}\to a^{*}$ means $\mathbb{E}_{\hat{a}\sim\alpha^{(t)}}\parallel\hat{a}-a^{*}\parallel\to0$
as $t\to\infty$. For $(\nu^{(t)})_{t}$ a sequence valued in $\Delta(\Delta(\overline{\Theta}_{g}))$
and $\mu^{*}\in\Delta(\overline{\Theta}_{g}),$ $\nu^{(t)}\to\mu^{*}$
means $\mathbb{E}_{\hat{\mu}\sim\nu^{(t)}}\parallel\hat{\mu}-\mu^{*}\parallel\to0$
as $t\to\infty.$
\begin{prop}
\label{prop:learning}Suppose the regularity assumptions in Section
\ref{subsec:Regularity-Assumptions-learning} hold, and suppose Assumption
\ref{assu:asymptotic_myopia} holds.

Suppose $\tau=0$. Suppose there exists $(a_{AA}^{*},a_{AB}^{*},a_{BA}^{*},a_{BB}^{*},\mu_{A}^{*},\mu_{B}^{*})\in\mathbb{A}^{4}\times\Delta(\overline{\Theta}_{A})\times\Delta(\overline{\Theta}_{B})$
so that $(\alpha_{AA}^{(t)},\alpha_{AB}^{(t)},\alpha_{BA}^{(t)},\alpha_{BA}^{(t)},\nu_{A}^{(t)},\nu_{B}^{(t)})\to(a_{AA}^{*},a_{AB}^{*},a_{BA}^{*},a_{BB}^{*},\mu_{A}^{*},\mu_{B}^{*})$
and for each agent $i$ in group $g,$ almost surely $(\tilde{a}_{iA}^{(t)},\tilde{a}_{iB}^{(t)},\tilde{\mu}_{i}^{(t)})\to(a_{gA}^{*},a_{gB}^{*},\mu_{g}^{*})$.
Then, $(a_{AA}^{*},a_{AB}^{*},a_{BA}^{*},a_{BB}^{*},\mu_{A}^{*},\mu_{B}^{*})$
is an EZ-SU.

Suppose for each $g,$ the extended model $\overline{\Theta}_{g}=\mathbb{A}^{2}\times\Theta_{g}$
for some model $\Theta_{g}$ -- that is, each group can make any
inference about opponents' strategies. There exists some $\underline{\tau}<1$
so that for every $\tau\in(\underline{\tau},1)$ and $(a_{AA}^{*},a_{AB}^{*},a_{BA}^{*},a_{BB}^{*},\mu_{A}^{*},\mu_{B}^{*})$
satisfying the above conditions, we have that $\mu_{A}^{*}$ puts
probability 1 on $(a_{AA}^{*},a_{AB}^{*})$, $\mu_{B}^{*}$ puts probability
1 on $(a_{BA}^{*},a_{BB}^{*}),$ and $(a_{AA}^{*},a_{AB}^{*},a_{BA}^{*},a_{BB}^{*},\mu_{A}^{*}|_{\Theta_{A}},\mu_{B}^{*}|_{\Theta_{B}})$
is an EZ, where $\mu_{g}^{*}|_{\Theta_{g}}$ is the marginal of the
belief $\mu_{g}^{*}$ on the model $\Theta_{g}.$
\end{prop}
\begin{proof}
We first consider the case of $\tau=0,$ so the uninformative ex-post
signals may be ignored.

For $\mu$ a belief and $g\in\{A,B\},$ let $u^{\mu}(a_{i};g)$ represent
subjective expected payoff from playing $a_{i}$ against group $g$.
Suppose $a_{AA}^{*}\notin\text{argmax}_{\hat{a}\in\mathbb{A}}u^{\mu_{A}^{*}}(\hat{a};A)$
(the other cases are analogous). By the continuity assumptions on
$U_{A}$ (which is also bounded because $\Theta_{A}$ is bounded),
there are some $\epsilon_{1},\epsilon_{2}>0$ so that whenever $\mu_{i}\in\Delta(\overline{\Theta}_{A})$
with $\parallel\mu_{i}-\mu_{A}^{*}\parallel<\epsilon_{1}$, we also
have $u^{\mu_{i}}(a_{AA}^{*};A)<\max_{\hat{a}\in\mathbb{A}}u^{\mu_{i}}(\hat{a};A)-\epsilon_{2}.$
By the definition of asymptotically empirical best responses, find
$N$ so that $\mathfrak{a}_{A,A}(h)$ must be a myopic $\epsilon_{2}$-best
response when there are at least $N$ periods of matches against A
and B. Agent $i$ has a strictly positive chance to match with groups
A and B in every period. So, at all except a null set of points in
the probability space, $i$'s history eventually records at least
$N$ periods of play by groups A and B. Also, by assumption, almost
surely $\tilde{\mu}_{i}^{(t)}\to\mu_{A}^{*}.$ This shows that by
asymptotically myopic best responses, almost surely $\tilde{a}_{iA}^{(k)}\not\to a_{AA}^{*},$
a contradiction.

Now suppose some $\theta_{A}^{*}=(a_{A}^{*},a_{B}^{*},f^{*})$ in
the support of $\mu_{A}^{*}$ does not minimize the weighted KL divergence
in the definition of EZ-SU (the case of a parameter $\theta_{B}^{*}$
in the support of $\mu_{B}^{*}$ not minimizing is similar). Then
we have 
\[
\theta_{A}^{*}\notin\underset{\hat{\theta}\in\overline{\Theta}_{A}}{\text{argmin}}\left[\begin{array}{c}
(\lambda+(1-\lambda)p_{A})\cdot D_{KL}(F^{\bullet}(a_{AA}^{*},a_{AA}^{*})\parallel\hat{F}(a_{AA}^{*},\hat{a}_{A}))\\
+(1-\lambda)(1-p_{A})\cdot D_{KL}(F^{\bullet}(a_{AB}^{*},a_{BA}^{*})\parallel\hat{F}(a_{AB}^{*},\hat{a}_{B}))
\end{array}\right]
\]
where $\hat{\theta}=(\hat{a}_{A},\hat{a}_{B},\hat{F}).$

This is equivalent to: 
\[
\theta_{A}^{*}\notin\underset{\hat{\theta}\in\overline{\Theta}_{A}}{\text{argmax}}\left[\begin{array}{c}
(\lambda+(1-\lambda)p_{A})\cdot\mathbb{E}_{y\sim F^{\bullet}(a_{AA}^{*},a_{AA}^{*})}\ln(\hat{f}(a_{AA}^{*},\hat{a}_{A})(y))\\
+(1-\lambda)(1-p_{A})\cdot\mathbb{E}_{y\sim F^{\bullet}(a_{AB}^{*},a_{BA}^{*})}\ln(\hat{f}(a_{AB}^{*},\hat{a}_{B})(y))
\end{array}\right]
\]

Let this objective, as a function of $\hat{\theta}$, be denoted $WL(\hat{\theta}).$
There exists $\theta_{A}^{opt}=(a_{A}^{opt},a_{B}^{opt},f^{opt})\in\overline{\Theta}_{A}$
and $\delta,\epsilon>0$ so that $(1-\delta)WL(\theta_{A}^{opt})-2\delta M-3\epsilon>(1-\delta)WL(\theta_{A}^{*}).$
By assumption on the primitives, find open neighborhoods $V^{opt}$
and $V^{*}$ of $\theta_{A}^{opt},\theta_{A}^{*}$ respectively, so
that for all $a_{i}\in\mathbb{A},$ $g\in\{A,B\},$ $y\in\mathbb{Y}$,
$1-\epsilon\le f^{opt}(a_{i},a_{g}^{opt})(y)/\hat{f}(a_{i},\hat{a}_{g})(y)\le1+\epsilon$,
for all $\hat{\theta}=(\hat{a}_{A},\hat{a}_{B},\hat{f})\in V^{opt}$,
and also $1-\epsilon\le f^{*}(a_{i},a_{g}^{*})(y)/\hat{f}(a_{i},\hat{a}_{g})(y)\le1+\epsilon$
for all $\hat{\theta}=(\hat{a}_{A},\hat{a}_{B},\hat{f})\in V^{*}$.
Also, by convergence of play in the populations, find $T_{1}$ so
that in all periods $t\ge T_{1},$ $\alpha_{AA}^{(t)}(a_{AA}^{*})\ge1-\delta$
and $\alpha_{BA}^{(t)}(a_{BA}^{*})\ge1-\delta$.

For $T_{2}\ge T_{1},$ consider a probability space defined by $\Omega:=(\{A,B\}\times\mathbb{A}^{2}\times(\mathbb{Y})^{\mathbb{A}^{2}})^{\infty}$
that describes the randomness in an agent's learning process starting
with period $T_{2}+1$. For a point $\omega\in\Omega$ and each period
$T_{2}+s$, $s\ge1$, $\omega_{s}=(g,a_{-i,A},a_{-i,B},(y_{a_{i},a_{-i}})_{(a_{i},a_{-i})\in\mathbb{A}^{2}})$
specifies the group $g$ of the matched opponent, the play $a_{-i,A},a_{-i,B}$
of hypothetical opponents from groups A and B, and the hypothetical
consequence $y_{a_{i},a_{-i}}$ that would be generated for every
pair of strategies $(a_{i},a_{-i})$ played. As notation, let $opp(\omega,s)$,
$a_{-i,A}(\omega,s),$ $a_{-i,B}(\omega,s)$, and $y_{a_{i},a_{-i}}(\omega,s)$
denote the corresponding components of $\omega_{s}.$ Define $\mathbb{P}_{T_{2}}$
over this space in the natural way. That is, it is independent across
periods, and within each period, the density (or probability mass
function if $\mathbb{Y}$ is finite) of $\omega_{s}=(g,a_{-i,A},a_{-i,B},(y_{a_{i},a_{-i}})_{(a_{i},a_{-i})\in\mathbb{A}^{2}})$
is 
\[
m_{g}\cdot\alpha_{AA}^{(T_{2}+s)}(a_{-i,A})\alpha_{BA}^{(T_{2}+s)}(a_{-i,B})\cdot\prod_{(a_{i},a_{-i})\in\mathbb{A}^{2}}f^{\bullet}(a_{i},a_{-i})(y_{a_{i},a_{-i}}),
\]
where $m_{g}$ is the probability of $i$ from group A being matched
up against an opponent of group $g,$ that is $m_{A}=(\lambda+(1-\lambda)p_{A})$,
$m_{B}=(1-\lambda)(1-p_{A}).$

For $\theta=(a_{A}^{\theta},a_{B}^{\theta},F^{\theta})\in\overline{\Theta}_{A}$
with $f^{\theta}$ the density of $F^{\theta}$, $\omega\in\Omega,$
consider the stochastic process 
\[
\ell_{s}(\theta,\omega):=\frac{1}{s}\sum_{t=T_{2}+1}^{T_{2}+s}\ln(f^{\theta}(a_{AA}^{*},a_{opp(\omega,t)}^{\theta})(y_{a_{AA}^{*},a_{-i,opp(\omega,t)}(\omega,t)}(\omega,t)).
\]
By choice of the neighborhood $V^{*},$ 
\begin{align*}
\limsup_{s}\sup_{\theta_{A}\in V^{*}}\ell_{s}(\theta_{A},\omega) & \le\epsilon+\frac{1}{s}\sum_{t=T_{2}+1}^{T_{2}+s}\ln(f^{*}(a_{AA}^{*},a_{opp(\omega,t)}^{*})(y_{a_{AA}^{*},a_{-i,opp(\omega,t)}(\omega,t)}(\omega,t))\\
 & \le\epsilon+\frac{1}{s}\sum_{t=T_{2}+1}^{T_{2}+s}\begin{array}{c}
1_{\{a_{-i,opp(\omega,t)}(\omega,t)=a_{opp(\omega,t),A}^{*}\}}\cdot\ln(f^{*}(a_{AA}^{*},a_{opp(\omega,t)}^{*})(y_{a_{AA}^{*},a_{opp(\omega,t),A}^{*}}(\omega,t))\\
(1-1_{\{a_{-i,opp(\omega,t)}(\omega,t)=a_{opp(\omega,t),A}^{*}\}})\cdot M.
\end{array}
\end{align*}
Since $T_{2}\ge T_{1},$ in every period $t,$ $\mathbb{P}_{T_{2}}(a_{-i,opp(\omega,t)}(\omega,t)=a_{opp(\omega,t),A}^{*})\ge1-\delta$.
Let $(\xi_{k})_{k\ge1}$ a related stochastic process: it is i.i.d.
such that each $\xi_{k}$ has $\delta$ chance to be equal to $M,$
$(1-\delta)m_{A}$ chance to be distributed according to $\ln(f^{*}(a_{AA}^{*},a_{A}^{*})(y))$
where $y\sim f^{\bullet}(a_{AA}^{*},a_{AA}^{*}),$ and $(1-\delta)m_{B}$
chance to be distributed according to $\ln(f^{*}(a_{AB}^{*},a_{B}^{*})(y))$
where $y\sim f^{\bullet}(a_{AB}^{*},a_{BA}^{*}).$ By law of large
numbers, $\frac{1}{s}\sum_{k=1}^{s}\xi_{k}$ converges almost surely
to $\delta M+(1-\delta)WL(\theta_{A}^{*}).$ By this comparison, $\limsup_{s}\sup_{\theta_{A}\in V^{*}}\ell_{s}(\theta_{A},\omega)\le\epsilon+\delta M+(1-\delta)WL(\theta_{A}^{*})$
$\mathbb{P}_{T_{2}}$-almost surely. By a similar argument, $\liminf_{s}\inf_{\theta_{A}\in V^{opt}}\ell_{s}(\theta_{A},\omega)\ge-\epsilon-\delta M+(1-\delta)WL(\theta_{A}^{opt})$
$\mathbb{P}_{T_{2}}$-almost surely.

Along any $\omega$ where we have both $\limsup_{s}\sup_{\theta_{A}\in V^{*}}\ell_{s}(\theta_{A},\omega)\le\epsilon+\delta M+(1-\delta)WL(\theta_{A}^{*})$
and $\liminf_{s}\inf_{\theta_{A}\in V^{opt}}\ell_{s}(\theta_{A},\omega)\ge-\epsilon-\delta M+(1-\delta)WL(\theta_{A}^{opt})$,
if $\omega$ also leads to $i$ always playing $a_{AA}^{*}$ against
group A and $a_{AB}^{*}$ against group B in all periods starting
with $T_{2}+1,$ then the posterior belief assigns to $V^{*}$ must
tend to 0, hence $\tilde{\mu}_{i}^{(t)}\not\to\mu_{A}^{*}.$ Starting
from any length $T_{2}$ history $h,$ there exists a subset $\hat{\Omega}_{h}\subseteq\Omega$
that leads to $i$ not playing the EZ-SU strategy in at least one
period starting with $T_{2}+1.$ So conditional on $h,$ the probability
of $\tilde{\mu}_{i}^{(t)}\to\mu_{A}^{*}$ is no larger than $1-\mathbb{P}_{T_{2}}(\hat{\Omega}_{h}).$
The unconditional probability is therefore no larger than $\mathbb{E}_{h}[1-\mathbb{P}_{T_{2}}(\hat{\Omega}_{h})],$
where $\mathbb{E}_{h}$ is taken with respect to the distribution
of period $T_{2}$ histories for $i.$ But this term is also the probability
of $i$ playing non-EZ-SU action at least once starting with period
$T_{2}.$ Since there are finitely many actions and $(\tilde{a}_{iA}^{(t)},\tilde{a}_{iB}^{(t)})\to(a_{AA}^{*},a_{AB}^{*})$
almost surely, $\mathbb{E}_{h}[1-\mathbb{P}_{T_{2}}(\hat{\Omega}_{h})]$
tends to 0 as $T_{2}\to\infty.$ We have a contradiction as this shows
$\tilde{\mu}_{i}^{(t)}\not\to\mu_{A}^{*}$ with probability 1.

Now consider the foundation for EZs. Suppose Let $\bar{K}<\infty$
be an upper bound on $K_{g,g^{'}}(a_{i},a_{-i};(a_{A},a_{B},F))$
across all $g,g^{'}\in\{A,B\},$ $a_{i},a_{-i}\in\mathbb{A},$ $(a_{A},a_{B},F)\in\overline{\Theta}_{g}.$
Here $\bar{K}$ is finite because $\mathbb{A}$ is finite and $K_{g,g^{'}}$
is continuous in the extended parameter, which is from a compact domain.
Let $F_{\tau}^{X}(a_{-i})\in\Delta(\mathbb{A})$ represent the distribution
of ex-post signals given precision $\tau,$ when opponent plays $a_{-i}\in\mathbb{A}.$
It is clear that there exists some $\underline{\tau}<1$ so that for
any $a_{-i}\ne a_{-i}^{'}$, $\tau\in(\underline{\tau},1),$ we get
$\min(m_{A},m_{B})\cdot D_{KL}(F_{\tau}^{X}(a_{-i})\parallel F_{\tau}^{X}(a_{-i}^{'}))>\bar{K}.$
Therefore, given any $(a_{AA}^{*},a_{AB}^{*},a_{BA}^{*})\in\mathbb{A}^{3},$
the solution to 
\[
\underset{\hat{\theta}\in\overline{\Theta}_{A}}{\min}\left[\begin{array}{c}
(\lambda+(1-\lambda)p_{A})\cdot[D_{KL}(F^{\bullet}(a_{AA}^{*},a_{AA}^{*})\parallel\hat{F}(a_{AA}^{*},\hat{a}_{A}))+D_{KL}(F_{\tau}^{X}(a_{AA}^{*})\parallel F_{\tau}^{X}(\hat{a}_{A}))]\\
+(1-\lambda)(1-p_{A})\cdot[D_{KL}(F^{\bullet}(a_{AB}^{*},a_{BA}^{*})\parallel\hat{F}(a_{AB}^{*},\hat{a}_{B}))+D_{KL}(F_{\tau}^{X}(a_{BA}^{*})\parallel F_{\tau}^{X}(\hat{a}_{B})]
\end{array}\right]
\]
must satisfy $\hat{a}_{A}=a_{AA}^{*},$ $\hat{a}_{B}=a_{BA}^{*}$,
because $(a_{AA}^{*},a_{BA}^{*},F)$ for any $F\in\Theta_{A}$ has
a KL divergence no larger than $\bar{K}$. On the other hand, any
$(\hat{a}_{A},\hat{a}_{B},\hat{F})$ with either $\hat{a}_{A}\ne a_{AA}^{*}$
or $\hat{a}_{B}\ne a_{BA}^{*}$ has KL divergence strictly larger
than $\bar{K}$ by the choice of $\tau$. The rest of the argument
is similar to the case of EZ-SU.
\end{proof}

\subsection{Multiple Situations}

Now suppose there are multiple situations $G\in\mathcal{G}$ and a
distribution $q\in\Delta(\mathcal{G})$, with $\mathcal{G}$ finite.
At the start of period $t=1,$ Nature draws a situation $G^{(1)}$
from $\mathcal{G}$ according to $q$, and consequences are generated
according to $F^{\bullet}(\cdot,\cdot,G^{(1)})$ until period $t=T+1.$
In period $T+1,$ Nature again draws a situation $G^{(2)}$ from $\mathcal{G}$
according to $q$, and consequences are generated according to $F^{\bullet}(\cdot,\cdot,G^{(2)})$
until period $t=2T+1,$ and so forth. Agents start with a prior over
their group's extended model, $\mu_{g}^{(0)}\in\Delta(\overline{\Theta}_{g})$.
In periods $T+1,2T+1,...$ agents reset their belief to $\mu_{g}^{(0)},$
and their belief in each period over the extended parameters in their
extended model only use histories since the last reset. This belief
corresponds to agents thinking that the data-generating process is
redrawn according to $\mu_{g}^{(0)}$ every $T$ periods.

Suppose $\tau=0$ and suppose for every $G\in\mathcal{G},$ the hypotheses
of Proposition \ref{prop:learning} hold in a society where $G$ is
the only true situation. Denote $(a_{AA}^{*}(G),a_{AB}^{*}(G),a_{BA}^{*}(G),a_{BB}^{*}(G),\mu_{A}^{*}(G),\mu_{B}^{*}(G))$
as the limit of the agents' behavior and beliefs with situation $G.$
Then it is straightforward to see that in a society with the situation
redrawn every $T$ periods, the expected undiscounted average payoff
of an agent in group $g$ approaches the fitness of $g$ in the EZ-SU
characterized by the behavior and beliefs $(a_{AA}^{*}(G),a_{AB}^{*}(G),a_{BA}^{*}(G),a_{BB}^{*}(G),\mu_{A}^{*}(G),\mu_{B}^{*}(G))_{G\in\mathcal{G}}$
with the distribution $q$ over situations, as $T\to\infty$. This
provides a foundation for fitness in EZ-SU as the agents' objective
payoffs when the true situation changes sufficiently slowly (a similar
foundation applies for the fitness in EZ.)

\section{The Single-Agent Case \label{OA:SingleAgent}}

This section records an observation related to our stability concepts
when applied to the single-agent case.  Specifically,  situation $G$
is a \emph{decision problem} if $(a_{i},a_{-i})\mapsto F^{\bullet}(a_{i},a_{-i},G)$
only depends on $a_{i}.$ If every situation is a  decision problem,
then the correctly specified model is evolutionarily stable against
\emph{any} other model, except when there are identification issues.
We adapt the notion of strong identification from \citet{esponda2016berk}.
\begin{defn}
Model $\Theta_{A}$ is \emph{strongly identified }in EZ $\mathfrak{Z}=(\mu_{A}(G),\mu_{B}(G),p,\lambda,a(G))_{G\in\mathcal{G}}$
if in every situation $G$, whenever $F',F''\in\Theta_{A}$ both solve
\begin{align*}
\min_{F\in\Theta_{A}}\left\{ (\lambda+(1-\lambda)p_{A})\cdot K(F;a_{AA},a_{AA},G)+(1-\lambda)(1-p_{A})\cdot K(F;a_{AB},a_{BA},G)\right\} ,
\end{align*}
we have $F^{'}(a_{i},a_{AA})=F^{''}(a_{i},a_{AA})$ and $F^{'}(a_{i},a_{BA})=F^{''}(a_{i},a_{BA})$
for all $a_{i}\in\mathbb{A}$.
\end{defn}
\begin{prop}
\label{prop:decision_problem}Suppose every situation is a decision
problem. Let $\lambda$ and two models $\Theta_{A},\Theta_{B}$
be given, where $\Theta_{A}$ is correctly specified. Suppose there
exists at least one  EZ with $p_{A}=1$, and $\Theta_{A}$ is strongly
identified in all such equilibria. Then $\Theta_{A}$ evolutionarily
stable under $\lambda$-matching against $\Theta_{B}$.
\end{prop}
\begin{proof}
In any  EZ, let $F\in\text{supp}(\mu_{A}(G))$ and note that $F^{\bullet}(\cdot,\cdot,G)\in\Theta_{A}$
since $\Theta_{A}$ is correctly specified. Both $F$ and $F^{\bullet}(\cdot,\cdot,G)$
solve the weighted minimization problem, the former because it is
in the support of $\mu_{A}$, the latter because it attains the lowest
minimization objective of 0. By strong identification, the set of
best responses to $a_{AA}(G)$ and $a_{BA}(G)$ under the belief $\mu_{A}$
is the same as set of actions that maximize payoffs in the decision
problem given by $F^{\bullet}(\cdot,\cdot,G)$. Therefore, adherents
of $\Theta_{A}$ obtain the highest possible objective payoffs in
the stage game in situation $G$. This applies to every situation,
so $\Theta_{A}$ has weakly higher fitness than $\Theta_{B}$ in the
 EZ.
\end{proof}
The result that a resident correct specification is immune to invasions
from misspecifications echoes related results in \citet{FL_mutation}
and \citet*{FII_welfare_based}. We primarily focus on stage games
where multiple agents' actions jointly determine their payoffs and
characterize which misspecifications can invade a rational society
in which environments. 

\noindent

\end{document}